%% file: planarext_arXiv.tex
\documentclass[11pt,letterpaper]{article}

\input{planarext_macros}

\newcommand*\samethanks[1][\value{footnote}]{\footnotemark[#1]}

\title{$\Hcal$-Planarity and Parametric Extensions: \\ when Modulators Act Globally
\thanks{The research leading to these results has been supported the Franco-Norwegian project PHC AURORA 2024-25, the Research Council of Norway via the  BWCA (grant no. 314528) and Extreme-Algorithms (grant no. 355137) projects, and the French-German Collaboration ANR/DFG Project UTMA (ANR-20-CE92-0027).}
}
\date{}

\author{
Fedor V. Fomin%
\thanks{Department of Informatics, University of Bergen, Norway.} 
\and
Petr A. Golovach%
\samethanks 
\and
Laure Morelle\thanks{LIRMM, Univ Montpellier, CNRS, Montpellier, France.}
\and
Dimitrios  M. Thilikos\samethanks}

\begin{document}
\graphicspath{{./Figures/}}

\maketitle

\begin{abstract}
\noindent  
We  introduce a series of graph decompositions based on the 
\textsl{modulator/target} scheme of modification problems that enable several algorithmic applications that parametrically extend the algorithmic potential of planarity. In the core of our approach is a polynomial time algorithm for computing planar $\Hcal$-modulators.
Given a graph class $\Hcal$, a  \emph{planar $\Hcal$-modulator} of a graph $G$ is a set $X \subseteq V(G)$ such that the ``torso'' of $X$ is \textsl{planar} and all connected components of $G - X$ belong to $\Hcal$. Here, the torso of  $X$ 
is obtained from $G[X]$ if, for every connected component of $G-X$, 
we form a clique out of  its neighborhood on $G[X]$. We introduce \probHplan as the problem of deciding whether a graph $G$ has a planar $\Hcal$-modulator. 
We prove that, if  $\Hcal$ is hereditary, CMSO-definable, and decidable in polynomial time, then \probHplan is solvable in polynomial time.

The polynomial-time algorithm has several notable algorithmic consequences. 
First, we introduce two parametric extensions of \probHplan{} by defining the notions of 
\emph{$\Hcal$-planar treedepth} and \emph{$\Hcal$-planar treewidth}, which generalize the concepts of \textsl{elimination distance} and \textsl{tree decompositions} to the class~$\Hcal$. 
For example, a graph is said to have $\Hcal$-planar treewidth at most $t$ if it admits a tree decomposition such that
  every leaf-bag in the decomposition is contained in~$\Hcal$, and
  all other bags are either of size at most $t+1$ or have a planar torso.
 By leveraging our polynomial-time algorithm for \probHplan{}, we prove the following theorem.
Suppose that $\Hcal$ is a hereditary, CMSO-definable, and union-closed graph class, 
and consider the problem of determining whether a graph~$G$ has an \emph{$\Hcal$-modulator} of size at most $k$ for some $k\in\Nbbb$, i.e., a set $X \subseteq V(G)$ of size at most $k$ such that every connected component of $G - X$ lies in~$\Hcal$.  
If this problem is fixed-parameter tractable (FPT) when parameterized by 
$k$, then deciding whether a graph $G$ has $\Hcal$-planar treedepth or $\Hcal$-planar treewidth at most $k$ is also FPT when parameterized by $k$.
By combining this result with existing FPT algorithms for various $\Hcal$-modulator problems, we thereby obtain FPT algorithms parameterized by $\Hcal$-planar treedepth and $\Hcal$-planar treewidth 
for numerous graph classes~$\Hcal$. 

Second, by combining the well-known algorithmic properties of planar graphs and graphs of bounded treewidth, our methods for computing $\Hcal$-planar treedepth and $\Hcal$-planar treewidth lead to a variety of algorithmic applications. For instance, once we know that a given graph has bounded $\Hcal$-planar treedepth or bounded $\Hcal$-planar treewidth, we can derive additive approximation algorithms for graph coloring and polynomial-time algorithms for counting (weighted) perfect matchings. Furthermore,  we design Efficient Polynomial-Time Approximation Schemes (EPTAS-es) for several problems, including \textsc{Maximum Independent Set}.

All previous (meta) algorithmic studies on $\Hcal$-modulators and $\Hcal$-treewidth
 deal with modulators of ``low bidimensionality'' in the sense that they 
cannot be spread along a grid-minor of the input graph 
   [\textsl{Jansen,  de Kroon,   Włodarczyk, STOC 2021; Agrawal, Kanesh, Lokshtanov, Panolan, Ramanujan, Saurabh, and Zehavi, SODA 2022; Fomin, Golovach, Sau, Stamoulis, Thilikos, ACM Trans. Comp. Log. 2025}].  
Our results are the first to handle the situation where the $\Hcal$-modulators 
act globally. For this, we introduce a new version of the ``irrelevant vertex technique''  that  can deal with modulators that may act ``everywhere''  in the input graph. Our technique goes further than the current meta-algorithmic framework on the applicability of the  irrelevant vertex  technique [\textsl{Sau,  Stamoulis,  Thilikos, SODA 2025}] that is conditioned by  bidimensionality.  
\end{abstract}

\newpage
\tableofcontents

\newpage

\section{Introduction}\label{sec_intro}

Graph decompositions play a crucial role in graph theory and have a plethora of algorithmic applications. In our paper, we consider decompositions that historically stem out of graph modification problems. A typical problem of this type asks whether it is possible to apply a series of modifications to a graph in order to transform it into one that exhibits a desired target property. Many classical \NP-complete problems can be described as the following graph modification problem to some class of 
 graphs $\Hcal$. In the \textsc{$\Hcal$-Deletion} problem for an input graph $G$, we aim to find a vertex set $S$ of the minimum size such that graph $G-S$ is in the target class $\Hcal$. 
 For example, if $\Hcal$ is the class of graphs without edges, then \textsc{$\Hcal$-Deletion} is \textsc{Vertex Cover}; and when $\Hcal$ is the class of bipartite graphs, then \textsc{$\Hcal$-Deletion} is \textsc{Odd Cycle Transversal}.
Most graph modification problems are known to be \NP-complete \cite{LewisY80theno,Yannakakis81edged}. However, there is a tremendous amount of study in parameterized algorithms establishing the fixed-parameter tractability of \textsc{$\Hcal$-Deletion} for 
various graph classes $\Hcal$ when parameterized by the size of the set $S$, often referred to in the literature as the \emph{$\Hcal$-modulator} 
(we refer to the book \cite{cygan2015parameterized} for many examples of such algorithms). 
The  \textsc{$\Hcal$-Deletion} problem could be seen as a graph decomposition problem 
where the task is to partition the set of vertices of a graph into a (small) modulator and 
the set of vertices inducing a graph from $\Hcal$.  The main idea is to use the 
constructed decomposition for algorithmic purposes.  In particular, in the 
parameterized complexity area, it is common to investigate problems under structural 
parameterizations by the vertex cover number or the minimum size of a feedback vertex 
set;  we refer to~\cite{Jansen2013ThePO} for basic examples.  

Our results are motivated by the recent developments in the area. 
The first step in these developments is due to Bulian and Dawar \cite{BulianD16graph,BulianD17fixe}, who noticed that many tractability results could 
be extended when, instead of the size of the modulator,  we parametrize by the 
treedepth of the torso of the modulator and, that way, introduced the concept of  \textsl{elimination distance} to $\Hcal$ (also called \textsl{$\Hcal$-treedepth}). These ideas were 
generalized to the scenario where the parameter is the treewidth of the torso of the 
modulator $S$ by Eiben, Ganian,  Hamm,  and  Kwon \cite{EibenGHK21}. Equivalently, this 
leads to the notion of \textsl{$\Hcal$-treewidth}, which is informally a tree decomposition with 
leaf bags from $\Hcal$ and all other bags of bounded 
size~\cite{EibenGHK21,JansenK021verte}. The attractiveness of the $\Hcal$-treewidth is 
that it combines the best algorithmic properties of both worlds: the nice properties of 
$\Hcal$ (e.g., being bipartite or chordal) with the nice properties of graphs of 
bounded treewidth \cite{AgrawalKLPRSZ22,EibenGHK21,JansenK021verte}. Research on 
$\mathcal{H}$-treewidth took an unexpected turn with the work of Agrawal, Kanesh,  
Lokshtanov,  Panolan,  Ramanujan,  Saurabh, and   Zehavi \cite{AgrawalKLPRSZ22}, who 
demonstrated that, from the perspective of (non-uniform) fixed-parameter tractability 
(\FPT), the size of the modulator, the treedepth, and the treewidth of the 
modulator's torso offers equally powerful parameterizations for every hereditary graph 
class $\mathcal{H}$, provided some mild additional conditions are met. We stress that in all aforementioned results, modulators' torsos have always bounded treewidth.

Our motivation for introducing \emph{$\Hcal$-planar treedepth} and \emph{$\Hcal$-planar treewidth} stems from the well-known algorithmic properties of planar graphs, which have been exploited extensively in approximation, parameterized, and counting algorithms. In much the same way that $\Hcal$-treewidth combines the best properties of treewidth and the class~$\Hcal$, we pose the following question:
\begin{quote}
    \textsl{How and when is it possible to combine the best of planarity, graph class $\Hcal$, and treewidth or treedepth?}
\end{quote}

To arrive at a meaningful answer, it is essential to have efficient procedures for computing an
 {$\Hcal$-planar treewidth} and an  {$\Hcal$-planar treedepth}.

\subsection{Looking for planar modulators} 
In the core of our approach to {$\Hcal$-planar treewidth} and   {$\Hcal$-planar treedepth} is the concept of 
modulators whose torsos are planar. To the best of our knowledge, this is the first example of efficiently computable modulators of unbounded treewidth (or rankwidth). As the treewidth and rankwidth of a planar graph can be arbitrarily large, answering this question 
goes beyond the modulator framework of {\cite{EibenGHK21}}.

 For facilitating our presentation, we first define the decision problem \textsc{$\Hcal$-planarity}, and later, we present its parametric extensions so as to display the full generality of our results.

 \begin{figure}[ht!]
\centering
\includegraphics[scale=1]{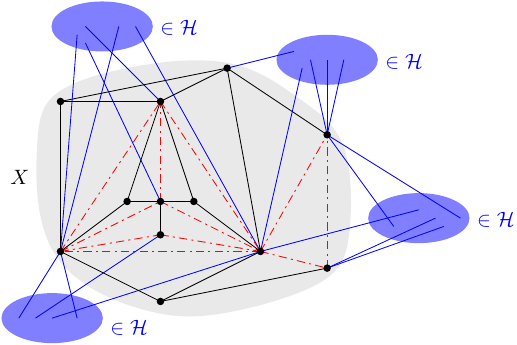}
\caption{Example of an $\Hcal$-planar graph $G$. The vertices of the $\Hcal$-modulator $X$ and the edges of $G[X]$ are black, the edges of the torso of $X$ that are not edges of $G$ are shown in dashed red, and the connected components of $G-X$ together with the edges joining them with $X$ are colored blue.}
\label{fig_example}
\end{figure}

 \paragraph{\mbox{\textsc{$\Hcal$-planarity}}.}
  It is necessary to define the torso concept to   
introduce   \textsc{$\Hcal$-planarity}.  
  (Οther graph-theoretical definitions will be provided in the subsequent section.)
For a vertex set $X \subseteq V(G)$, the \emph{torso} of $X$ in $G$, denoted by $\torso(G,X)$ is the graph derived from the induced subgraph $G[X]$ by turning $N_G(V(C))$, the set of vertices in $X$ adjacent to each connected component $C$ of $G - X$, into a clique.
Let $\Hcal$ be a graph class. A \emph{planar $\Hcal$-modulator} of a graph $G$ is a set $X \subseteq V(G)$ such that the torso of $X$ is planar and every connected component of $G - X$ belongs to $\Hcal$. We introduce \probHplan as the problem of deciding whether a graph $G$ admits a planar $\Hcal$-modulator. 
If graph $G$ admits a planar $\Hcal$-modulator, we will refer to it as a \emph{$\Hcal$-planar} graph (see~\autoref{fig_example}).
Being $\Hcal$-planar can also be viewed as an extension of planarity as it combines it with the property $\Hcal$: an $\Hcal$-planar graph is a graph that becomes planar if we remove disjoint subgraphs belonging to $\Hcal$ 
and connect their neighbors in the rest of a graph as cliques. 
The torso operation represents the ``reminiscence''  of the removed graphs 
in the body of the resulting planar graph.

We say that a class of graphs $\Hcal$ is \emph{CMSO-definable} if for some 
Counting Monadic Second Order (CMSO) sentence $\psi$ for any graph $G$, $G\in \Hcal$ if and only if $G\models \psi$.
Additionally, we say that $\Hcal$ is \emph{polynomial-time decidable} if we can check whether a graph belongs to $\Hcal$ in polynomial-time.
Our first result is the following theorem.

\begin{theorem}\label{th_th}
Let $\Hcal$ be a hereditary, CMSO-definable, and polynomial-time decidable graph class. 
Then there exists a polynomial-time algorithm solving \Hpl.
\end{theorem} 
 
For example, by taking $\Hcal$ to be the class of bipartite graphs, which is a hereditary, 
CMSO-definable,  and polynomial-time decidable graph class, the $\Hcal$-planarity problem becomes the problem of deciding whether the vertex set of a graph $G$ could be partitioned into two sets $X$ and $Y$, such that the torso of $X$ is a planar graph and the vertices of $Y$ induce a bipartite graph.

\paragraph{Irrelevant vertices for global modulators.}
For the proof of \cref{th_th}, we first use the meta-theorem of Lokshtanov, Ramanujan,  Saurabh, and  Zehavi from \cite{LokshtanovR0Z18redu}, in order to prove that 
the problem of solving \Hpl{} can be reduced to the problem of solving a restricted version of the problem, called \Hkpl, which may be of independent interest. The task of this problem is to decide, given  a graph $G$ and a positive integer $k$, whether $G$ has an 
 $\Hcal^{(k)}$-planar modulator where  $\Hcal^{(k)}$ is the subclass of $\Hcal$ composed by the graphs of $\Hcal$ with at most $k$ vertices. 
 Next we prove that yes-instances of \Hkpl exclude some bounded size clique 
 as a minor. This permits us to apply the flat wall theorem from 
  \cite{RobertsonS95XIII,KawarabayashiTW18anew,SauST24amor},
 in order to find either that the graph has bounded treewidth, in which case we may conclude using Courcelle's Theorem~\cite{Courcelle90them,ArnborgLS91easy}, or to find a ``big enough'' flat wall in it. In this last case, our algorithm produces 
 an equivalence instance of the problem by removing an ``irrelevant vertex''
 in the middle of this wall. This is the general framework of 
 the ``irrelevant vertex technique'' whose general applicability can be guaranteed 
 by a series of  algorithmic meta-theorems that have been recently developed in \cite{FominGSST21compound,GolovachST23model} and, in its more powerful version, in \cite{SauST2024parame}.
 Unfortunately, none of these meta-algorithmic results can be applied in our case.
 Typically, the applicability of the algorithmic meta-theorems in \cite{SauST2024parame,FominGSST21compound,GolovachST23model}
 require that the modifications are affecting only a \textsl{restricted area} of a big flat wall. 
 In technical terms, according to \cite{SauST2024parame},  they should affect a set of vertices with ``low bidimensionality'', i.e., 
 a set  through which it is not possible to root a big grid minor. In our case the modifications are \textsl{global} and they may occur everywhere along the two-dimensional  area of the flat wall. Therefore, 
precisely  because our modulator is a planar graph (of potentially unbounded  treewidth),
it is not possible to apply any of the results in \cite{SauST2024parame,FominGSST21compound,GolovachST23model} in order to find 
an irrelevant vertex. This initiates 
the challenge of applying the irrelevant vertex technique under the setting that the modification potentially affects  any part of the flat wall. The most technical 
part of the paper is devoted to this problem. Our approach applies 
the irrelevant vertex technique beyond the meta-algorithmic  framework of \cite{SauST2024parame}  and, we believe, 
may have independent use for other problems involving global modification.

\paragraph{On the conditions on $\Hcal$.}
We proceed with a few remarks on the conditions in \cref{th_th}. 
First, it is important that we demand the torso of the modulator to be planar rather than to allow the whole modulator to be planar.  Without this condition, the problem becomes   \NP-hard. Indeed,  Farrugia \cite{Farrugia04vert} proved that for two graph classes $\Pcal$ and $\Qcal$ that are hereditary and closed under disjoint-union, the problem of deciding whether a graph $G$ admits a partition $(A,B)$ of $V(G)$ such that $G[A]\in\Pcal$ and $G[B]\in\Qcal$ is \NP-hard, unless $\Pcal$ and $\Qcal$ are both the class of edgeless graphs.
Hence, for any hereditary graph class $\Hcal$ closed under the disjoint union operation, the problem of deciding, given a graph $G$, whether there exists $S\subseteq V(G)$ such that $G[S]$ is planar and that each connected component $C$ of $G-S$ belongs to $\Hcal$, is \NP-hard.

Second, the necessity of $\Hcal$ to be a polynomial-time decidable graph class is apparent---it is easy to show that \NP-hardness of deciding whether a graph is in  $\Hcal$ implies \NP-hardness of \Hpl\ as well. However, the necessity of $\Hcal$ to be a hereditary property is less obvious. We prove in \cref{thm_lower} that for non-hereditary classes $\Hcal$, \Hpl\ becomes \NP-hard.

 Finally, we do not know whether  CMSO-definability of $\Hcal$ is necessary.
  We require this condition because the initial step of our proof uses the meta-theorem of Lokshtanov, Ramanujan,  Saurabh, and  Zehavi from \cite{LokshtanovR0Z18redu} on reducing CMSO-definable graph problems to problems on unbreakable graphs. Whether CMSO-definability could be replaced by a more general property remains an intriguing open question.
 
Using the meta-theorem of ~\cite{LokshtanovR0Z18redu} comes with the following drawback:     \cref{th_th} is non-constructive. Similar to the results in \cite{LokshtanovR0Z18redu}, it allows us to infer the existence of a polynomial-time algorithm for every CMSO formula $\varphi$, but it does not provide the actual algorithm. The reason for that is that   \cite{LokshtanovR0Z18redu} relies on the existence of representative subgraphs for equivalence classes of bounded-boundary graphs but does not provide a procedure to compute such representatives. 
 In other words, \cref{th_th} provides a \emph{non-constructive} polynomial-time algorithm. 

 \subsection{Parametric extensions of $\Hcal$-planarity}

As we already mentioned, asking whether a graph is $\Hcal$-planar is a decision problem combining $\Hcal$ with planarity. In order to position our work to the 
more general area of graph modification parameters, 
we extend our results by combining $\Hcal$ with parametric extensions of planarity. To explain this, we present below 
a general framework for defining graph modification parameters based on the ``modulator/target'' scheme, inspired from \cite{EibenGHK21}. 

In this paper, every graph parameter $\p$, that is a function mapping a graph to an integer, is assumed to be \emph{minor-monotone}, that is, if $G_{1}$ 
is a minor of $G_{2}$, then $\p(G_{1})≤\p(G_{2})$. The parameter $\p$ will serve as  the measure of modulation.

\paragraph{Extending planar modulation.}
Let now $\Hcal$ be some graph class expressing some ``target property''. 
We further demand that $\Hcal$ is \emph{union-closed}, i.e., the disjoint union of two graphs in $\Hcal$ 
also belongs in $\Hcal$. 
Given a graph $G$, we compose the modulation parameter $\p$ and the target property $\Hcal$ to the graph parameter ${\Hcal}\mbox{-}\p$
such that 
\begin{eqnarray}
\begin{minipage}{14cm}
\textsl{${\Hcal}\mbox{-}\p(G)=\min\{k\mid \exists X\subseteq V(G):  \text{$\p(\torso(G,X))≤k$ and  $G-X\in \Hcal$\}}.$}
\end{minipage}
\label{constrxaint}
\end{eqnarray}

In case $\p$ is the size of a graph, denoted by $\size$, then 
${\Hcal}\mbox{-}\size$ is the classic  \textsc{$\Hcal$-Deletion} problem.
If $\p$
is the treedepth $\td$ of a graph, then 
${\Hcal}\mbox{-}\td$ is the elimination distance to $\Hcal$ defined 
by Bulian and Dawar in \cite{BulianD16graph,BulianD17fixe}.
If $\p$
is the treewidth, $\tw$ of a graph, then 
${\Hcal}\mbox{-}\tw$ is the $\Hcal$-treewidth
parameter introduced by Eiben, Ganian,  Hamm,  and  Kwon in \cite{EibenGHK21}. 
Consider the following general question that expresses a vast variety 
of modification problems. 
\begin{eqnarray}
\begin{minipage}{14cm}
\textsl{Given some union-closed class $\Hcal$ and   some minor-monotone 
parameter $\p$, is there an \textsf{FPT}-algorithm   deciding,   given 
a graph $G$ and a  $k\in\Nbbb$,   whether ${\Hcal}\mbox{-}\p(G)≤k$?}
\end{minipage}
\label{constraint}
\end{eqnarray}
An important contribution  towards finding a general answer to \eqref{constraint} was given by the  results of  \cite{AgrawalKLPRSZ22}.
\begin{proposition}[\!\!{\cite{AgrawalKLPRSZ22}}]\label{prop1}
Let $\p\in \{\td, \tw\}$ and let $\Hcal$ be a graph class that is hereditary, CMSO-definable, and union-closed. If  there is an \FPT-algorithm that solves \pbdel parameterized by the solution size $k$, then  there is a (non-uniform) \FPT-algorithm deciding   whether ${\Hcal}\mbox{-}\p(G)≤k$. 
\end{proposition}
  
\autoref{prop1} immediately transfers all algorithmic knowledge on the  \textsc{$\Hcal$-Deletion} problem (that is checking whether ${\Hcal}\mbox{-}\size(G)≤k$) to parametric extensions ``above treewidth''.
This last condition is fundamental to all techniques 
 employed so far for dealing with the problem in \eqref{constraint}. 
The emerging question is to derive FPT-algorithms 
for choices of  $\p$ ``below treewidth'' (more precisely, ``not above treewidth'') which means that the  torso of the modulator may contain as a minor an \textsl{arbitrary big planar graph}.
These parameters are bounded on planar 
graphs and likewise they can be seen as parametric extensions of planarity. 
Designing algorithms for \eqref{constraint} on such parameters 
  is  a completely uninvestigated  territory  in modification problems.
In this paper we make a decisive step to this direction by introducing two  parametric extensions of planarity.

\paragraph{Planar treedepth.}
The first parameter  is the \textsl{planar treedepth}, denoted by  $\ptd$,
where $\ptd(G)=0$ iff $G$ is empty, and 
$\ptd(G)≤k$ if there is some $X\subseteq V(G)$ where $\torso(G,X)$
is planar and $\ptd(G-X)≤k-1$. 
Notice that if $G$ is planar, then $\ptd(G)≤1$. 
Intuitively, $\ptd(G)$ expresses the minimum number of ``planar layers''
that should be removed in order  to obtain the empty graph.
We use the term \emph{\hptd{$\Hcal$}} for $\Hcal\mbox{-}\ptd$.

\paragraph{Planar treewidth.}
The second parameter is  \emph{planar treewidth}, denoted by $\ptw$, 
and is based on tree decompositions (see \cref{sec_prelim} for the formal definition). We say that 
$\ptw(G)≤k$ if there is a tree decomposition $(T,\beta)$
where, for each node $t\in V(T)$,  if $\torso(G,\beta(t))$ is non-planar, then it  has size at most $k+1$. Notice that if $G$ is planar, $\ptw(G)=0$.
Planar treewidth appeared for the first time by Robertson and Seymour in \cite{RobertsonS91}
where they proved that, if a graph excludes some singly-crossing\footnote{A graph $H$ is \emph{singly-crossing} if it can be embedded in the sphere so that at most one pair of edges intersect.}  graph $H$, then
its planar treewidth is bounded by some constant depending on $H$.
We use the term \emph{$\Hcal$-planar treewidth} for $\Hcal\mbox{-}\ptw$.

Clearly both $\Hcal$-planar treedepth and $\Hcal$-planar treewidth
are more general parameters than  $\Hcal$-treedepth or $\Hcal$-treewidth. To see this, consider some $\Hcal$-planar graph where the planar torso of the modulator has 
treewidth $≥k$. While both $\Hcal$-planar treedepth and $\Hcal$-planar treewidth are bounded by constants on this graph, the values $\Hcal$-treedepth or $\Hcal$-treewidth 
are conditioned by $k$ and can be arbitrarily big. Moreover if $\Hcal$-treedepth or $\Hcal$-treewidth are small then also their planar counterparts are also small by using the same decompositions where the planar torso is just a vertex per component.

By extending the techniques that we introduce for proving \cref{th_th}, we can  prove the following  parametric extensions.
Throughout the paper, $n$ is the number of vertices of the input graph $G$, and $m$ its number of edges.
 
\begin{theorem}\label{th_param}
Let $\p\in \{\ptd, \ptw\}$ and let $\Hcal$ be a graph class that is hereditary, CMSO-definable, and union-closed. If  there is an \FPT-algorithm that solves \pbdel parameterized by the solution size $k$ in time $\Ocal_k(n^c)$, then  there is a (non-uniform) algorithm deciding  whether $\Hcal\mbox{-}\p(G)\le k$ in time $\Ocal_k(n^4+n^c\log n)$.\footnote{By $g(n)=\Ocal_{k}(f(n))$ we mean that there is a function $h:\Nbbb\to\Nbbb$ such that $g(n)=O(h(k)\cdot f(n))$.}
\end{theorem}

\subsection{Algorithmic applications of the decompositions}
\label{appli_rxio}  
\cref{th_th} and \cref{th_param} have numerous applications for algorithm design.  As with constructing decompositions, the crucial step is to understand the computational complexity of specific (classes) of problems on $\Hcal$-planar graphs. Thus, in this paper, we focus on algorithmic applications of \cref{th_th} and then briefly discuss the recursive applicability of these algorithmic paradigms on graphs of bounded $\Hcal$-planar  treewidth and $\Hcal$-planar treedepth. 

\paragraph{Computing the chromatic number.}
Our first algorithmic application of \Cref{th_th} concerns graph coloring. 
A surjective function $c :  V(G)\to \{1,2, \dots, k\}$ is a \emph{proper $k$-coloring} if, for any pair of  adjacent vertices $u$ and $v$, $c(u)\neq c(v)$. The \emph{chromatic} number of a graph $G$, $\chi(G)$, is the minimum number $k$ such that $G$ admits a proper $k$-coloring. 

By the Four Color Theorem \cite{AppelH89,RobertsonSST96effic}, the chromatic number of a planar graph is at most four. 
 By pipelining the Four Color Theorem with  \cref{th_th}, we obtain the following. 
 
 \begin{restatable}{theorem}{coloring}\label{lemma_coloring}
 Let $\Hcal$ be a hereditary, CMSO-definable, and polynomial-time decidable graph class. Moreover, assume that there is a polynomial-time algorithm computing the chromatic number $\chi(H)$ for $H\in \Hcal$. 
Then, there exists a polynomial-time algorithm that, given an $\Hcal$-planar graph $G$,  produces a proper coloring of $G$ using at most $\chi(G)+4$ colors.
\end{restatable} 

For example, the class of perfect graphs is hereditary, CMSO-definable, and polynomial-time decidable \cite{chudnovsky2005recognizing}. Moreover, the chromatic number of a perfect graph is computable in polynomial time \cite{grotschel1984polynomial}. Thus by \Cref{lemma_coloring}, when 
$\Hcal$ is the class of perfect graphs, for any $\Hcal$-planar graph $G$, there is a polynomial-time algorithm coloring $G$ in at most $\chi(G)+4$ colors.

\paragraph{Counting perfect matchings.}
Our second example of applications of  \Cref{th_th} concerns counting perfect matchings.  While counting perfect matchings on general graphs is  \#P-complete, on planar graphs, it is polynomial-time solvable by the celebrated Fisher–Kasteleyn–Temperley (FKT) algorithm \cite{fisher1961statistical,kasteleyn1961statistics,temperley1961dimer}. 
 By making use of Valiant's  ``matchgates'' \cite{Valiant08holo},  we use \cref{th_th} to extend the FKT algorithm to $\Hcal$-planar graphs.

Let us remind that a \emph{perfect matching} in a graph $G$ is a set $M\subseteq E(G)$ such that every vertex of $G$ occurs in exactly one edge of $M$.
Let $w: E(G)\to\Nbbb$ be a weight function.
The \emph{weighted number of perfect matchings} in $G$, denoted by $\pmm(G)$, is defined as $$\pmm(G)=\sum_{M}\prod_{e\in M} w(e),$$
where the sum is taken over all perfect matchings $M$. 
If $w=1$, then $\pmm(G)$ is the number of perfect matchings in $G$.

\begin{restatable}{theorem}{perfect}\label{thm_perfectmatching}
Let $\Hcal$ be a hereditary, CMSO-definable, and polynomial-time decidable graph class. Moreover, assume that the weighted (resp.~unweighted) number of perfect matchings  $\pmm(H)$ is computable in polynomial time for graphs in $\Hcal$. Then, there exists a polynomial-time algorithm that, given a weighted (resp.~unweighted) $\Hcal$-planar graph $G$, computes its weighted  number of perfect matchings $\pmm(G)$.
\end{restatable}
 
 Examples of classes of graphs $\Hcal$  where counting perfect matchings can be done in polynomial time are graphs 
 excluding a shallow-vortex as a minor  \cite{ThilikosW22} (see also \cite{GalluccioL99onthe,CurticapeanX15param,Curticapean14count,EppsteinV19ncalg,ArnborgLS91easy}), bounded clique-width graphs \cite{CurticapeanM16}, and
chain, co-chain, and threshold graphs \cite{OkamotoUU09}.

\paragraph{EPTAS's.}
 Our third example is the extension of Baker's technique on planar graphs used for approximation and \FPT\-algorithms. While we  provide an example for  {\sc Independent Set}, 
 similar results about efficient polynomial-time approximation schemes (EPTAS) could be obtained for many other graph optimization problems. 
 
Let us remind that an independent set in a graph $G$ is a set of pairwise nonadjacent vertices. We use $\alpha(G)$ to denote the maximum size of an independent set of $G$.

 \begin{restatable}{theorem}{independent}\label{lem_IS}
 Let $\Hcal$ be a hereditary, CMSO-definable, and polynomial-time decidable graph class. We also assume that there is a polynomial-time algorithm computing a maximum independent set of graphs in $\Hcal$. 
Then, there is an algorithm that, given $\varepsilon>0$ and an $\Hcal$-planar graph $G$, computes in time $2^{\Ocal(1/\varepsilon)}\cdot |V(G)|^{\Ocal(1)}$ an independent set of $G$ of size at least $(1-\varepsilon)\cdot \alpha(G)$. 
\end{restatable}

Examples of graph classes where $\alpha(G)$ is computable in polynomial time are perfect graphs \cite{grotschel1984polynomial} or graphs excluding $P_6$ as an induced subgraph \cite{GrzesikKPP19}.  \Cref{lem_IS} could also be modified for graph classes where computing a maximum independent set is quasi-polynomial or subexponential. In these cases, the approximation ratio will remain $1-\varepsilon$, but the approximation algorithm's running time will also become quasi-polynomial or subexponential.

\paragraph{Applications of $\Hcal$-planar treedepth and $\Hcal$-planar treewidth.}
The above applications of  $\Hcal$-planarity mentioned 
can be extended to the parametric setting of graphs where $\Hcal\mbox{-}\ptd$ and $\Hcal\mbox{-}\ptw$ are bounded by a parameter $k$.

Suppose  first that we have an   elimination sequence certifying $\Hcal\mbox{-}\ptd(G)\leq k$. By a repetitive application of \cref{lemma_coloring}, we can derive a polynomial time algorithm producing a proper coloring of $G$ using at most $\chi(G)+O(k)$ colours.  Similarly, using \cref{thm_perfectmatching}, we can derive a polynomial  time algorithm that, given a weighted graph, computes the weighted  number of its perfect matchings.

Suppose  now that we have a  tree decomposition certifying that  $\Hcal\mbox{-}\ptw(G)≤k$. By combining \cref{lemma_coloring} with dynamic programming 
on the tree decomposition, we may derive a polynomial time algorithm producing a proper coloring of $G$ using at most $\chi(G)+O(k)$ colours. Combining \cref{thm_perfectmatching} with the dynamic programming approach of \cite{ThilikosW22}, it is easy to derive  an  algorithm that, given a weighted graph, computes the weighted  number of its perfect matchings in time $n^{O(k)}$ (see \cref{coudif_jkiolmd}).

Extending \cref{lem_IS} in the parametric setting is not as straightforward, but we still sketch how to extend Baker's technique for both bounded $\Hcal\mbox{-}\ptd$ and bounded $\Hcal\mbox{-}\ptw$  (see \cref{baker_more_r}).

\medskip
We stress that \Cref{lemma_coloring,thm_perfectmatching,lem_IS} and their corollaries for graphs with bounded $\Hcal\mbox{-}\ptd$ and $\Hcal\mbox{-}\ptw$ are only indicative snapshots of the algorithmic applicability of $\Hcal$-planarity, $\Hcal$-planar treewidth, and $\Hcal$-planar treedepth. Numerous planar applications exist in various algorithmic subfields, ranging from distributed algorithms to kernelization and subexponential algorithms. Of course, not all such methods and results for planar graphs could be transferred even to 
$\Hcal$-planar graphs.
Exploring the full set of algorithmic applications of $\Hcal$-planarity, $\Hcal$-planar treewidth, and $\Hcal$-planar treedepth escapes the purposes of this paper. 
However, we expect that our results will appear to be useful in extending  
algorithmic paradigms where dynamic programming on graphs of bounded treewidth can be extended to dynamic programming on graphs of bounded $\Hcal$-treewidth.

\subsection{Related work} 
Different types of bounded-treewidth modulators were introduced in 
\cite{BulianD16graph,EibenGHK21,BulianD17fixe} and have been studied 
extensively in \cite{JansenK021verte,AgrawalKLPRSZ22} for various 
instantiations of the target property~$\Hcal$. 
Eiben, Ganian, and Szeider~\cite{eiben2018solving,eiben2018meta} 
considered ``well-structured modulators'' of bounded rank-width, 
together with certain (rather restrictive) conditions on how the 
modulator interacts with the connected components of the graph. 
Furthermore, meta-algorithmic conditions enabling the automatic derivation 
of \FPT\ algorithms for problems defined via the modulator/target scheme 
were introduced by Fomin \emph{et al.}~\cite{FominGSST21compound}, 
and the meta-algorithmic potential of the ``irrelevant vertex'' technique 
was recently explored in \cite{SauST2024parame,GolovachST23model}.

Our decomposition-based parameters, $\Hcal\mbox{-}\ptd$ and 
$\Hcal\mbox{-}\ptw$, extend beyond the state of the art found in 
\cite{BulianD16graph,EibenGHK21,BulianD17fixe,JansenK021verte,AgrawalKLPRSZ22}. 
Moreover, the applicability of our techniques addresses problems that 
cannot be expressed by any of the meta-algorithmic frameworks 
in \cite{FominGSST21compound} or \cite{SauST2024parame,GolovachST23model}.

Concerning the techniques employed, 
an irrelevant vertex technique somewhat similar to the one presented in this paper is introduced in~\cite{GroheKR13asimp} in order to obtain an algorithmic version of the Graph Minor structure theorem.

 \subsection{Organization of the paper}
 
 The remaining part of our paper is organized as follows. In~\autoref{sec_prelim}, we introduce basic notions used throughout the paper. In~\autoref{sec_alg}, we outline 
 our algorithmic techniques  for \Hpl, and their extensions for the parameters $\Hcal$-\ptd, and $\Hcal$-\ptw. Also we explain how  to create equivalent instances when replacing  $\Hcal$ by  $\Hcal^{(k)}$  using the meta-theorem of \cite{LokshtanovR0Z18redu}.
 To prove~\autoref{th_th}, we show that \Hkpl{} is \FPT when parameterized by $k$. The proof of this result is given in \Cref{lem_small_leaves}, the most technical part of our paper. For clarity's sake, we first sketch the proof of~\autoref{lem_small_leaves} in~\autoref{sec_alg}, and we give the formal proof in~\autoref{sec_renditions}.  
In \autoref{cosi46y} and \autoref{sec_tw}, we prove \autoref{th_param} for $\p=\ptd$ and $\p=\ptw$ respectively.
In~\autoref{sec_applications},  we demonstrate algorithmic applications of  \autoref{th_th} and \cref{th_param} for several classical problems. Further, in~\autoref{sec_lower},   
we discuss the necessity of conditions on $\Hcal$ in \cref{th_th} and \cref{th_param}.
We conclude our paper in~\autoref{sec_concl} by discussing further research directions.   
 
\section{Preliminaries}\label{sec_prelim}
\paragraph{Sets and integers.}
We denote by $\Nbbb$ the set of non-negative integers.
Given two integers $p$ and $q$, $[p,q]$ is the set of all integers $r$ such that $p\leq r\leq q$.
For an integer $p\geq 1$, we set $[p]=[1,p]$ and $\Nbbb_{\geq p}=\Nbbb\setminus [0,p-1]$.
For a set $S$, we denote by $2^{S}$ the set of all subsets of $S$ and, given an integer $r\in[0,|S|]$,
we denote by $\binom{S}{r}$ the set of all subsets of $S$ of size $r$.
The function $\odd:\mathbb{R}\to\Nbbb$ maps $x$ to the smallest odd non-negative integer larger than $x$.

\paragraph{Basic concepts on graphs.}
All graphs considered in this paper are undirected, finite, and without loops or multiple edges.
We use standard graph-theoretic notation and we refer the reader to \cite{Diestel10grap} for any undefined terminology.
For convenience, we use $uv$ instead of $\{u,v\}$ to denote an edge of a graph.
Let $G$ be a graph. In the rest of this paper we always use $n$ for the size of $G$, i.e., the cardinality of $V(G)$.
We say that a pair $(L,R)\in 2^{V(G)}\times 2^{V(G)}$ is a {\em separation} of $G$
if $L\cup R=V(G)$ and there is no edge in $G$ between $L\setminus R$ and $R\setminus L$.
The \emph{order} of $(L,R)$ is $|L\cap R|$.
Given a vertex $v\in V(G)$, we denote by $N_{G}(v)$ the \emph{open neighborhood} of $v$, that is, the set of vertices of $G$ that are adjacent to $v$ in $G$; $N_G(X)=\big(\bigcup_{v\in X}N_G(v)\big)\setminus X$ for a set of vertices $X$. Slightly abusing notation, we may write $N_G(H)$ instead of $N_G(V(H))$ for a subgraph $H$ of $G$.
We use $N_G[X]$ to denote $N_G(X)\cup X$.
For $S \subseteq V(G)$, we set $G[S]=(S,E(G)\cap \binom{S}{2} )$
and use the shortcut $G - S$ to denote $G[V(G) \setminus S]$.
We may also use $G- v$ instead of $G-\{v\}$ for $v\in V(G)$.
We use $\cc(G)$ to denote the set of connected components of $G$.
Given $X\subseteq V(G)$, we define the \emph{torso} of $X$ in $G$, denoted by $\torso(G,X)$, to be the graph obtained from $G[X]$ by making a clique out of $N_G(V(C))$ for each $C\in\cc(G-X)$.
The class of planar graphs is denoted by $\Pcal$.

\paragraph{Treewidth.}
A \emph{tree decomposition} of a graph $G$ is a tuple $\mathcal{T}=(T,\beta)$ where $T$ is a tree and $\beta\colon V(T)\rightarrow 2^{V(G)}$ is a function, whose images are called the \emph{bags} of $\mathcal{T},$ such that 
  \begin{itemize}
  \item $\bigcup_{t\in V(T)} \beta (t)= V (G),$ 
  \item for every $e\in E(G)$, there exists $t\in V(T)$ with $e\subseteq \beta(t),$ and 
  \item for every $v\in V(G)$, the set $\{t\in V(T) \mid v\in \beta(t)\}$ induces a subtree of $T.$
  \end{itemize} 
The \emph{width} of $\mathcal{T}$ is the maximum size of a bag minus one
 and the \emph{treewidth} of $G$, denoted by $\tw(G),$ is the minimum width of a tree decomposition of $G$.
The \emph{planar treewidth} of a graph $G$, denoted by $\ptw(G),$ is the minimum $k$ such that there exists a tree decomposition of $G$ such that each bag either has size at most $k+1$ or has a planar torso.

\paragraph{Treedepth.}
The \emph{treedepth} of a graph $G$, denoted by $\td(G),$  is zero if $G$ is the empty graph, and one plus the minimum treedepth of the graph obtained by removing one vertex from each connected component of $G$ otherwise.
The \emph{planar treedepth} of $G$, denoted by $\ptw(G),$  is defined similarly to the treedepth, but where we remove a modulator whose torso is planar from each connected component of $G$ instead of a vertex.
In other words, the planar treedepth of $G$ is the elimination distance of $G$ to the graph class that contains only the empty graph.
The fact that $\ptd(G)≤k$ is certified 
by the sequence $X_{1},\ldots,X_{k}$ of successive planar modulators that need to be removed. We refer to such a sequence as a $\emph{certifying elimination sequence}$ (see \autoref{fig_planarED}).
\begin{figure}[h]
\center
\includegraphics[scale=0.7]{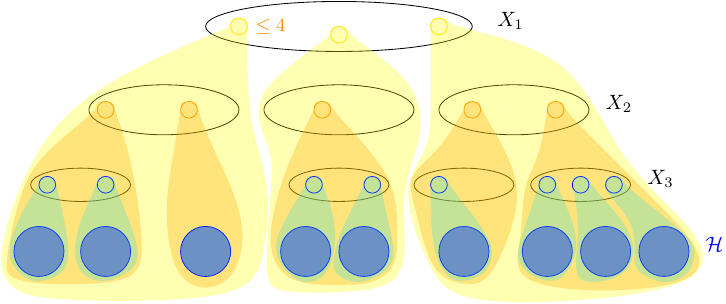}
\caption{Illustration of a graph with \hptd{$\Hcal$} at most three. The sequence $X_1,X_2,X_3$ is a certifying elimination sequence.}
\label{fig_planarED}
\end{figure}

\paragraph{$\Gcal\triangleright\Hcal$-modulators.}
Let $\Hcal$ and $\Gcal$ be graph classes.
We define $\Gcal\triangleright\Hcal$ to be the class of graphs $G$ that contain a vertex subset $X\subseteq V(G)$, called \emph{$\Gcal\triangleright\Hcal$-modulator}, such that $\torso(G,X)\in\Gcal$ and, for each $C\in\cc(G-X)$, $C\in\Hcal$.
Thus, $\Pcal\triangleright\Hcal$ is the class of $\Hcal$-planar graphs.
Note that the operator $\triangleright$ is associative, i.e. $\Fcal\triangleright(\Gcal\triangleright\Hcal)=(\Fcal\triangleright\Gcal)\triangleright\Hcal$.
For $k\in\Nbbb_{\ge1}$, we set $\Gcal^{k+1}=\Gcal\triangleright\Gcal^k$, where $\Gcal^1=\Gcal$.
In particular, $\Pcal^k$ is the class of graphs of planar treedepth at most $k$ and $\Pcal^k\triangleright\Hcal$ is the class of graphs with \hptd{$\Hcal$} at most $k$.

Note that, given $k\in\Nbbb$, if $\Gcal_k$ is class of graphs with treewidth (resp. treedepth / size / planar treewidth / planar treedepth) at most $k$, then $\Gcal_k\triangleright\Hcal$ is the class of graphs with $\Hcal$-treewidth (resp. elimination distance to $\Hcal$ / $\Hcal$-deletion\footnote{If $\Hcal$ is union-closed.} / $\Hcal$-planar treewidth / \hptd{$\Hcal$}) at most $k$.
The class of graphs with planar treewidth at most $k$ will be denoted by $\Pcal\Tcal_k$.

\paragraph{Counting Monadic Second-Order Logic.}
\emph{Monadic Second-Order Logic} (MSO) is a language to express properties in graphs.
The syntax of MSO includes logical connectives $\land$, $\lor$, $\neg$, $\Leftrightarrow$, $\Rightarrow$, variables for vertices, edges, vertex sets, and edge sets, quantifiers $\exists$, $\forall$ over these variables, the atomic expressions 
$u\in U$ when $u$ is a vertex variable and $U$ is a vertex set variable;
$e\in E$ when $e$ is a edge variable and $E$ is a edge set variable;
${\rm adj}(u,v)$ when $u$ and $v$ are vertex variables, with the interpretation that $u$ and $v$ are adjacent;
${\rm inc}(u,e)$ when $u$ is a vertex variable and $e$ is an edge variable, with the interpretation that $e$ is incident to $u$;
and equality of variables representing vertices, edges, vertex sets, and edge sets.
\emph{Counting Monadic Second-Order Logic} (CMSO) extends MSO by including atomic sentences testing whether the cardinality of a set is equal to $q$ modulo $r$, where $r\in\Nbbb_{\ge 2}$ and $q\in[r-1]$.

\smallskip
Note that planarity and connectivity are expressible in CMSO logic, see e.g. \cite[Subsection 1.3.1]{CourcelleJ12grap}. 
Additionally, $\torso(G,X)$ is the graph with vertex set $X$, and edge set the pairs $(u,v)\in X^2$ such that either $uv\in E(G)$ or $u$ and $v$ are connected in $G-(X\setminus\{u,v\})$, which is also easily expressible in CMSO logic.
Therefore, we observe the following.

\begin{observation}\label{obs_CMSO}
If $\Hcal$ is a CMSO-definable graph class, then \Hpl\ is expressible in CMSO logic.
\end{observation}

Courcelle's theorem is a powerful result essentially saying that any CMSO-definable problem 
is solvable in polynomial time on graphs of bounded treewidth. For  a CMSO sentence $\psi$, the task of  $CMSO[\psi]$ is to decide whether  $G\models\psi$ for a graph $G$.

\begin{proposition}[Courcelle's Theorem \cite{Courcelle90them,ArnborgLS91easy}]\label{Courcelle}
Let $\psi$ be a CMSO sentence.
Then there is a function $f:\Nbbb^2\to\Nbbb$ and an algorithm that, given a graph $G$ of treewidth at most $\tw$, solves $CMSO[\psi]$ on $G$ in time $f(|\psi|,\tw)\cdot n$.
\end{proposition}

\paragraph{Minors.} The \emph{contraction} of an edge $e = uv\in E(G)$ results in a graph $G'$ obtained from $G-\{u,v\}$ by adding a new vertex $w$ adjacent to all vertices in the set $N_G(\{u,v\})$. 
A graph $H$ is a \emph{minor} of a graph $G$ if $H$ can be obtained from a subgraph of $G$ after a series of edge contractions. 
It is easy to verify that $H$ is a minor of $G$ if and only if there is a collection $\Scal=\{S_v\mid v\in V(H)\}$ of pairwise-disjoint connected subsets of $V(G)$, called \emph{branch sets}, such that, for each edge $xy\in E(H)$, the set $S_x\cup S_y$ is connected in $V(G)$. 
Such a collection $\Scal$ is called a \emph{model} of $H$ in $G$.
A graph class $\Hcal$ is \emph{minor-closed} if, for each graph $G$ and each minor
$H$ of $G$, the fact that $G\in\Hcal$ implies that $H\in\Hcal$. 
A \emph{(minor-)obstruction} of a graph class $\Hcal$ is a graph $F$ that is not in $\Hcal$, but whose minors are all in $\Hcal$.

\section{The algorithms}\label{sec_alg}

In \autoref{subsec_subsec}, we prove \autoref{th_th} and \autoref{th_param}, assuming some results that will be proved later in the paper but that we sketch in \autoref{sec_outline} (for $\Hcal$-planarity and $\Hcal$-planar treewidth) and \autoref{subsec_changes_ed} (for \hptd{$\Hcal$}).

\subsection{The algorithms}\label{subsec_subsec}

The starting point of our algorithmms  is the result of Lokshtanov, Ramanujan, Saurabh, and Zehavi \cite{LokshtanovR0Z18redu} reducing a CMSO-definable graph problem to the same problem on unbreakable graphs.

\begin{definition}
[Unbreakable graph]
Let $G$ be a graph and let $c,s\in\Nbbb$. If there exists a separation $(X,Y)$ of order at most $c$ such that $\vert X\setminus Y\vert \geq s$ and $\vert Y\setminus X\vert \geq  s$, called an {\em $(s,c)$-witnessing separation}, then $G$ is {\em $(s,c)$-breakable}. Otherwise, $G$ is {\em $(s,c)$-unbreakable}.
\end{definition}

\begin{proposition}[Theorem 1, \cite{LokshtanovR0Z18redu}]\label{prop_CMSOMetaTheorem}
Let $\psi$ be a CMSO sentence and let  $d>4$ be a positive integer. Then there exists a function $\alpha : \mathbb{N} \rightarrow \mathbb{N}$, such that, for every $c\in \mathbb{N}$, if there exists an algorithm that solves {\sc CMSO}$[\psi]$ on $(\alpha(c),c)$-unbreakable graphs in time $\Ocal(n^d)$, then there exists an algorithm that solves {\sc CMSO}$[\psi]$ on general graphs in time $\Ocal(n^{d})$.
\end{proposition}

In our case, we have the following observation.

\begin{observation}\label{obs_clique}
Let $\Hcal$ and $\Gcal$ be graph classes.
If every graph in $\Gcal$ is $K_{k+1}$-minor-free, then, for any $\Gcal\triangleright\Hcal$-modulator $X$ of a graph $G$ and for any $C\in\cc(G-X)$, $|N_G(V(C))|\le k$.
In particular, if $\Gcal=\Pcal$, then $|N_G(V(C))|\le 4$, if $\Gcal=\Pcal^k$, then $|N_G(V(C))|\le 4k$, and if $\Gcal=\Pcal\Tcal_k$, then $|N_G(V(C))|\le \max\{k+1,4\}$.
\end{observation}

Let $a,k\in\Nbbb$, $\Hcal$ and $\Gcal$ be two graph classes such that every graph in $\Gcal$ is $K_{k+1}$-minor-free, and $G$ be an $(a,k)$-unbreakable graph.
We say that a $\Gcal\triangleright\Hcal$-modulator $X$ of a graph $G$ is a \emph{big-leaf $\Gcal\triangleright\Hcal$-modulator} of $G$ if there is exists a (unique) component $D\in\cc(G-X)$ of size at least $a$, called \emph{big leaf with respect to $X$}.

Therefore, if we work on $(\alpha(4),4)$-unbreakable graphs, then given a planar $\Hcal$-modulator $S$ of a graph $G$, either $S$ is a big-leaf planar $\Hcal$-modulator of $G$, that is, there is a unique component $C\in\cc(G-S)$ such that $|V(C)|\ge \alpha(4)$ 
and $|V(G)\setminus V(C)|<\alpha(4)+|N_G(V(C))|$
or for each $C\in\cc(G-S)$, $|V(C)|< \alpha(4)$.
In this paper, we will thus solve \Hpl\ on $(\alpha(4),4)$-unbreakable graphs, which, by applying \autoref{prop_CMSOMetaTheorem}, immediately implies \autoref{th_th}.

More specifically, we will split \Hpl\ into two complementary subproblems.
The first one is \blHpl, which is defined as follows.

\begin{center}
	\fbox{
		\begin{minipage}{12cm}
			\noindent\blHpl\\
			\noindent\textbf{Input}:~~A graph $G$.\\
			\textbf{Question}:~~Does $G$ admit a planar $\Hcal$-modulator $S$ such that there is $D\in \cc(G-S)$ of size at least $\alpha(4)$? 
		\end{minipage}
	}
\end{center}

This problem is easy to solve using a brute-force method.

\begin{lemma}\label{lem_big_leaf}
Let $\Hcal$ be a polynomial-time decidable graph class.
Then there is an algorithm that solves \blHpl\ on $(\alpha(4),4)$-unbreakable graphs in polynomial time.
\end{lemma}

\begin{proof}
Let $G$ be an $(\alpha(4),4)$-unbreakable graph.
For each set $X\subseteq V(G)$ of size at most four, we check whether there is a connected component $C$ of $G-X$ of size at least $\alpha(4)$. If that is the case, then $(A=X\cup V(C), B=V(G)\setminus V(C))$ is a separation of order at most four such that $|B\setminus A|< \alpha(4)$, given that $G$ is an $(\alpha(4),4)$-unbreakable graph.
For each such separation $(A,B)$, 
we consider all sets $S$ with $X=A\cap B\subseteq S\subseteq B$ (there are at most $2^{\alpha(4)-1}$ such sets). For every $S$, we check whether the torso of $S$ is planar and that each connected component of $G-S$ belongs to $\Hcal$.  If there is such a set $S$, then we conclude that $G$ is  a \yes-instance of \blHpl. 
If, for each such $(A,B)$ and each such $S$, we did not report a \yes-instance, then we report a \no-instance.
These checks take polynomial time given that $|B\setminus A|<\alpha(4)$. This concludes the proof. 
\end{proof}

The second subproblem is the following.

\begin{center}
	\fbox{
		\begin{minipage}{12cm}
			\noindent\slHpl\\
			\noindent\textbf{Input}:~~A graph $G$.\\
			\textbf{Question}:~~Does $G$ admit a planar $\Hcal$-modulator $S$ such that, for each $D\in \cc(G-S)$, $|V(D)|<\alpha(4)$?
		\end{minipage}
	}
\end{center}

Let $k\in\Nbbb$ and $\Hcal$ be a graph class.
Recall that $\Hcal^{(k)}$ is the subclass of  $\Hcal$ containing graphs with  at most $k$ vertices.
In this setting, \slHpl\ is exactly {\sc $\Hcal^{(\alpha(4)-1)}$-Planarity}.
More generally than \slHpl, we prove that \Hkpl\ is solvable in \FPT-time parameterized by $k$ in 
the following theorem (see \autoref{subsec_proof} for the proof).

 \begin{restatable}{theorem}{smallleaves}\label{lem_small_leaves}
Let $k\in\Nbbb$ and let $\Hcal$ be a polynomial-time decidable
hereditary graph class.
Then there is an algorithm that solves \Hkpl\ in time $f(k)\cdot n(n+m)$ for some computable function $f$.
\end{restatable} 

Given that a graph $G$ is a \yes-instance of \Hpl\ if and only if it is a \yes-instance of at least one of \blHpl\ and \slHpl, \autoref{th_th} immediately follows from \autoref{prop_CMSOMetaTheorem}, \autoref{lem_big_leaf}, and \autoref{lem_small_leaves}.

\smallskip
We do exactly the same for \hptd{$\Hcal$} and $\Hcal$-planar treewidth, but here, by \autoref{obs_clique}, we consider $(\alpha(k'),k')$-unbreakable graphs for $k'=4k$ and $k'=\max\{4,k+1\}$, respectively, instead of $k'=4$.

In this case, similarly to \autoref{lem_big_leaf}, we prove the following.
\begin{lemma}\label{cor_bigH}
Let $\Hcal$ be a hereditary graph class that is union-closed and 
such that there is an \FPT-algorithm that solves \pbdel parameterized by the solution size $h$ in time $f(h)\cdot n^c$.
Let $a,k\in\Nbbb$.
Let $\Gcal_k$ be the class of graphs with planar treewidth (resp. treewidth / planar treedepth / treedepth) at most $k$.
Let $k':=\max\{4,k+1\}$ (resp. $k+1$ / $4k$ / $k$).
Then there is an algorithm that, given an $(a,k')$-unbreakable graph $G$, either reports that $G$ has no big-leaf $\Gcal_k\triangleright\Hcal$-modulator, or outputs a $\Gcal_k\triangleright\Hcal$-modulator of $G$, in time $f(k)\cdot 2^{\Ocal((a+k)^2)}\cdot \log n\cdot (n^c+n+m)$.
\end{lemma}

The proof of \autoref{cor_bigH} is based on random sampling technique from \cite{ChitnisCHPP16desi} (\autoref{rd_sampling}).
Assuming $G$ has a big-leaf $\Gcal_k\triangleright\Hcal$-modulator $X$ with big leaf $D$,
we guess a set $U\subseteq V(G)$ such that $N_G(V(D))\subseteq U$ and $V(G)\setminus V(D)\subseteq V(G)\setminus U$, and deduce $X$ and $D$ from this guess. 
See \autoref{subsec_big_leaf} for the proof.

\smallskip
Meanwhile, similarly to \autoref{lem_small_leaves}, we prove the two following results, whose proofs can be found respectively in \autoref{subsec_small_ed} and \autoref{subsec_small_tw}.

\begin{lemma}\label{lem_edk}
Let $\Hcal$ be a graph class that is hereditary. 
Let $a,k\in\Nbbb$, $k'=4k$. 
Then there is an algorithm that, given an $(a,k')$-unbreakable graph $G$, checks whether 
$G$ has \hptd{$\Hcal^{(a-1)}$} at most $k$
in time $\Ocal_{k,a}(n\cdot(n+m))$.
\end{lemma}

\begin{lemma}\label{lem_twk}
Let $\Hcal$ be a graph class that is hereditary and union-closed.
Let $a,k\in\Nbbb$. 
Then there is an algorithm that, given an $(a,3)$-unbreakable graph $G$, checks whether $G$ has $\Hcal^{(a-1)}$-planar treewidth at most $k$ in time $\Ocal_{k,a}(n\cdot(n+m))$.
\end{lemma}

Therefore, \autoref{th_param} for $\p=\ptd$ immediately follows from \autoref{prop_CMSOMetaTheorem}, \autoref{cor_bigH}, and \autoref{lem_edk}, and \autoref{th_param} for $\p=\ptw$ immediately follows from \autoref{prop_CMSOMetaTheorem}, \autoref{cor_bigH}, and \autoref{lem_twk}.

\subsection{Outline of our technique of $\Hcal$-planarity (and $\Hcal$-planar treewidth)}\label{sec_outline}
The proof of \autoref{lem_small_leaves} for $\Hcal$-planar graphs is quite involved, and in particular, requires the introduction of many notions.
In this section, we sketch the proof, and 
the formal proof with all details is deferred to~\autoref{sec_renditions}.
Also, as we will argue at the end of this part, the proof of \autoref{lem_twk} for graphs of bounded $\Hcal$-planar treewidth is very similar to the one of \autoref{lem_small_leaves}.
For the sake of a good understanding, we give here informal 
definitions of the necessary notions. 
The notion of flat wall introduced below originates from \cite{RobertsonS95XIII}, see also 
\cite{KawarabayashiTW18anew}, and especially \cite{SauST24amor}.
The notion of rendition comes from \cite{KawarabayashiTW18anew}, see also \cite{kawarabayashi2020quickly}, where the more general notion of $\Sigma$-decomposition is defined (in our case, $\Sigma$ is the sphere).

\paragraph{Renditions.} See \autoref{fig_rendition} for an illustration and \autoref{subsec_rend} for accurate definitions.
A \emph{society} is a pair $(G,\Omega)$, where $G$ is a graph and $\Omega$ is a cyclic permutation with $V(\Omega)\subseteq V(G)$.
A \emph{rendition} $\rho$ of $(G,\Omega)$ is a pair $(\Gamma,\Dcal)$, where $\Gamma$ is a drawing (with crossings) of $G$ in a disk $\Delta$ such that the vertices of $\Omega$ are on the boundary of $\Delta$, in the order dictated by $\Omega$, and $\Dcal$ is a set of pairwise internally-disjoint disks  called \emph{cells} whose boundary is crossed by at most three vertices of $G$ and no edges, and such that every edge of $G$ is contained inside a cell. 
The vertices on the boundary of cells are called \emph{ground vertices}.
A rendition $\rho$ is \emph{well-linked} if every cell $c$ of $\rho$ has $|\tilde{c}|$ vertex disjoint paths between $\tilde{c}$ and  $V(\Omega)$, where $\tilde{c}$ is the set of vertices on the boundary of $c$.
In this sketch, we will only use renditions for the sake of simplicity, even though we actually need the more general notion of \emph{\sdecomp} in the formal proof, which is essentially a rendition on the sphere instead of a disk.
\begin{figure}[h]
\begin{minipage}[c]{.5\textwidth}
\center
\includegraphics[scale=0.5]{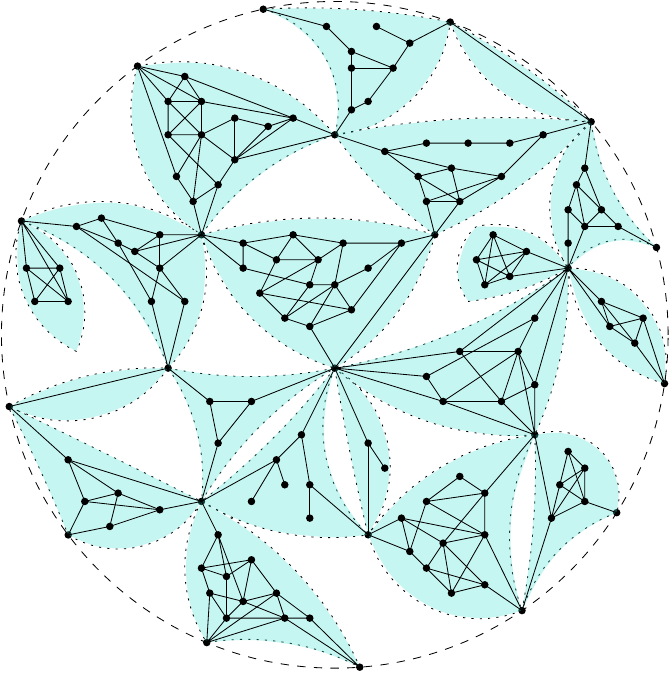}
\caption{A rendition.}
\label{fig_rendition}
\end{minipage}%
\begin{minipage}[c]{.5\textwidth}
\center
\includegraphics[scale=0.5]{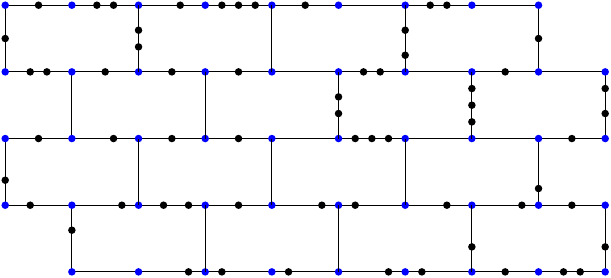}
\caption{A  wall.}
\label{fig_wall}
\end{minipage}
\end{figure}

\paragraph{Flat walls.} See \autoref{fig_wall} for an illustration of a wall and \autoref{subsec_flat} for accurate definitions.
A \emph{wall} is a graph that is a subdivision of a hexagonal grid.
A wall $W$ is \emph{flat} in $G$ if there is a separation $(X,Y)$ of $V(G)$ such that $X\cap Y$ is a part of the perimeter $D(W)$ of $W$ and contain the degree-3 vertices of $D(W)$, that the vertices of the interior of $W$ are in $Y$, and that $(G[Y],\Omega)$ has a rendition, that we can choose to be well-linked, where $\Omega$ is a cyclic permutation of the perimeter of $W$ with $V(\Omega)=X\cap Y$.
$G[Y]$ is called the \emph{compass} of $W$.
An \emph{apex grid} is a graph composed of a square grid and an additional vertex adjacent to every vertex of the grid.

\paragraph{The first step of the algorithm}\hspace{-0.2cm}is to find a flat wall in $G$.
For this, one would usually use the classical Flat Wall theorem (see \autoref{prop_FWth} and also \cite{KawarabayashiTW18anew,RobertsonS95XIII,SauST24amor,GiannopoulouT13opti}) that says that either $G$ contains a big clique as a minor, or $G$ has bounded treewidth, or there is a set $A\subseteq V(G)$, called \emph{apex set}, of bounded size and a wall $W$ that is flat in $G-A$.
However, in our case, we want $W$ to be flat in $G$ with no apex set.
Therefore, we use a variant of the Flat Wall theorem that is implicit in~\cite[Theorem 1]{GiannopoulouT13opti} and~\cite[Lemma 4.7]{KorhonenPS24mino}.
However, we are not aware of an algorithmic statement of this variant, and so we prove it here (see \autoref{prop_flatwallth}).
We thus prove that there is a function $f$ and an algorithm that, given a graph $G$ and $h,r\in\Nbbb$, 
outputs one of the following:
\begin{itemize}
\item {\bf Case 1:} a report that $G$ contains an apex  $(h\times h)$-grid  as a minor, or
\item {\bf Case 2:} a tree decomposition of $G$ of width at most $f(h)\cdot r$, or
\item {\bf Case 3:} a wall $W$ of height $r$ that is flat in $G$ and whose compass has treewidth at most $f(h)\cdot r$.
\end{itemize}

We apply the above algorithm to the instance graph $G$ of \Hkpl{} for $h,r=\Theta(\sqrt{k})$ (precise values are given in~\autoref{subsec_proof}) 
 and consider three cases depending on the output in the {\bf next step}.

\paragraph{Case 1:} If $G$ contains the apex grid of height $h$ as a minor, then we are able to argue that $G$ is a \no-instance (\autoref{lem:obstructions}). This is easy if the apex grid is a subgraph of $G$ as for any separation of the apex grid of order at most four, there is a nonplanar part whose size is bigger than $k$, for our choice of $h$. We use this observation to show that the same holds if the apex grid is a minor.

\paragraph{Case 2:} If $G$ has a tree decomposition of bounded treewidth, then we can use Courcelle's theorem to solve the problem.
Recall that Courcelle's theorem (\autoref{Courcelle}) says that if a problem can be defined in CMSO logic, then it is solvable in \FPT-time on graphs of bounded treewidth.
In our case, observe that, in \autoref{lem_small_leaves}, we do not ask for $\Hcal$ to be CMSO-definable.
Nevertheless, given that there is a finite number of graphs of size at most $k$, $\Hcal^{(k)}$ is trivially CMSO-definable by enumerating all the graphs in $\Hcal^{(k)}$.
Then, \Hkpl\ is easily expressible in CMSO logic (\autoref{obs_CMSO}), hence the result.

It remains to consider the third and most complicated case.

\paragraph{Case 3:} There is a flat wall $W$ in $G$ of height $r$ whose compass $G'$ has  treewidth upper bounded by $f(h)\cdot r$.
Then we apply Courcelle's theorem again, on $G'$ this time, since it has bounded treewidth.
If $G'$ is not $\Hcal^{(k)}$-planar, then neither is $G$ because $\Hcal$ is hereditary, so we conclude that $G$ is a no-instance of \Hkpl.
Otherwise, $G'$ is $\Hcal^{(k)}$-planar.
Let $v$ be a central vertex of $W$.
Then we prove that $v$ is an \emph{irrelevant vertex} in $G$, i.e. a vertex such that $G$ and $G-v$ are equivalent instances of the problem.
Therefore, we call our algorithm recursively on the instance $G-v$ and return the obtained answer.

\medskip
The description of our algorithm in Case~3 is  simple. However, proving its correctness is the crucial and most technical part of our paper which deviates significantly from other applications of the irrelevant vertex technique. The reason for this is that the modulator is a planar graph that might have big treewidth and might be
spread ``all around the input graph'', therefore it does not satisfy any locality condition that might make possible the application of standard irrelevant-vertex arguments (such as those crystallized in algorithmic-meta-theorems in \cite{FominGSST21compound,GolovachST23model,SauST2024parame}). 
Most of the technical part of our paper is devoted to dealing with this 
situation where the modulator is not anymore ``local''.\medskip

It is easy to see that if $G$ has a planar $\Hcal^{(k)}$-modulator, then the same holds for $G'$ and $G-v$ as $\Hcal$ is hereditary. The difficult part is to show the opposite implication.   

First, we would like to emphasize that in the standard irrelevant vertex technique of Robertson and Seymour~\cite{RobertsonS95XIII}, the existence of a big flat wall $W$ with some specified properties implies that a central vertex is irrelevant. Here, we cannot make such a claim because the compass of $W$ may be nonplanar and we cannot guarantee that $v$ does not belong to $G-X$ for 
a (potential) planar $\Hcal^{(k)}$-modulator $X$. Therefore, we have to verify that $G'$ is a yes-instance before claiming that $v$ is irrelevant.  Then we have to show that if both $G'$ and $G-v$ have  
planar $\Hcal^{(k)}$-modulators, then $G$ also has a planar $\Hcal^{(k)}$-modulator. The main idea is to combine modulators for $G'$ and $G-v$ and construct a modulator for $G$. However, this is nontrivial because the choice of a connected component $C$ which is outside of a modulator restricts the possible choices of other such components and this may propagate arbitrarily around a rendition. 
In fact, this propagation effect is used in the proof of \autoref{thm_lower} where we show that \Hpl{} is \NP-complete for nonhereditary classes. Still, for hereditary classes $\Hcal$, we are able to show that if both $G'$ and $G-v$ are \yes-instances of \Hkpl, then
these instances have some particular solutions that could be glued together to obtain a planar  $\Hcal^{(k)}$-modulator for $G$. In the remaining part of this section, we informally explain the choice of compatible solutions.

\paragraph{$\Hcal$-compatible renditions.} 
Let $G$ be an $\Hcal$-planar graph. Let $X$ be a planar $\Hcal$-modulator in $G$, and $\Gamma$ be a planar embedding of the torso of $X$.
For each component $D$ in $G-X$, the neighborhood of $D$ induces a clique of size at most four in the torso of $X$.
Therefore, it is contained in a disk $\Delta_D$ in $\Gamma$ whose boundary is the outer face of the clique, that is, a cycle of size at most three.
Therefore, for each pair of such disks, either one is included in the other, or their interiors is disjoint.
If we take all \textsl{inclusion-wise maximal} such disks, along with the rest of the embedding disjoint from these disks, this then essentially defines the set of cells $\Dcal$ of a rendition $\rho$ of $G$.
In particular, for each cell $c$ of $\rho$, the graph induced by $c$ has a planar $\Hcal$-modulator $X_c$ such that the torso of $X_c$ admits a planar embedding with the vertices of $\tilde{c}$ on the outer face.
Such a rendition $\rho$ (resp. cell $c$) is said to be \emph{$\Hcal$-compatible}, and it is direct to check that if $G$ admits an $\Hcal$-compatible rendition, then $G$ is $\Hcal$-planar.
This is what we prove, with additional constraints and more formally, in \autoref{obs_sol_compatible}. In particular, we talk here in terms of renditions for simplicity, while we actually use \sdecomps, which are similar to renditions, but defined on the sphere instead of a disk, with no cyclic permutation $\Omega$ defining a boundary.
See \autoref{fig_compatible} for an illustration and \autoref{subsec_compatible} for an accurate definition.

\begin{figure}[h]
\centering
\includegraphics[scale=0.8]{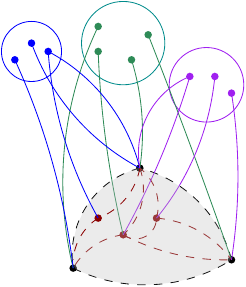}
\caption{An $\Hcal$-compatible cell $c$ (in gray). The $\Hcal$-modulator $X$ is the set of red and black vertices, with the black vertices being the vertices of $\tilde{c}$. The blue, green, and purple balls are the connected components obtained after removing $X$, and the dashed lines are used for the edges of the torso that are not necessarily edges of the graph.}
\label{fig_compatible}
\end{figure}
Therefore, since $G'$ and $G-v$ are $\Hcal^{(k)}$-planar, there are a rendition $\rho'$ of $G'$ and a rendition $\rho_v$ of $G-v$ that are $\Hcal^{(k)}$-compatible.
A way to prove that $G$ is $\Hcal^{(k)}$-planar would then be to 
 glue these two renditions together to find an $\Hcal^{(k)}$-compatible rendition $\rho^*$ of $G$.
For instance, we may take a disk $\Delta$ in the rendition corresponding to the flat wall $W$ with $v\in\Delta$ and define $\rho^*$ to be the rendition equal to $\rho'$ when restricted to $\Delta$, and equal to $\rho_v$ outside of $\Delta$.
A major problem in such an approach is however that it may only work if we are able to guarantee that the cells of $\rho'$ and $\rho_v$ are pairwise disjoint on the boundary of $\Delta$.
Given that a graph $G$ may have many distinct renditions, the ideal would be to manage to find a \textsl{unique} rendition of $G$ that is minimal or maximal in some sense and could then guarantee that its cells are crossed\footnote{A cell $c_1$ of a rendition $\rho_1$ \emph{crosses} a cell $c_2$ of a rendition $\rho_2$ if their vertex sets have a non-empty intersection, but that neither of them is contained in the other.} by no cell of another rendition of $G$.
Unfortunately, we do not have such a result; {in fact, we have no way to ``repair'' these renditions so as to achieve cell-compatibility}.   
Instead, we prove a new property of renditions that allows us to glue renditions together.

\paragraph{Well-linked and ground-maximal renditions.}
Recall that a rendition $\rho$ is well-linked if every cell $c$ of $\rho$ has $|\tilde{c}|$ vertex disjoint paths to $V(\Omega)$, where $\tilde{c}$ is the set of vertices on the boundary of $c$.
A rendition $\rho$ is \emph{more grounded} than a rendition $\rho'$ if every cell of $\rho$ is contained in a cell of $\rho'$, or equivalently, if $\rho$ can be obtained from $\rho'$ by splitting cells to ground more vertices.
A rendition $\rho$ is \emph{ground-maximal} if no rendition is more grounded than $\rho$.
See \autoref{subsec_compatible} and \autoref{subsec_compare} for the accurate definition of ground-maximality and well-linkedness respectively.
Note that the idea of working with tight and ground-maximal renditions already appeared in~\cite{LokshtanovPPS22high}.

Let $(G,\Omega)$ be a society.
What we prove is that for any ground-maximal rendition $\rho_1$ of $(G,\Omega)$ and for any well-linked rendition $\rho_2$ of $(G,\Omega)$, $\rho_1$ is always more grounded than $\rho_2$, or, in other words, every cell $c_1$ of $\rho_1$ is contained in a cell $c_2$ of $\rho_2$ (see \autoref{cor_well-ground}).
To prove this, we assume toward a contradiction that either $c_1$ and $c_2$ cross, or that $c_2$ is contained in $c_1$.
In both cases, we essentially prove that we can replace $c_1$ by the cells obtained from the restriction of $\rho_2$ in $c_1$, which gives a rendition $\rho^*$ that is more grounded than $\rho_1$, contradicting the ground-maximality of $\rho_1$.
See \autoref{fig_minmax} for some illustration.

\begin{figure}[h]
\centering
\includegraphics[scale=1]{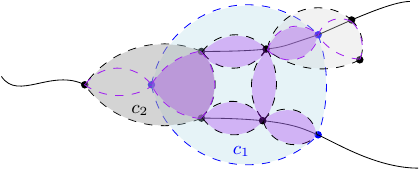}
\caption{$c_2$ is a cell of a well-linked rendition $\rho_2$ (in gray). If $c_2$ crosses a cell $c_1$ of another (blue) rendition $\rho_1$, then there is a new (purple) rendition $\rho^*$ that is more grounded than the blue one.}
\label{fig_minmax}
\end{figure}

Going back to $G'$ and $G-v$, given that $\Hcal$, and thus $\Hcal^{(k)}$, is \emph{hereditary}, it implies that any rendition that is more grounded than an $\Hcal^{(k)}$-compatible rendition is also $\Hcal^{(k)}$-compatible.
Therefore, we can assume $\rho'$ and $\rho_v$ to be ground-maximal.
Additionally, by definition, there is a rendition $\rho$ of $(G',\Omega)$ witnessing that $W$ is a flat wall of $G$, where $V(\Omega)$ is a vertex subset of the perimeter of $W$.
Then, by known results (\autoref{lem_rend_to_min_or_max}), we can assume $\rho$ to be well-linked.

Therefore, the cells of $\rho'$ and $\rho_v$ are contained in cells of $\rho$ when restricted to $G'-v$.
Actually, it is a bit more complicated given that, while $\rho$ is a rendition of the society $(G',\Omega)$, we just know that $\rho'$ and $\rho_v$ are renditions of $(G',\Omega')$ and $(G-v,\Omega_v)$ respectively, for some cyclic permutations $\Omega'$ and $\Omega_v$ that we do not know.
What we actually prove, thanks to the structure of the wall, is that the cells of $\rho'$ and $\rho_v$ in $G'-v$ that are far enough from $v$ and the perimeter of $W$ are contained in cells of $\rho$.
For this, rather than \autoref{cor_well-ground} that says that a ground-maximal rendition of a society is more-grounded than a well-linked rendition of the exactly same rendition, we prefer to use an auxiliary lemma (see \autoref{cell_in_cell2}) that says that, given two renditions $\rho_1$ and $\rho_2$ of possibly different societies, if a ground-maximal cell $c_1$ of $\rho_1$ that does not intersect $V(\Omega)$ intersects only well-linked cells of $\rho_2$, then $c_1$ is contained in some cell $c_2$ of $\rho_2$.

\begin{figure}[h]
\centering
\includegraphics[scale=.9]{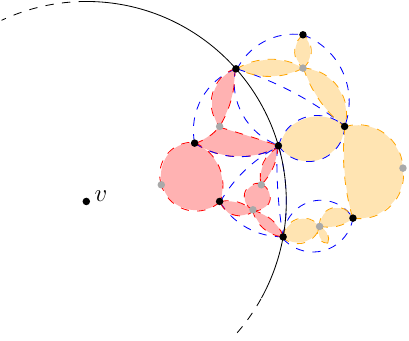}
\caption{The black circle represents the boundary of $\Delta$, $\rho$ is represented in blue, $\rho'$ in red, and $\rho_v$ in orange.}
\label{fig_glue}
\end{figure}

Given that we found a strip around $v$ in $W$ where the cells of $\rho'$ and $\rho_v$ are contained in cells of $\rho$, we can pick a disk $\Delta$, whose boundary is in this strip, that intersects $\rho$ only on ground vertices. In other words, the interior of each cell of $\rho$ is either contained in $\Delta$ or outside of $\Delta$.
But then, the same holds for the cells of $\rho'$ and $\rho_v$.
Therefore, we can finally correctly define the rendition $\rho^*$ of $G$ that is equal to $\rho'$ when restricted to $\Delta$ and equal to $\rho_v$ outside of $\Delta$ (see \autoref{lem_newirr}).
See \autoref{fig_glue} for an illustration.
Each cell of $\rho^*$ is either a cell of $\rho'$ or of $\rho_v$ and is thus $\Hcal^{(k)}$-compatible.
Therefore, $\rho^*$ is $\Hcal^{(k)}$-compatible, and thus $G$ is $\Hcal^{(k)}$-planar.
This concludes this sketch of the proof.

\medskip
Interestingly, this new irrelevant vertex technique (\autoref{lem_newirr}) can essentially be applied to any problem $\Pi$ as long as we can prove the following:
\begin{enumerate}
\item there is a big enough flat wall $W$ in the input graph $G$, and
\item there is a hereditary property $\Pcal_\Pi$ such that $\Pi$ can be expressed as the problem of finding a ground-maximal rendition $\rho$ where each cell of $\rho$ has the property $\Pcal_\Pi$.
\end{enumerate}
This is actually what we do for $\Hcal$-planar treewidth, where, instead of asking the cells to be $\Hcal$-compatible, we ask them to have another property $\Pi_{\Hcal,k}$ (see \autoref{subsec_twrend}).

\medskip
For item 2 above, there is actually also another constraint that was hidden under the carpet for this sketch, demanding that we should be able to choose $\rho$ such that no cell of $\rho$ contains $W$.
In our case, this is what requires us to consider an embedding on the sphere instead of a disk to wisely choose the disks that will become $\Hcal^{(k)}$-compatible cells, hence our use of sphere decompositions instead of renditions.
Additionally, we crucially use the fact that a big wall is not contained in a component of bounded size, so the same argument would not work if we were to consider {\sc $\Hcal$-Planarity} in general instead of {\sc $\Hcal^{(k)}$-Planarity} (it could actually work as long as $\Hcal$ is a class of bounded treewidth, given the parametric duality between walls and treewidth).

\subsection{Changes for \hptd{$\Hcal$}}\label{subsec_changes_ed}

We sketch here \autoref{lem_edk}, or more exactly, we explain the differences with \autoref{lem_small_leaves}.
Contrary to the $\Hcal^{(a)}$-planarity case, a big apex grid is not an obstruction for a graph of bounded \hptd{$\Hcal^{(a)}$}.
Therefore, we cannot use the Flat Wall theorem variant of \autoref{prop_flatwallth}.
Instead, we use the classical Flat Wall theorem (\autoref{prop_FWth}).
If it reports that $G$ has bounded treewidth, then we use Courcelle's theorem (\autoref{Courcelle}) to conclude.
If we obtain that $G$ has a clique-minor of size $a+4k$, then this is a \no-instance.
So we can assume that we find a small apex set $A$ and a wall $W$ that is flat in $G-A$ and whose compass has bounded treewidth.
Let $A^+$ be the set of vertices of $A$ with many disjoint paths to $W$ and $A^-=A\setminus A^+$.
Assuming $G$ has a certifying elimination sequence $X_1,\dots, X_k$, what we can essentially prove is that $W$ is mostly part of $X_i$ for some $i\in[k]$, and that the vertices of $A^+$ must belong to $\bigcup_{j<i}X_j$ and can hence be somewhat ignored.
Hence, we do the following. We find a subwall $W'$ of $W$ that avoids the vertices of $A^-$, and thus is flat in $G-A^+$.
We divide $W'$ further in enough subwalls $W_1,\dots,W_r$ (which are flat in $G-A^+$) so that we are sure that one of them, say $W_p$, is totally contained in $X_i$ (though we do not know which one).
Given that the compass of the walls $W_\ell$ have bounded treewidth, we can use Courcelle's theorem to compute the \hptd{$\Hcal^{(a)}$} $d_\ell$ of each of them.
If, for each $\ell\in[r]$, $d_\ell>k$, then this is a \no-instance.
Otherwise, we choose $\ell$ such that $d_\ell$ is minimum.
Then, given a central vertex $v$ of $G-v$, we report a \yes-instance if and only if $G-v$ has \hptd{$\Hcal$} at most $k$.

The argument for the irrelevancy of $v$ uses the core result as in the $\Hcal^{(a)}$-planarity case, that is \autoref{lem_newirr}.
What we essentially prove is that, given a certifying elimination sequence $Y_1,\dots, Y_k$ of $G-v$, most of $W_\ell$ is in $Y_i$, and that, as such, we can replace $Y_i,\dots,Y_k$ according to the certifying elimination sequence we found for the compass of $W_\ell$, to obtain a certifying elimination sequence $Y_1',\dots, Y_k'$ of $G'$.

\section{The \FPT\-algorithm for \Hkpl}\label{sec_renditions} 

To prove \autoref{lem_small_leaves}, we need the notion of \emph{\sdecomps}.
In \autoref{subsec_rend}, we thus define \sdecomps\ and related notions.
In \autoref{subsec_flat}, we define \emph{flat walls} and introduce the Flat Wall theorem.
In \autoref{subsec_obs}, we provide an obstruction set to \Hkpl. 
In \autoref{subsec_compatible}, we observe that a \yes-instance of \Hkpl\ is essentially a graph with a special \sdecomp, called \emph{$\Hcal^{(k)}$-compatible \sdecomp}, and we introduce the notion of \emph{ground-maximality} for \sdecomps.
In \autoref{subsec_compare}, we introduce the notion of \emph{well-linkedness} and prove that a ground-maximal \sdecomp\ is always \emph{more grounded} than a well-linked \sdecomp.
Finally, in \autoref{subsec_proof}, we combine the results of the previous subsections together to prove \autoref{lem_small_leaves}.

\subsection{Drawing in a surface}

We begin by defining drawings, \sdecomps, vortices, and other related notions.\label{subsec_rend}
  
 \paragraph{Drawing a graph in the sphere.}
  Let $\mathbb{S}^2=\{(x,y,z)\in\Rbbb^3\mid x^2+y^2+z^2=1\}$ be the sphere.
  A \emph{drawing} (with crossings) in $\mathbb{S}^2$ is a triple $\Gamma=(U,V,E)$ such that
  \begin{itemize}[itemsep=-2pt] 
  \setlength\itemsep{0em}
 \item $V$ is a finite set of point of $\mathbb{S}^2$ and  $E$ is a finite family of subsets of   $\mathbb{S}^2$,
 \item    $V\cup(\bigcup_{e\in E}e)=U$ and $V\cap (\bigcup_{e\in E}e)=\emptyset,$
  \item for every $e\in E,$ $e=h((0,1)),$ where $h\colon[0,1]_{\mathbb{R}}\to U$ is a homeomorphism onto its image with $h(0),h(1)\in V,$ 
  and
  \item if $e,e'\in E$ are distinct, then $|e\cap e'|$ is finite.
  \end{itemize}
  We call the elements of the set $V,$ sometimes denoted by $V(\Gamma),$ the \emph{vertices of $\Gamma$} and the set $E,$ denoted by $E(\Gamma),$ the set of \emph{edges of $\Gamma$}. We also denote $U(\Gamma)=U.$
  If $G$ is a graph and $\Gamma=(U,V,E)$ is a drawing with crossings in $\mathbb{S}^2$ 
  such that $V$ and $E$ naturally correspond to $V(G)$ and $E(G)$ respectively, we say that $\Gamma$ is a \emph{drawing of $G$} in $\mathbb{S}^2$ 
  (possibly with crossings).

 \medskip 
 We remind that a \emph{closed disk} $D$ on $\mathbb{S}^2$ is a set of points homeomorphic to the set $\{(x,y)\in \mathbb{R}^2\mid x^2+y^2\leq 1\}$ on the Euclidean plane. We use $\bd(D)$ to denote the boundary of $D$.

\paragraph{Sphere decompositions.}
  A \emph{\sdecomp} of a graph $G$ is a pair $\delta=(\Gamma,\mathcal{D}),$ where $\Gamma$ is a drawing of $G$ in the sphere $\mathbb{S}^2$ (with crossings), and $\mathcal{D}$ is a collection of closed disks, each a subset of $\mathbb{S}^2$, such that 
  \begin{enumerate}[itemsep=-2pt]
  \setlength\itemsep{0em}
  \item the disks in $\mathcal{D}$ have pairwise disjoint interiors, 
  \item the boundary of each disk $D$ in $\mathcal{D}$ intersects $\Gamma$ in vertices only, and $|\bd(D)\cap V(\Gamma)|\le3$, 
  \item if $\Delta_1,\Delta_2\in\mathcal{D}$ are distinct, then $\Delta_1\cap\Delta_2\subseteq V(\Gamma),$ and 
  \item every edge of $\Gamma$ belongs to the interior of one of the disks in $\mathcal{D}.$ 
  \end{enumerate} 
The notion of a sphere decomposition has been introduced in \cite{kawarabayashi2020quickly} in the more general case where, instead 
of a sphere, we have a surface $\Sigma$ with a boundary, where it is called a $\Sigma$-decomposition. It has also been recently used in \cite{thilikos2023approx}.

\paragraph{Sphere embeddings.}
A \emph{sphere embedding} of a graph $G$ is a \sdecomp\ $\delta = (\Gamma,\Dcal)$ where  $\Dcal$ is a collection of closed disks such that, for any disk in $\Dcal,$  only a single edge of $\Gamma$ is drawn in its interior.  In words, $\Gamma$ is a drawing of $G$ on the sphere without crossings.

\paragraph{Nodes, cells, and ground vertices.}
For a \sdecomp\ $\delta = (\Gamma, \Dcal)$, 
  let $N$ be the set of all vertices of $\Gamma$ that do not belong to the interior of the disks in $\mathcal{D}.$ 
  We refer to the elements of $N$ as the \emph{nodes} of $\delta.$ 
  If $\Delta\in\mathcal{D},$ then we refer to the set $\Delta-N$ as a \emph{cell} of $\delta.$ 
  We denote the set of nodes of $\delta$ by $N(\delta)$ and the set of cells by $C(\delta).$ 
  For a cell $c\in C(\delta)$ the set of nodes that belong to the closure of $c$ is denoted by $\tilde{c}.$ Given a cell $c\in C(\delta),$ we define 
  its \emph{disk} as 
  $\Delta_{c}=\bd(c)\cup c.$
  We define $\pi_{\delta}\colon N(\delta)\to V(G)$ to be the mapping that assigns to every node in $N(\delta)$ the corresponding vertex of $G.$
We also define \emph{ground vertices} in $\delta$ as $\ground(\delta) = π_{\delta}(N(\delta)).$
For a cell $c\in C(\delta)$, we define the graph $\sigma_{\delta}(c),$ or $\sigma(c)$ if $\delta$ is clear from the context, to be the subgraph of $G$ consisting of all vertices and edges drawn in $\Delta_{c}.$ 
Note that, for any cell $c\in C(\delta)$ such that $\sigma_{\delta}(c)$ is not connected, we can split $c$ into $|\cc(\sigma_{\delta}(c))|$ cells, and obtain a sphere decomposition $\delta'$ such that each cell $c\in C(\delta')$ is such that $\sigma_{\delta}(c)$ is connected.
Hence, without loss of generality, when introducing a sphere rendition $\delta$, we assume $\sigma_\delta(c)$ to be connected for each $c\in C(\delta)$.
 
\paragraph{Deleting a set.}
Given a set $X\subseteq V(G),$ we denote by $\delta-X$ the \sdecomp\ $\delta'=(\Gamma',\mathcal{D}')$ of $G-X$ where $\Gamma'$ is obtained by the drawing $\Gamma$ after removing 
all points in $X$ and all drawings of edges with an endpoint in $X.$ For every point $x\in π_{\delta}^{-1}(X\cap N(\delta))$, we pick $\Delta_{x}$ to be an open disk containing $x$ and not containing any point of some remaining vertex or edge and such that no two such disks intersect.
We also 
set $\Delta_{X}=\bigcup_{x\in π_{\delta}^{-1}(X)} \Delta_x$ and we define $\mathcal{D}'=\{D\setminus \Delta_{X}\mid D\in\Dcal\}$.
Clearly, there is a one to one correspondence between the cells of $\delta$ and the cells of $\delta'.$
If a cell $c$ of $\delta$ corresponds to a cell $c'$ of $\delta',$ then 
we call $c'$ the \emph{heir} of $c$ in $\delta'$ and we call $c$ the \emph{precursor}
of $c'$ in $\delta.$

\paragraph{Grounded graphs.}
Let $δ$ be a \sdecomp\ of a graph $G$ in $\mathbb{S}^2.$ 
We say that a cycle $C$ of $G$ is \emph{grounded in $δ$} if $C$ uses edges of $\sigma(c_{1})$ and $\sigma(c_{2})$ for two distinct cells $c_{1}, c_{2} \in C(δ).$
A $2$-connected subgraph $H$ of $G$ is said to be \emph{grounded in $δ$} if every cycle in $H$ is grounded in $δ.$

\paragraph{Tracks.}
Let $δ$ be a \sdecomp\ of a graph $G$ in $\mathbb{S}^2.$ 
For every cell $c \in C(δ)$ with $|\tilde{c}| = 2$, we select one of the components of $\bd(c) - \tilde{c}.$  
This selection is called a \emph{tie-breaker} in $δ,$ and we assume every $Σ$-decomposition to come equipped with a tie-breaker. 
Let $C$ be a cycle grounded in $δ$.
We define the \emph{track} of $C$ as follows.
Let $P_{1}, \dots, P_{k}$ be distinct maximal subpaths of $C$ such that $P_{i}$ is a subgraph of $\sigma(c)$ for some cell $c.$ Fix an index $i.$
The maximality of $P_{i}$ implies that its endpoints are $π_δ(n_{1})$ and $π_δ(n_{2})$ for distinct $δ$-nodes $n_{1}, n_{2} \in N(δ).$
If $|\tilde{c}| = 2,$ define $L_{i}$ to be the component of $\bd(c) - \{ n_{1}, n_{2} \}$ selected by the tie-breaker, and if $|\tilde{c}| = 3,$ define $L_{i}$ to be the component of $\bd(c) - \{ n_{1}, n_{2} \}$ that is disjoint from $\tilde{c}.$
Finally, we define $L'_{i}$ by slightly pushing $L_{i}$ to make it disjoint from all cells in $C(δ).$ We define such a curve $L'_{i}$ for all $i$ while ensuring that the curves intersect only at a common endpoint.
The \emph{track} of $C$ is defined to be $\bigcup_{i \in [k]} L'_{i}.$
So the track of a cycle is the homeomorphic image of the unit circle.

\paragraph{Societies.}
  Let $\Omega$ be a cyclic permutation of the elements of some set which we denote by $V(\Omega).$
  A \emph{society} is a pair $(G,\Omega),$ where $G$ is a graph and $\Omega$ is a cyclic permutation with $V(\Omega)\subseteq V(G).$

\paragraph{Renditions.}
Let $(G,\Omega)$ be a society, and let $\Delta$ be a closed disk in $\mathbb{S}^2$.
A \emph{rendition} of $(G,\Omega)$ is a \sdecomp\ $ρ=(\Gamma,\Dcal)$ of $G$ such that, for each $D\in\Dcal$, $D\subseteq\Delta$, and $\pi_{ρ}(N(ρ)\cap \bd(\Delta))=V(\Omega)$,
mapping one of the two cyclic orders (clockwise or counter-clockwise) of $\bd(\Delta)$ to the order of $\Omega.$

\paragraph{$δ$-aligned disks.}
We say  a closed disk $Δ$ in $\mathbb{S}^2$ is \emph{$δ$-aligned} if
 its boundary intersects $Γ$ only in nodes of $δ$.
 We denote by $Ω_{Δ}$ one of the cyclic orderings of the vertices on the boundary of $Δ.$
We define the \emph{inner graph} of a $δ$-aligned closed disk $Δ$ as 
$$\inG_{δ}(Δ) \coloneqq \bigcup_{\textrm{$c \in C(δ)$ and $c \subseteq Δ$}} \sigma(c)$$ 
and the \emph{outer graph of $Δ$} as 
$$\outG_{δ}(Δ) \coloneqq \bigcup_{\textrm{$c \in C(δ)$ and $c \cap Δ \subseteq \ground(δ)$}} \sigma(c).$$ 
  
If $Δ$ is $δ$-aligned, we define $Γ \cap Δ$ to be the drawing of $\inG_{δ}(Δ)$ in $Δ$ which is the restriction of $Γ$ in $Δ.$
If moreover $|V(Ω_{Δ})|≥4,$ we denote by $δ[Δ]$ the rendition $(Γ \cap Δ, \{ Δ_{c} \in \mathcal{D} \mid c \subseteq Δ\})$ of $(\inG_{δ}(Δ), Ω_{Δ})$ in $Δ.$

Let $δ = (\Gamma,\mathcal{D})$ be a \sdecomp\ of a graph $G$ in $\mathbb{S}^2.$
Let $C$ be a cycle in $G$ that is grounded in $δ,$ such that the track $T$ of $C$ bounds a closed disk $Δ_{C}$ in $\mathbb{S}^2.$
We define the \emph{outer} (resp. \emph{inner}) \emph{graph} of $C$ in $δ$ as the graph $\outG_{δ}(C) := \outG_{δ}(Δ_C)$ (resp. $\inG_{δ}(C) := \inG_{δ}(Δ_C)$).

\subsection{Flat walls}\label{subsec_flat}

We will need to use an appropriate version of the Flat Wall theorem \cite{RobertsonS95XIII,KawarabayashiTW18anew}.
Let us first define walls, flat walls, and the relevant vocabulary.

\paragraph{Walls.}
Let $k,r\in\Nbbb$. The
\emph{$(k\times r)$-grid}, denoted by $\Gamma_{k,r}$, is the
graph whose vertex set is $[k]\times[r]$ and two vertices $(i,j)$ and $(i',j')$ are adjacent if and only if $|i-i'|+|j-j'|=1$.
An \emph{elementary $r$-wall} (see \autoref{fig_wall}), for some odd integer $r\geq 3$, is the graph obtained from a
$(2 r\times r)$-grid
with vertices $(x,y)
	\in[2r]\times[r]$,
after the removal of the
``vertical'' edges $\{(x,y),(x,y+1)\}$ for odd $x+y$, and then the removal of
all vertices of degree one.
Notice that, as $r\geq 3$, an elementary $r$-wall is a planar graph
that has a unique (up to topological isomorphism) embedding in the plane $\Rbbb^{2}$
such that all its finite faces are incident to exactly six
edges.
The {\em perimeter} of an elementary $r$-wall is the cycle bounding its infinite face,
while the cycles bounding its finite faces are called {\em bricks}.
Also, the vertices
in the perimeter of an elementary $r$-wall that have degree two are called {\em pegs},
while the vertices $(1,1), (2,r), (2r-1,1), (2r,r)$ are called {\em corners} (notice that the corners are also pegs).

An {\em $r$-wall} is any graph $W$ obtained from an elementary $r$-wall $\bar{W}$
after subdividing edges. A graph $W$ is a {\em wall} if it is an $r$-wall for some odd $r\geq 3$
and we refer to $r$ as the {\em height} of $W$. Given a graph $G$,
a {\em wall of} $G$ is a subgraph of $G$ that is a wall.
We always assume that for every $r$-wall, the number $r$ is  odd.

We call the vertices of degree three of a wall $W$ {\em branch vertices}.
We use $D(W)$ to denote the perimeter of the wall $W$.

\paragraph{Tilts.}
The \emph{interior} of a wall $W$ is the graph obtained from $W$ if we remove from it all edges of $D(W)$ and all vertices of $D(W)$ that have degree two in $W$.
Given two walls $W$ and $\tilde{W}$ of a graph $G$, we say that $\tilde{W}$ is a \emph{tilt} of $W$ if $\tilde{W}$ and $W$ have identical interiors.

\paragraph{Layers.}
The {\em layers} of an $r$-wall $W$ are recursively defined as follows.
The first layer of $W$ is its perimeter.
For $i=2,\ldots,(r-1)/2$, the $i$-th layer of $W$ is the $(i-1)$-th layer of the subwall $W'$
obtained from $W$ after removing from $W$ its perimeter and
removing recursively all occurring vertices of degree one.
We refer to the $(r-1)/2$-th layer as the {\em inner layer} of $W$.
The {\em central vertices} of an $r$-wall are its two branch vertices that do not belong to any of its layers and that are connected by a path of $W$ that does not intersect any layer. 

\paragraph{Central walls.}
Given an $r$-wall $W$ and an odd $q\in\Nbbb_{\ge3}$ where $q\le r$, we define the \emph{central $q$-subwall}
of $W$, denoted by $W^{(q)}$, to be the $q$-wall obtained from $W$ after removing its first $(r-q)/2$ layers
and all occurring vertices of degree one.

\paragraph{Flat walls.}
Let $G$ be a graph and let $W$ be an $r$-wall of $G$, for some odd integer $r\geq 3$.
We say that a pair $(P,C)\subseteq V(D(W))\times V(D(W))$ is a {\em choice of pegs and corners for $W$} if $W$ is a subdivision of an elementary $r$-wall $\bar{W}$ where $P$ and $C$ are the pegs and the corners of $\bar{W}$, respectively (clearly, $C\subseteq P$).

We say that $W$ is a {\em flat $r$-wall}
of $G$ if there is a separation $(X,Y)$ of $G$ and a choice $(P,C)$
of pegs and corners for $W$ such that:
\begin{itemize}[itemsep=-2pt] 
  \setlength\itemsep{0em}
	\item $V(W)\subseteq Y$,
	\item $P\subseteq X\cap Y\subseteq V(D(W))$, and
	\item if $\Omega$ is the cyclic ordering of the vertices $X\cap Y$ as they appear in $D(W)$,
	      then there exists a rendition $\rho$ of  $(G[Y],\Omega)$.
\end{itemize}

We say that $W$ is a {\em flat wall}
of $G$ if it is a flat $r$-wall for some odd integer $r \geq 3$.

\paragraph{Flatness pairs.}
Given the above, we say that the choice of the 5-tuple $\mathfrak{R}=(X,Y,P,C,\rho)$
\emph{certifies that $W$ is a flat wall of $G$}.
We call the pair $(W,\mathfrak{R})$ a \emph{flatness pair} of $G$ and define
the \emph{height} of the pair $(W,\mathfrak{R})$ to be the height of $W$.

We call the graph $G[Y]$ the \emph{$\mathfrak{R}$-compass} of $W$ in $G$,
denoted by $\compass_\mathfrak{R}(W)$.
We can assume that $\compass_\mathfrak{R} (W)$ is connected, updating $\mathfrak{R}$ by removing from $Y$ the vertices of all the connected components of $\compass_\mathfrak{R} (W)$
except for the one that contains $W$ and including them in $X$ ($\Gamma$ can also be easily modified according to the removal of the aforementioned vertices from $Y$).
We define the \emph{flaps} of the wall $W$ in $\mathfrak{R}$ as
$\flaps_\mathfrak{R}(W):=\{\sigma(c)\mid c\in C(\rho)\}$.

\paragraph{Cell classification.}
Given a cycle $C$ of $\compass_\mathfrak{R}(W)$, we say that
$C$ is \emph{$\mathfrak{R}$-normal} if it is {\sl not} a subgraph of a flap $F\in \flaps_\mathfrak{R}(W)$.
Given an $\mathfrak{R}$-normal cycle $C$ of $\compass_\mathfrak{R}(W)$,
we call a cell $c$ of $\mathfrak{R}$ \emph{$C$-perimetric} if
$\sigma(c)$ contains some edge of $C$.
Notice that if $c$ is $C$-perimetric, then $\pi_\rho(\tilde{c})$ contains two points $p,q\in N(\rho)$
such that $\pi_\rho(p)$ and $\pi_\rho(q)$ are vertices of $C$ where one,
say $P_{c}^{\rm in}$, of the two $(\pi_\rho(p),\pi_\rho(q))$-subpaths of $C$ is a subgraph of $\sigma(c)$ and the other,
denoted by $P_{c}^{\rm out}$, $(\pi_\rho(p),\pi_\rho(q))$-subpath contains at most one internal vertex of $\sigma(c)$,
which should be the (unique) vertex $z$ in $\pi_\rho(\tilde{c})\setminus\{\pi_\rho(p),\pi_\rho(q)\}$.
We pick a $(p,q)$-arc $A_{c}$ in $\hat{c}:={c}\cup\tilde{c}$ such that $\pi_\rho^{-1}(z)\in A_{c}$ if and only if $P_{c}^{\rm in}$ contains
the vertex $z$ as an internal vertex.

We consider the circle $K_{C}=\cupall\{A_{c}\mid \mbox{$c$ is a $C$-perimetric cell of $\mathfrak{R}$}\}$
and we denote by $\Delta_{C}$ the closed disk bounded by $K_{C}$ that is contained in $\Delta$.
A cell $c$ of $\mathfrak{R}$ is called \emph{$C$-internal} if $c\subseteq \Delta_{C}$
and is called \emph{$C$-external} if $\Delta_{C}\cap c=\emptyset$.
Notice that the cells of $\mathfrak{R}$ are partitioned into $C$-internal, $C$-perimetric, and $C$-external cells.

\paragraph{Influence.}
For every $\mathfrak{R}$-normal cycle $C$ of $\compass_\mathfrak{R}(W)$, we define the set
$\influence_\mathfrak{R}(C)=\{\sigma(c)\mid \mbox{$c$ is a cell of $\mathfrak{R}$ that is not $C$-external}\}$.

A wall $W'$ of $\compass_\mathfrak{R}(W)$ is \emph{$\mathfrak{R}$-normal} if $D(W')$ is $\mathfrak{R}$-normal.
Notice that every wall of $W$ (and hence every subwall of $W$) is an $\mathfrak{R}$-normal wall of $\compass_\mathfrak{R}(W)$. We denote by ${\cal S}_\mathfrak{R}(W)$ the set of all $\mathfrak{R}$-normal walls of $\compass_\mathfrak{R}(W)$. Given a wall $W'\in {\cal S}_\mathfrak{R}(W)$ and a cell $c$ of $\mathfrak{R}$,
we say that $c$ is \emph{$W'$-perimetric/internal/external} if $c$ is $D(W')$-perimetric/internal/external, respectively.
We also use $K_{W'}$, $\Delta_{W'}$, $\influence_\mathfrak{R}(W')$ as shortcuts
for $K_{D(W')}$, $\Delta_{D(W')}$, $\influence_\mathfrak{R}(D(W'))$, respectively.

\paragraph{Tilts of flatness pairs.}
Let $(W,\mathfrak{R})$ and $(\tilde{W}',\tilde{\mathfrak{R}}')$ be two flatness pairs of a graph $G$ and let $W'\in {\cal S}_\mathfrak{R}(W)$.
We assume that $\mathfrak{R}=(X,Y,P,C,\rho)$ and $\tilde{\mathfrak{R}}'=(X',Y',P',C',\rho')$.
We say that $(\tilde{W}',\tilde{\mathfrak{R}}')$ is a \emph{$W'$-tilt} of $(W,\mathfrak{R})$ if
\begin{itemize}
	\item $\tilde{\mathfrak{R}}'$ does not have $\tilde{W}'$-external cells,
	\item $\tilde{W}'$ is a tilt of $W'$,
	\item the set of $\tilde{W}'$-internal cells of $\tilde{\mathfrak{R}}'$ is the same as the set of $W'$-internal
	      cells of $\mathfrak{R}$ and their images via $\sigma_{\rho'}$ and $\sigma_\rho$ are also the same,
	\item $\compass_{\tilde{\mathfrak{R}}'}(\tilde{W}')$ is a subgraph of $\cupall\influence_\mathfrak{R}(W')$, and
	\item if $c$ is a cell in $C(\Gamma') \setminus C(\Gamma)$, then $|\tilde{c}| \leq 2$.
\end{itemize}

In particular, one can find a tilt of a subwall of a flat wall in linear time.

\begin{proposition}[Theorem 5, \cite{SauST24amor}]\label{prop_tilt}
There exists an algorithm that, given a graph $G$, a flatness pair $(W,\frak{R})$ of $G$, and a wall $W'\in\Scal_{\frak{R}}(W)$, outputs a $W'$-tilt of $(W,\frak{R})$ in $\Ocal(n+m)$ time.
\end{proposition}

\medskip
The Flat Wall theorem \cite{KawarabayashiTW18anew,RobertsonS95XIII} essentially states that a graph $G$ either contains a big clique as a minor or has bounded treewidth, or contains a flat wall after removing some vertices.
We need here the version of the Flat Wall theorem of \cite[Theorem 8]{SauST24amor}.

\begin{proposition}[Theorem 8, \cite{SauST24amor}]\label{prop_FWth}
There exist two functions $f_{\ref{prop_FWth}},g_{\ref{prop_FWth}}:\Nbbb\to\Nbbb$ and an algorithm that, given a graph $G$, $t\in\Nbbb_{\ge1}$, and an odd $r\in\Nbbb_{\ge3}$, outputs one of the following in time $2^{\Ocal_t(r^2)}\cdot n$:\footnote{We use parameters as indices in the $\Ocal$-notation to indicate that the hidden constant depends (only) of these parameters.} 
\begin{itemize}[itemsep=-2pt] 
  \setlength\itemsep{0em}
\item a report that $K_t$ is a minor of $G$, or
\item a tree decomposition of $G$ of width at most $f_{\ref{prop_FWth}}(t)\cdot r$, or
\item a set $A\subseteq V(G)$ of size at most $g_{\ref{prop_FWth}}(t)$, a flatness pair $(W,\frak{R})$ of $G-A$ of height $r$, and a tree decomposition of the $\frak{R}$-compass of $W$ of width at most $f_{\ref{prop_FWth}}(t)\cdot r$.
\end{itemize}
Moreover, $f_{\ref{prop_FWth}}(t)=2^{\Ocal(t^2\log t)}$ and $g_{\ref{prop_FWth}}(t)=\Ocal(t^{24})$.
\end{proposition}

\paragraph{Apex grid.}
The \emph{apex grid} of height $k$ is the graph $\Gamma_k^+$ obtained by adding a universal vertex to the $(k\times k)$-grid, i.e., a vertex adjacent to every vertex of $\Gamma_{k,k}$.

\begin{proposition}[Lemma 3.1, \cite{DemaineFHT04bidimen}]\label{prop_apex}
Let $m,k\in\Nbbb$ with $m\ge k^2+2k$ and let $H$ be the $(m\times m)$-grid. 
Let $X$ be a subset of at least $k^4$ vertices in the central $((m-2k)\times (m-2k))$-subgrid of $H$.
Then there is a model of the $(k\times k)$-grid in $H$ in which every branch set intersects $X$.
\end{proposition}

\smallskip
The following result is a version of the Flat Wall theorem that is already somewhat known (see for instance in the proof of \cite[Lemma 4.7]{KorhonenPS24mino}).
We write here an algorithmic version that has yet to be stated, to the authors' knowledge.

\begin{theorem}\label{prop_flatwallth}
There exist a function $f_{\ref{prop_flatwallth}}:\Nbbb\to\Nbbb$ and an algorithm that, 
given a graph $G$ and $k,r\in\Nbbb$ with $r$ odd, 
outputs one of the following in time $2^{\Ocal_k(r^2)}\cdot(n+m)$: 
\begin{itemize}[itemsep=-2pt] 
  \setlength\itemsep{0em}
\item a report that $G$ contains an apex grid of height $k$ as a minor,
\item a report that $\tw(G)\le f_{\ref{prop_flatwallth}}(k)\cdot r$, or
\item a flatness pair $(W,\frak{R})$ of $G$ of height $r$ whose $\frak{R}$-compass has treewidth at most $f_{\ref{prop_flatwallth}}(k)\cdot r$.
\end{itemize}
Moreover, $f_{\ref{prop_flatwallth}}(k)=2^{\Ocal(k^4\log k)}$.
\end{theorem}

\begin{proof}
Let  $t:=k^2+1$, $a:=g_{\ref{prop_FWth}}(t)$, $d:=k^4$, $s:=(d-1)\cdot(a-1)+1$, $r_2=\odd(\lceil \sqrt{s}\cdot (r+2)\rceil)$, and $r_1:=r_2+2k$.

We apply the algorithm of \autoref{prop_FWth} that, in time $2^{\Ocal_t(r_1^2)}\cdot n$, 
either reports that $K_t$ is a minor of $G$, 
or finds a tree decomposition of $G$ of width at most  $f_{\ref{prop_FWth}}(t)\cdot r_1$, 
or finds a set $A\subseteq V(G)$ of size at most $a$, a flatness pair $(W_1,\frak{R}_1)$ of $G-A$ of height $r_1$, and a tree decomposition of the $\frak{R}_1$-compass of $W_1$ of width at most $f_{\ref{prop_FWth}}(t)\cdot r_1$.
In the first case, $G$ thus contains an apex grid of order $k$ as a minor, so we conclude.
In the second case, we also immediately conclude.
We can thus assume that we found a flatness pair $(W_1,\frak{R}_1)$ of $G-A$.
Let $W_2$ be the central $r_2$-subwall of $W_1$.

Given that $r_2\ge\lceil \sqrt{s}\cdot (r+2)\rceil$, we can find a collection $\Wcal'=\{W_1',\dots,W_s'\}$ of $r$-subwalls of $W_2$ such that $\influence_{\frak{R}_1}(W_i')$ and $\influence_{\frak{R}_1}(W_j')$ are disjoint for distinct $i,j\in[s']$.
Then, by applying the algorithm of \autoref{prop_tilt}, in time $\Ocal(n+m)$, we find a collection $\Wcal=\{(W_1,\frak{R}^1),\dots,(W_s,\frak{R}^s)\}$ such that, for $i\in[s]$, $(W_i,\frak{R}^i)$ is a $W_i'$-tilt of $(W_1,\frak{R}_1)$, and the $\frak{R}^i$-compasses of the $W_i$s are pairwise disjoint and have treewidth at most $f_{\ref{prop_FWth}}(t)\cdot r_1$.

For each vertex $v\in A$, we check whether $v$ is adjacent to vertices of the compass of $d$ subwalls in $\Wcal$. 
If that is the case for some $v\in A$, then observe that $G$ contains as a minor an $(r_1\times r_1)$-grid (obtained by contracting the intersection of horizontal and vertical paths of $W_1$) along with a vertex (corresponding to $v$) that is adjacent to $d$ vertices of its central $(r_2\times r_2)$-subgrid (corresponding to $W_2$).
But then, by \autoref{prop_apex}, $G$ contains an apex grid of height $k$ as a minor, so we once again conclude.

We can thus assume that every vertex in $A$ is adjacent to vertices of the compass of at most $d-1$ subwalls in $\Wcal$. 
Given that $|\Wcal|=s$ and that $|A|\le a$, it implies that there is at least one wall $W_i$ in $\Wcal$ whose $\frak{R}^i$-compass is adjacent to no vertex in $A$.
Hence the result.
\end{proof}

\subsection{An obstruction to \Hkpl}\label{subsec_obs}
In this section, we show that apex grids are obstructions to the existence of planar $\Hcal^{(k)}$-modulators.

\begin{lemma}\label{lem:obstructions}
Let $\Hcal$ be an arbitrary graph class and let $k$ be a positive integer. Then any graph $G$ containing the apex grid $\Gamma_{k'}^+$ for $k'\geq \sqrt{k+4}+2$ as a minor does not admit a planar $\Hcal^{(k)}$-modulator. 
\end{lemma}

\begin{proof}
Let $\Gcal$ be the class of all graphs. Because any graph class $\Hcal\subseteq\Gcal$, it is sufficient to show the lemma for $\Hcal=\Gcal$. For this, we prove the claim for $G=\Gamma_{k'}^+$.

\begin{claim}\label{cl:apex}
The apex grid $\Gamma_{k'}^+$ does not admit a planar $\Gcal^{(k)}$-modulator.
\end{claim}

\begin{cproof}
The proof is by contradiction. Assume that $X\subseteq V(\Gamma_{k'}^+)$ is a planar $\Gcal^{(k)}$-modulator.
Because $|V(\Gamma_{k'}^+)|>k+5$, it implies that $X\ne\emptyset$, and because $\Gamma_{k'}^+$ is not planar as $k'\geq 3$, it implies that $V(\Gamma_{k'}^+)\setminus X\neq\emptyset$.
Let $v$ be the apex of $\Gamma_{k'}^+$. 
Denote by $B$ the vertices of $\Gamma_{k'}$ of degree at most four,
and set $S:=V(\Gamma_{k'}^+)\setminus N_{\Gamma_{k'}^+}(B)$, that is $S$ is the set of vertices of $\Gamma_{k'}^+-v$ that do not belong to the two outermost cycles of the grid.
It is straightforward to verify that for any two distinct nonadjacent vertices $x,y\in S$, $\Gamma_{k'}^+$ has five internally vertex disjoint $x$-$y$-paths: 
four paths in $\Gamma_{k'}^+-v$ and one path with the middle vertex $v$. 
Therefore, for any separation $(L,R)$ of $\Gamma_{k'}^+$ of order at most four with $L\setminus R\neq\emptyset$ and $R\setminus L\neq\emptyset$, either $S\subseteq L$ or $S\subseteq R$.
 Furthermore, because $v$ is universal, $v\in L\cap R$. 
For each connected component $C$ of $\Gamma_{k'}^+-X$, $(N_{\Gamma_{k'}^+}[V(C)],V(G)\setminus V(C))$ is a separation of order at most four.
This implies that either $S\subseteq N_{\Gamma_{k'}^+}[V(C)]$ for a connected component $C$ of $\Gamma_{k'}^+-X$ or $S\subseteq X$.
However, because $|S|\geq k+4$ and $v\in N_{\Gamma_{k'}^+}(V(C))$, in the first case, we would have that $|V(C)|>k$ contradicting that each connected component of $\Gamma_{k'}^+-X$ has at most $k$ vertices. 
Thus, $S\subseteq X$.  
We also have that $v\in X$. 
Then because $k'\geq 5$, $\Gamma_{k'}^+[S]$ contains $\Gamma_{3,3}$ as a subgraph and, therefore, $\Gamma_{k'}^+[S\cup\{v\}]$ is not planar. 
This contradicts that the torso of $X$ is planar and proves the claim. 
\end{cproof}

Given that $\Gcal^{(k)}$ is a minor-closed graph class, so is the class of $\Gcal^{(k)}$-planar graphs. 
Therefore, for any graph $G$ containing $\Gamma_{k'}^+$ as a minor, $G$ is not a $\Gcal^{(k)}$-planar graph by \autoref{cl:apex}. This completes the proof.
\end{proof}

\subsection{$\Hcal$-compatible \sdecomps}\label{subsec_compatible}

In this subsection, we observe that the problem of \Hpl\ has an equivalent definition using \sdecomps.

\paragraph{$\Hcal$-compatible \sdecomps.}
Let $\Hcal$ be a graph class.
Let also $G$ be a graph and $\delta=(\Gamma,\Dcal)$ be a \sdecomp\ of $G$.
We say that a cell $c$ of $\delta$ is \emph{$\Hcal$-compatible} if 
there is a set $S_c\subseteq V(\sigma(c))$ containing $\pi_\delta(\tilde{c})$ such that $\torso(\sigma(c),S_c)$ has a planar embedding with the vertices of $\pi_\delta(\tilde{c})$ on the outer face and such that, for each $D\in\cc(\sigma(c)-S_c)$, $D\in\Hcal$.
We say that $\delta$ is \emph{$\Hcal$-compatible} if every cell of $\delta$ is $\Hcal$-compatible.
See \autoref{fig_compatible} for an illustration.

We show the following lemma.

\begin{lemma}\label{obs_sol_compatible}
Let $\Hcal$ be a graph class, $k\in\Nbbb$, and $G$ be a graph.
Then $G$ is $\Hcal^{(k)}$-planar if and only if $G$ has an $\Hcal^{(k)}$-compatible \sdecomp\ $\delta$.
Additionally, for any $r$-wall $W$ of $G$ with $r\ge\max\{\sqrt{(k+7)/2}+2,7\}$, we can choose $\delta$ such that the $(r-2)$-central wall $W'$ of $W$ is grounded in $\delta$.
\end{lemma}

As a side note, the following proof can be easily adapted to prove that $G$ is $\Hcal$-planar if and only if $G$ has an $\Hcal$-compatible \sdecomp\ $\delta$. However, in this case, the wall $W$ may be completely contained in a cell of $\delta$.

\begin{proof}
\emph{Suppose that $G$ has an $\Hcal^{(k)}$-compatible \sdecomp\ $\delta$.}
Then, for each cell $c\in C(\delta)$, there is a set $S_c\subseteq V(\sigma(c))$ containing $\pi_\delta(\tilde{c})$ such that $\torso(\sigma(c),S_c)$ has a planar embedding with the vertices of $\pi_\delta(\tilde{c})$ on the outer face and such that, for each $D\in\cc(\sigma(c)-S_c)$, $D\in\Hcal^{(k)}$.
Then we immediately get that $\bigcup_{c\in C(\delta)}S_c$ is a planar $\Hcal^{(k)}$-modulator of $G$.
This comes from the fact that the disks in the sphere decomposition are disjoint apart from shared boundary vertices.
Therefore, we can take the planar embeddings of the torsos of the individual cells that have the size-3 boundaries on the outer face, mirror them as needed to get the ordering along the boundary to match the ordering in the sphere decomposition, and then to use the cell-torso drawings into the sphere decomposition to get a complete drawing of the entire torso on the sphere, which is a planar drawing. 

{\em Suppose now that $G$ is $\Hcal^{(k)}$-planar.}
Let $S$ be a planar $\Hcal^{(k)}$-modulator in $G$.
Let $V'\subseteq V(W')$ be the branch vertices of $W$ that are vertices of $W'$.
Then $|V'|=2(r-2)^2-2\ge k+5$.
We first make the following observation.
\begin{claim}\label{obs_smallsepwall}
For any separation $(A,B)$ of order at most three in $G$, 
the graph induced by one of $A$ and $B$, say $B$, contains no cycle of $W'$, and $B\setminus A$ contains at most one vertex of $V'$.
\end{claim}

Let $\delta=(\Gamma,\Dcal)$ be a sphere embedding of $\torso(G,S)$.
Note that, for a sphere embedding, the set of its nodes contains all vertices of the drawing, since each disk of the sphere embedding surrounds a single edge and therefore no vertex of the drawing lies in the interior of any disk. 
For each $\delta$-aligned disk $\Delta$, let 
$Z_\Delta:=\{C\in\cc(G-S)\mid N_G(V(C))\subseteq \pi_\delta(N(\delta)\cap \Delta)\}$.
Let also $V(Z_\Delta):=\bigcup_{C\in Z_\Delta}V(C)$ be the set of vertices of components in $Z_\Delta$.

\begin{claim}\label{cl_compat}
Let $D\in\cc(G-S)$.
There is a $\delta$-aligned disk $\Delta_D$ such that 
\begin{itemize}[itemsep=-2pt] 
  \setlength\itemsep{0em}
 \item the vertices of $N_G(V(D))$ are in the disk $\Delta_D$, i.e. $N_G(V(D))\subseteq \pi_\delta(N(\delta)\cap \Delta_D)$, with all but at most one (in the case  $|N_G(V(D))|\le 4$) being exactly the vertices of the boundary of $\Delta_D$, i.e. there is a set $X_D\subseteq N_G(V(D))$ of size $\min\{|N_G(V(D))|,3\}$ such that $X_D=\pi_\delta(\bd(\Delta_D)\cap N(\delta))$, and
\item the graph induced by $B_D:=V(\inG_\delta(\Delta_D))\cup V(Z_\Delta)$ contains no cycle of $W'$.
\end{itemize}
\end{claim}

\begin{cproof}
Let $Y_D:=\{C\in\cc(G-S)\mid N_G(V(C)) \subseteq N_G(V(D))\}$ 
and $V(Y_D):=\bigcup_{C\in Y_D}V(C)$ be the set of vertices of components in $Y_D$. We consider three cases depending on the size of $N_G(V(D))$.

\paragraph{Case 1:} $|N_G(V(D))|\le 2$. 
We set $\Delta_D$ to be
\begin{itemize}[itemsep=-2pt] 
  \setlength\itemsep{0em}
\item the empty disk if $N_G(V(D))=\emptyset$
\item $\pi_\delta^{-1}(v)$ if $N_G(V(D))=\{v\}$
\item the closure of the cell $c\in C(\delta)$ such that $\sigma(c)$ is the edge induced by $u,v$ in $\torso(G,S)$ if $N_G(V(D))=\{u,v\}$.
\end{itemize}

\paragraph{Case 2:} $|N_G(V(D))|= 3$.
Informally, the triangle induced by $N_G(V(D))$ in $\torso(G,S)$ defines two disks $\Delta_1$ and $\Delta_2$ in $\delta$. We need to show that one of them is the desired disk, that does not contain a cycle of $W'$.

Let $T$ be the cycle in the embedding $\delta$ induced by $N_G(V(D))$.
$\mathbb{S}^2\setminus T$ is the union of two open disks whose closure is respectively called $\Delta_1$ and $\Delta_2$.
For $i\in[2]$, let $A_i:=V(\inG_\delta(\Delta_i))\cup V(Z_{\Delta_i})\setminus V(Y_D)$.
Then $(A_1,A_2)$ is a separation of $G-V(Y_D)$ with $A_1\cap A_2=N_G(V(D))$.
By \autoref{obs_smallsepwall}, the graph induced by one side of the separation, 
say $A_2$ contains no cycle of $W'$.
We set $\Delta_D$ to be a $\delta$-aligned disk containing $\Delta_1$.

It remains to prove that $G[B_D]$ contains no cycle of $W'$, where $B_D:=V(\inG_\delta(\Delta_D))\cup V(Z_\Delta)=A_2\cup V(Y_D)$.
$(A_1,B_D)$ is a separation of $G$ with $A_1\cap B_D=N_G(V(D))$, so by \autoref{obs_smallsepwall}, one of $A_1$ and $B_D$ induce a graph containing no cycle of $W'$.
Assume toward a contradiction that $G[A_1]$ contains no cycle of $W'$.
Note that, given that $r-2\ge 5$, $W'$ contains a set $\Ccal$ of pairwise disjoint cycles with $|\Ccal|\ge 4$.
Given that $|A_1\cap A_2|=|N_G(V(D))|=3$, at most three cycles of $\Ccal$ intersects $N_G(V(D))$ and $G[A_1\cup A_2]$ contains at most one cycle of $\Ccal$.
Therefore, $Y_D$ contains at least one cycle of $\Ccal$.
Let $C\in Y_D$ by a component containing such a cycle.
We have that $(V(C)\cup N_G(V(C)),V(G)\setminus V(C))$ is a separation of order at most three so, given that $C$ contains a cycle of $W'$, by \autoref{obs_smallsepwall}, $C$ contains at least $|V'|-4>k$ vertices of $V'$.
This contradicts the fact that $|V(C)|\ge k$.

\paragraph{Case 3:} $|N_G(V(D))|=4$.
Let $\{v_i\mid i\in[4]\}$ be the vertices in $N_G(V(D))$ and $X_i:=N_G(V(D))\setminus\{v_i\}$.
For $i\in[4]$, we define $\Delta_i$ similarly to $\Delta_D$ in the previous case. 
Then $X_i=\pi_\delta(\bd(\Delta_D)\cap N(\delta))$ and
the graph induced by $B_i:=V(\inG_\delta(\Delta_i))\cup V(Z_{\Delta_i})$ contains no cycle of $W'$.
However, it might be the case that $v_i\notin \pi_\delta(N(\delta)\cap \Delta_i)$.

If $v_i\notin \pi_\delta(N(\delta)\cap \Delta_i)$ for all $i\in[4]$, then the interior of the $\Delta_i$ are pairwise disjoint.
Moreover, $\bigcup_{i\in[4]}B_i=V(G)$, $\bigcap_{i\in[4]}B_i=V(Y_D)\cup N_G(V(D))$, and $N_G(V(Y_D))\subseteq N_G(V(D))$.
As shown in the previous case, $G[B_i]$ contains no cycle of $W'$ for $i\in[4]$.
This implies that, for any cycle of $W'$, there are distinct $i,j\in[4]$ such that the cycle has a vertex in $B_i\setminus B_j$ and $B_j\setminus B_i$, and thus intersects $N_G(V(D))$ twice.
However, $W'$ has at least three pairwise disjoint cycles, a contradiction to the fact that $|N_G(V(D))|=4$.
Therefore, there is $i\in[4]$ such that $v_i\in \pi_\delta(N(\delta)\cap \Delta_i)$.
We then set $\Delta_D:=\Delta_i$. This completes the case analysis and the proof of the claim.
\end{cproof}

Note that, if $N_G(V(D))=N_G(V(D'))$, then we can assume that $\Delta_D=\Delta_{D'}$.
Let $\Dcal^*$ be the inclusion-wise maximal elements of $\Dcal\cup\{\Delta_D\mid D\in\cc(G-S)\}$.
By maximality of $\Dcal^*$ and planarity of $\torso(G,S)$, any two distinct $\Delta_D,\Delta_{D'}\in\Dcal^*$ may only intersect on their boundary.
For each $C\in\cc(G-S)$, we draw $C$ in a $\Delta_D\in\Dcal$ such that $V(C)\subseteq B_D$, and add the appropriate edges with $\pi_\delta(N(\delta)\cap \bd(\Delta_D))$.
We similarly draw the edges of $G[S]$ to obtain a drawing $\Gamma^*$ of $G$.

For each cell $c$ of $\delta^*=(\Gamma^*,\Dcal^*)$, there is $D\in\cc(G-S)$ such that $\sigma_{\delta^*}(c)$ contains no cycle of $W'$, since $\sigma_{\delta^*}(c)$ is a subgraph of $G[B_D]$, so $W'$ is grounded in $\delta^*$.
Moreover, $S_c:=V(\sigma(c))\cap V(\torso(G,S))$ certifies that $c$ is $\Hcal^{(k)}$-compatible.
\end{proof}

Moreover, if $\Hcal$ is a hereditary graph class, then we can ``ground'' an $\Hcal$-compatible \sdecomp\ as much as possible.
This is what we prove in \autoref{lem_compatible} after defining the relevant definitions.

\paragraph{Containment of cells.}
Let $\delta=(\Gamma,\Dcal)$ and $\delta'=(\Gamma',\Dcal')$ be two \sdecomps\ of $G$.
Let $c\in C(\delta)$ and $c'\in C(\delta')$ be two cells.
We say that \emph{$c$ is contained in $c'$} if $V(\sigma(c))\subseteq V(\sigma_{\delta'}(c'))$. 
We say that \emph{$c$ and $c'$ are equivalent} if $c$ is contained in $c'$ and $c'$ is contained in $c$.

\paragraph{Ground-maximal \sdecomps.}
Let $\delta=(\Gamma,\Dcal)$ and $\delta'=(\Gamma',\Dcal')$ be two \sdecomps\ of $G$.
We say that $\delta$ is \emph{more grounded} than $\delta'$ (and that $\delta'$ is \emph{less grounded} than $\delta$) 
if each cell $c\in C(\delta)$ is contained in a cell $c'\in C(\delta')$, and in case $c$ and $c'$ are equivalent, if $\pi_{\delta'}(\tilde{c}')\subseteq\pi_\delta(\tilde{c})$.
We say that $\delta$ is \emph{ground-maximal} 
if no other \sdecomp\ of $G$ is more grounded 
than $\delta$.
We say that a cell $c\in C(\delta)$ is \emph{ground-maximal} if, for any \sdecomp\ $\delta'=(\Gamma',\Dcal')$ that is more grounded than $\delta$ and for any cell $c'\in C(\delta')$ that is contained in $c$, $V(\sigma(c))=V(\sigma_{\delta'}(c'))$ and $\pi_{\delta'}(\tilde{c}')=\pi_\delta(\tilde{c})$.

\begin{lemma}\label{lem_compatible}
Let $\Hcal$ be a hereditary graph class and $G$ be a graph.
Let $\delta$ be a \sdecomp\ of $G$ that is $\Hcal$-compatible.
Then any \sdecomp\ of $G$ that is more grounded than $\delta$ is also $\Hcal$-compatible.
\end{lemma}

\begin{proof}
Let $\delta'=(\Gamma',\Dcal')$ be a \sdecomp\ that is more grounded than $\delta=(\Gamma,\Dcal)$.
Let $c'\in C(\delta')$. Let us show that $c'$ is $\Hcal$-compatible.
Given that $\delta'$ is more grounded than $\delta$, there is a cell $c\in C(\delta)$ such that $V(\sigma_{\delta'}(c'))\subseteq V(\sigma(c))$.

Let $U=V(\sigma(c))\setminus V(\sigma_{\delta'}(c'))$.
Given that $c$ is $\Hcal$-compatible, there is a set $S_c\subseteq V(\sigma(c))$ containing $\pi_{\delta}(\tilde{c})$ such that $\torso(\sigma(c),S_c)$ has a planar embedding with the vertices of $\pi_{\delta}(\tilde{c})$ on the outer face and such that, for each $D\in\cc(\sigma(c)-S_c)$, $D\in\Hcal$.
Let $S_{c'}:=S_c\setminus U$.
By heredity of $\Hcal$, each connected component $C$ of $\sigma(c)-S_c-U=\sigma_{\delta'}(c')-S_{c'}$ belong to $\Hcal$.
Additionally, we have $\torso(\sigma_{\delta'}(c'),S_{c'})=\torso(\sigma(c)-U,S_{c}\setminus U)=\torso(\sigma(c),S_{C})-U$, so $\torso(\sigma_{\delta'}(c'),S_{c'})$ is planar.
The result follows.
\end{proof}

Combining \autoref{obs_sol_compatible} and \autoref{lem_compatible}, we obtain the following result.

\begin{corollary}\label{cor_compatible}
Let $\Hcal$ be a hereditary graph class, $k\in\Nbbb$, $G$ be a graph, and $W$ be an $r$-wall in $G$ with $r\ge\max\{\sqrt{(k+7)/2}+2,7\}$.
A graph $G$ is $\Hcal^{(k)}$-planar if and only if $G$ has a ground-maximal $\Hcal^{(k)}$-compatible \sdecomp\ $\delta$.
Additionally, we can choose $\delta$ such that the $(r-2)$-central wall of $W$ is grounded in $\delta$.
\end{corollary}

\subsection{Comparing \sdecomps}\label{subsec_compare}

In this subsection, we essentially want to prove that, given two \sdecomps\ $\delta_1$ and $\delta_2$ on the same graph, if $\delta_1$ is \emph{ground-maximal} and $\delta_2$ is \emph{well-linked}, then $\delta_1$ is always \emph{more grounded} than $\delta_2$ (see \autoref{cor_well-ground}).

\paragraph{Well-linkedness.}
Let $\rho=(\Gamma,\Dcal)$ be a rendition of $(G,\Omega)$.
We say that a cell $c\in C(\delta)$ is \emph{well-linked} if there are $|\tilde{c}|$ vertex-disjoint paths from $\pi_\delta(\tilde{c})$ to $V(\Omega)$.
We say that $\delta$ is \emph{well-linked} if every cell $c\in C(\delta)$ is well-linked.

The following result is a corollary of \cite{SauST24amor}.

\begin{proposition}[Lemma 3, \cite{SauST24amor}]\label{lem_rend_to_min_or_max}
If a society $(G,\Omega)$ has a rendition, then $(G,\Omega)$ has a well-linked rendition.
\end{proposition}

\paragraph{Intersection and crossing of cells.}
Let $\delta=(\Gamma,\Dcal)$ and $\delta'=(\Gamma',\Dcal')$ be two \sdecomps\ of a graph $G$.
Let $c\in C(\delta)$ and $c'\in C(\delta')$ be two cells.
We say that \emph{$c$ and $c'$ intersect} if $(V(\sigma(c))\cap V(\sigma_{\delta'}(c')))\setminus(\pi_\delta(\tilde{c})\cap\pi_{\delta'}(\tilde{c}'))\neq\emptyset$.
We say that \emph{$c$ and $c'$ cross} if $c$ and $c'$ intersect but that neither of them is contained in the other.

\smallskip
Our goal in this subsection is to prove that if $c$ is a ground-maximal cell in some \sdecomp\ $\delta$ and that $c'$ is some well-linked cell in some rendition $\delta'$, such that $c$ and $c'$ intersect, then $c$ is contained in $c'$ (see \autoref{cell_in_cell2}).
We first define a splitting operation on a \sdecomp\ that will be extensively used in this subsection.

\paragraph{Splitting a cell.}
Let $\delta=(\Gamma,\Dcal)$ and $\delta'=(\Gamma',\Dcal')$ be two \sdecomps\ of a graph $G$.
Let $c\in C(\delta)$ be a cell and $v\in V(\sigma(c))$ be a vertex such that at least two connected components of $\sigma(c)-v$ contains a vertex of $\pi_{\delta}(\tilde{c})$.
Such a vertex $v$ is called a \emph{cut-vertex of $c$ in $\delta$}.
We say that $\delta^*$ is \emph{obtained from $\delta$ by splitting $c$ at $v$} if $\delta'$ can be constructed from $\delta-(V(\sigma(c))\setminus\pi_{\delta(\tilde{c})})=(\Gamma',\Dcal')$ as follows.

Let $C_1$ be a disjoint union of components of $\sigma(c)-v$ such that $C_1$ and $C_2:=\sigma(c)-v-V(C_1)$ are both non-empty.
Let $A_1=V(C_1)\cup\{v\}$ and $A_2=V(\sigma(c))\setminus V(C_1)$.
For $i\in[2]$, let $B_i=A_i\cap\pi_{\delta}(\tilde{c})$. Note that $1\le|B_i|\le 2$.

We set $N(\delta^*)=N(\delta')\cup\{x\}$ for some arbitrary point $x$ contained in $c$.
We set $\Dcal^*=\Dcal'\cup\{\Delta_{c_1},\Delta_{c_2}\}$, where $c_1,c_2\subseteq c$ be two new cells such that $\tilde{c}_1=\{x\}\cup\pi_{\delta}^{-1}(B_1)$ and $\tilde{c}_2=\{x\}\cup\pi_{\delta}^{-1}(B_2)$.
Then $\Gamma^*$ is obtained from $\Gamma'$ by arbitrarily drawing $G[A_i]$ in $c_i$, for $i\in[2]$.
See \autoref{fig_split} for an illustration.

\begin{figure}[h]
\centering
\includegraphics[scale=0.8]{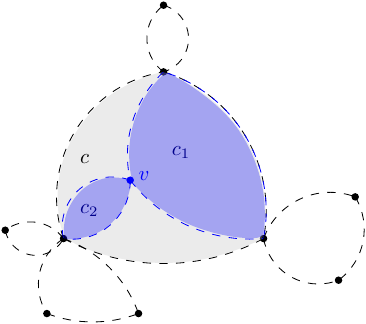}
\caption{The blue \sdecomp\ is obtained from the gray one by splitting $c$ at $v$.}
\label{fig_split}
\end{figure}

Observe that $\delta'$ is well-defined and that it is more grounded than $\delta$.
Thus, $c$ is not ground-maximal. 

\smallskip
In the following lemma, we prove that if two cells $c$ and $c'$ cross and that $c'$ has exactly one vertex of its boundary in $c$, then $c$ is not ground-maximal.

\begin{lemma}\label{lem_cut_vtx}
Let $\delta=(\Gamma,\Dcal)$ and $\delta'=(\Gamma',\Dcal')$ be two \sdecomps\ of a graph $G$.
Let $c\in C(\delta)$ and $c'\in C(\delta')$ be two cells that cross and suppose that $v\in V(G)$ is the unique vertex of $\pi_{\delta'}(\tilde{c}')$ in $\sigma(c)- \pi_{\delta}(\tilde{c})$.
Then $v$ is a cut-vertex of $c$ in $\delta$, and therefore, $c$ is not ground-maximal.
\end{lemma}

\begin{proof}
Let $A=\pi_{\delta}(\tilde{c})\setminus V(\sigma_{\delta'}(c'))\ne\emptyset$ and $A'=\pi_{\delta}(\tilde{c})\cap V(\sigma_{\delta'}(c'))\ne\emptyset$.
Let $P$ be an $(A-A')$-path in $V(\sigma(c))$ (it exists by the connectivity of $\sigma(c)$).
Given that one endpoint of $P$ is in $V(\sigma_{\delta'}(c'))$ and that the other is not, it implies that $P$ intersects $\pi_{\delta'}(\tilde{c}')$.
Given that $v$ is the unique vertex of $\pi_{\delta'}(\tilde{c}')$ in $V(\sigma(c))$, it implies that $v\in V(P)$ for all $(A-A')$-paths $P$.
We conclude that at least two connected components of $\sigma(c)-v$ ocontains a vertex of $\pi_{\delta}(\tilde{c})$, and thus, $v$ is a cut-vertex of $c$.
\end{proof}

We now prove that a ground-maximal cell and a well-linked cell cannot cross. 

\begin{lemma}\label{cell_in_cell}
Let $(G,\Omega)$ be a society.
Let $\delta$ be a \sdecomp\ of $G$ and let $\delta'$ be a rendition  of $(G,\Omega)$.
Let $c\in C(\delta)$ be a ground-maximal cell such that $V(\Omega)\cap (V(\sigma(c))\setminus\pi_{\delta}(\tilde{c}))=\emptyset$ and 
$c'\in C(\delta')$ be a well-linked cell.
Then $c$ and $c'$ do not cross.
\end{lemma}

\begin{proof}
Assume toward a contradiction that $c$ and $c'$ cross.
Therefore, by connectivity of $\sigma(c)$ and $\sigma_{\delta'}(c')$, there is at least one vertex of $\pi_{\delta}(\tilde{c})$ in $V(\sigma_{\delta'}(c'))\setminus \pi_{\delta'}(\tilde{c}')$ and at least one in $V(\sigma(c))\setminus V(\sigma_{\delta'}(c'))$.
Similarly, there is at least one vertex of $\pi_{\delta'}(\tilde{c}')$ in $V(\sigma(c))\setminus \pi_{\delta}(\tilde{c})$ and at least one in $V(\sigma_{\delta'}(c'))\setminus V(\sigma(c))$.
Given that $|\tilde{c}|\le3$ and $|\tilde{c}'|\le3$, we can distinguish two cases:
\begin{itemize}
\item {\bf Case 1:} there is exactly one vertex of $\pi_{\delta'}(\tilde{c}')$ in $\sigma(c)- \pi_{\delta}(\tilde{c})$ (\autoref{fig_case1}).
\item {\bf Case 2:} there are exactly two vertices of $\pi_{\delta'}(\tilde{c}')$ in $\sigma(c)- \pi_{\delta}(\tilde{c})$ (\autoref{fig_case2}).
\end{itemize}

\emph{Assume first that Case 1 happens.} 
\begin{figure}[h]
\centering
\includegraphics[scale=0.8]{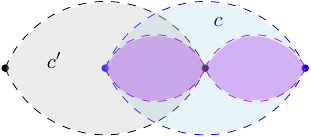}
\caption{Case 1 of \autoref{cell_in_cell}.}
\label{fig_case1}
\end{figure}
Let $v$ be the unique vertex of $\pi_{\delta'}(\tilde{c}')$ in $\sigma(c)- \pi_{\delta}(\tilde{c})$.
By \autoref{lem_cut_vtx}, $v$ is a cut-vertex of $c$ in $\delta$, and therefore $c$ is not ground-maximal.
Hence, Case 1 does not apply by the maximality of $c$.

\emph{Assume that Case 2 applies.}
\begin{figure}[h]
\centering
\includegraphics[scale=1]{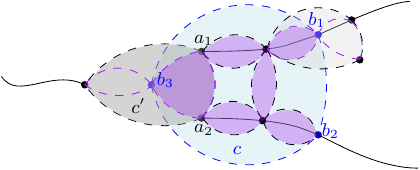}
\caption{Case 2 of \autoref{cell_in_cell}.}
\label{fig_case2}
\end{figure}
Let $a_1$ and $a_2$ be the vertices of $\pi_{\delta'}(\tilde{c}')$ in $\sigma(c)- \pi_{\delta}(\tilde{c})$.
Let $B=\pi_{\delta}(\tilde{c})\setminus V(\sigma_{\delta'}(c'))$ and $B'=\pi_{\delta}(\tilde{c})\cap V(\sigma_{\delta'}(c'))$.
Remember that $1\le|B|,|B'|\le2$.
Given that $c'$ is well-linked, there are two vertex-disjoint paths $P_1$ and $P_2$ in $G-V(\sigma(c'))$, the first one from $a_1$ to $V(\Omega)$ and the other from $a_2$ to $V(\Omega)$.
Given that $a_1,a_2\in V(\sigma(c))$ and that $V(\Omega)\cap (V(\sigma(c))\setminus \pi_{\delta}(\tilde{c}))=\emptyset$, it implies that $P_1$ and $P_2$ intersect $B$, say in $b_1$ and $b_2$ respectively.
Given that $B=\{b_1,b_2\}$, we thus have $|B'|=1$. By \autoref{lem_cut_vtx}, the unique vertex $b_3$ in $B'$ is a cut-vertex for $c'$ in $\delta'$.
Given that every edge of $G$ is contained in a cell of $\delta'$, for $i\in[2]$, there is a cell $c_i\in C(\delta')$ containing the edge $b_i'b_i$ where $b_i'$ is the neighbor of $b_i$ in $V(P_i)\cap V(\sigma(c))$.

Observe that $c_1$ and $c_2$ are distinct cells since otherwise we would have 
$|\pi_{\delta'}(\tilde{c_1})|\ge|\pi_{\delta'}(\tilde{c_1})\cap V(P_1)|+|\pi_{\delta'}(\tilde{c_1})\cap V(P_2)|\ge 4$.
Therefore, for $i\in[2]$,
either $V(\sigma_{\delta'}(c_i))\subseteq V(\sigma(c))$, or by \autoref{lem_cut_vtx}, $b_i$ is a cut-vertex for $c_i$ in $\delta'$.
Let $\delta''$ be the rendition obtained from $\delta'$ by splitting $c'$ at $b_3$, $c_1$ at $b_1$, and $c_2$ at $b_2$ (only if those are cut-vertices for the last two cases).
Note that $\delta''$ is more grounded that $\delta'$ and that no cell $c''\in C(\delta'')$ crosses $c$.
Let $\Ccal:=\{c''\in C(\delta'')\mid c'' \text{ contained in } c\}$ and $\Delta''$ be a $\delta''$-aligned disk containing exactly the cells of $\Ccal$.
Up to homeomorphism, we may assume that $\Delta''=\Delta_c$ and that $\pi_{\delta''}(v)=\pi_{\delta}(v)$ for all $v\in \bd(\Delta'')\cap N(\delta'')=\bd(\Delta_c)\cap N(\delta)$.
Then we define $\delta^*$ to be the \sdecomp\ of $G$
that is equal to $\delta''$ when restricted to $c$, and equal to $\delta$ otherwise.
$\delta^*$ is more grounded than $\delta$, and, in particular, $c$ is not ground-maximal.
Hence, Case 2 does not apply by maximality of $c$.
This contradiction concludes the proof.
\end{proof}

\begin{lemma}\label{cell_in_cell2}
Let $(G,\Omega)$ be a graph.
Let $\delta=(\Gamma,\Dcal)$ be a \sdecomp\ of $G$ and let $\delta'=(\Gamma',\Dcal')$ be a rendition  of $(G,\Omega)$.
Let $c\in C(\delta)$ be a ground-maximal cell such that $V(\Omega)\cap (V(\sigma(c))\setminus\pi_{\delta}(\tilde{c}))=\emptyset$.
Suppose that every cell 
$c'\in C(\delta')$ that intersects $c$ is well-linked.
Then $c$ is contained in a cell $c'\in C(\delta')$.
\end{lemma}

\begin{proof}
By \autoref{cell_in_cell}, no well-linked cell $c'\in C(\delta')$ crosses $c$.
Hence, every cell $c'\in C(\delta')$ that intersects $c$ is either contained in $c$ or contains $c$.
If $c$ is contained in $c'$, we can immediately conclude, so let us assume that every cell $c'\in C(\delta')$ that intersects $c$ is contained in $c$.
We define $\delta^*=(\Gamma^*,\Dcal^*)$ to be the rendition of $(G,\Omega)$ that is equal to $\delta'$ when restricted to $c$, and equal to $\delta$ otherwise, similarly to $\delta^*$ in the Case 2 of \autoref{cell_in_cell}
Thus, $\delta^*$ is more grounded than $\delta$.
Hence, given that $c$ is ground-maximal, we conclude that, for any $c^*\in C(\delta^*)$ contained in $c$, we have $V(\sigma(c))= V(\sigma_{\delta^*}(c^*))$.
But, for any $c^*\in C(\delta^*)$ contained in $c$, there is $c'\in C(\delta')$ such that $V(\sigma_{\delta^*}(c^*))=V(\sigma_{\delta'}(c'))$.
Therefore, $c$ is contained in a cell $c'\in C(\delta')$.
\end{proof}

Given that, if $\delta$ is a rendition of a society $(G,\Omega)$, then $V(\Omega)\subseteq \pi_{\delta}(N(\delta))$, and thus $V(\Omega)\cap (V(\sigma(c))\setminus\pi_{\delta}(\tilde{c}))=\emptyset$ for any cells in $\delta$, we immediately get the following corollary from \autoref{cell_in_cell2}.

\begin{corollary}\label{cor_well-ground}
Let $G$ be a connected graph.
Let $\delta$ be a ground-maximal rendition of $(G,\Omega)$ and let $\delta'$ be a well-linked rendition of $(G,\Omega)$.
Then $\delta$ is more grounded than $\delta'$.
\end{corollary}

\subsection{Combining sphere decompositions}

To prove \autoref{lem_small_leaves}, as well as \autoref{th_param} later, the main ingredient is the following result that says that, given a ground-maximal \sdecomp\ 
of the compass of some flat wall $W$ of $G$ and given a ground-maximal \sdecomp\ of $G-Y$, where $Y$ is a central part of $W$, these two \sdecomps\ can be glued to obtain a \sdecomp\ of $G$. 

\begin{lemma}\label{lem_newirr}
Let $k,r,q\in\Nbbb$ with $r,q$ odd and $r\ge q+10$.
Let $G$ be a graph, $(W,\frak{R}=(A,B,P,C,\delta))$ be a flatness pair of $G$ of height $r$, $G'$ be the $\frak{R}$-compass of $W$, and $Y$ be the vertex set of the $\frak{R}^{(q)}$-compass of a $W^{(q)}$-tilt $(\tilde{W}^{(q)},\frak{R}^{(q)})$ of $(W,\frak{R})$.
Suppose also that:
\begin{itemize}
\item $\delta=(\Gamma,\Dcal)$ is a well-linked rendition of $(G',\Omega)$, where $\Omega$ is the cyclic ordering of the vertices of $A\cap B$ as they appear in $D(W)$, and
\item there are two ground-maximal 
\sdecomps\ $\delta'=(\Gamma',\Dcal')$ and $\delta_Y=(\Gamma_Y,\Dcal_Y)$ of $G'$ and $G-Y$, respectively, 
such that the $(r-2)$-central wall $W'$ of $W$ is grounded in both $\delta'$ and $\delta_Y$. 
\end{itemize}
Let $T$ be the track of third layer $C_3$ of $W$ with respect to $\delta$.
Then $\pi_\delta(T)\subseteq N(\delta')\cap N(\delta_Y)$.

Consequently, 
there is a ground-maximal \sdecomp\ $\delta^*$ of $G$ such that 
each cell of $\delta^*$ is either a cell of $\delta'$ or a cell of $\delta_Y$, and more specifically:
\begin{itemize}
\item $T$ is the track of $C_3$ with respect to $\delta^*$,
\item the restriction of $\delta^*$ to the closed disk delimited by $T$ containing $Y$ (resp. \emph{not} containing $Y$) is, up to homeomorphism, the restriction of $\delta'$ (resp. $\delta_Y$) to the closed disk delimited by $T$ containing $Y$ (resp. \emph{not} containing $Y$).
\end{itemize}
\end{lemma}

\begin{figure}[h]
\centering
\includegraphics[scale=1]{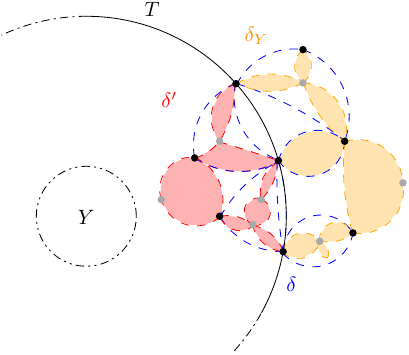}
\caption{The black circle represents $T$, $\delta$ is represented in blue, $\delta'$ in red, and $\delta_Y$ in orange. $\delta^*$ is obtained by combining $\delta'$ inside of $T$ and $\delta_Y$ outside of $T$.}
\label{fig_glue2}
\end{figure}

\begin{proof}
Let $C_i$ be the $i$-th layer of $W$ for $i\in[(r-1)/2]$ (so $C_1$ is the perimeter of $W$).
Let $T$ be the track of $C_3$ in $\delta$. Informally, we want to prove that near $T$, the cells of $\delta'$ and $\delta_Y$ are contained in cells of $\delta$. Hence, we may replace $\delta$ by $\delta'$ in the disk bounded by $T$ containing $v$, and we may replace $\delta$ by $\delta_Y$ in the disk bounded by $T$ not containing $v$. See \autoref{fig_glue2} for an illustration.

To prove that each cell $c$ of $\delta'$ near $T$ is contained in a cell of $\delta$, we want to apply \autoref{cell_in_cell2}.
To do so, we need to check that the interior of $c$ does not contain any vertex of $V(\Omega)$.
Here, "near" $T$ means intersecting $T$ in at least two points, that is $|V(\sigma_{\delta'}(c))\cap\pi_{\delta}(T)|\ge 2$.
This is what we prove in \autoref{claim:1}.
We also prove it for $\delta_Y$ at the same time.

\begin{claim}\label{claim:1}
Let $\delta_1\in\{\delta',\delta_Y\}$.
For each cell $c_1\in C(\delta_1)$ such that 
$|V(\sigma_{\delta_1}(c_1))\cap\pi_{\delta}(T)|\ge 2$, 
it holds that $V(\Omega)\cap V(\sigma_{\delta_1}(c_1))=\emptyset$.
\end{claim}

\begin{cproof}
Let $c_1\in C(\delta_1)$ be such that $V(\sigma_{\delta_1}(c_1))\cap\pi_{\delta}(T)\ne\emptyset$.
If $V(\Omega)\cap V(\sigma_{\delta_1}(c_1))\ne\emptyset$, then $\sigma_{\delta_1}(c_1)$ contains vertices of both $C_1$ and $C_3$, and thus also of $C_2$.
Given that $W'$ is grounded in $\delta_1$,
it holds that $\sigma_{\delta_1}(c_1)$ does not contain $C_2$ nor $C_3$.
Therefore, $\pi_{\delta}(\tilde{c}_1)$ contains two vertices of $C_2$.
But given that $C_2$ and $C_3$ are vertex disjoint and that $\pi_{\delta}(\tilde{c}_1)$ contains two vertices of $\pi_{\delta}(T)\subseteq v(C_3)$, this implies that $|\tilde{c}_1|\ge 4$, a contradiction.
Therefore, $V(\Omega)\cap V(\sigma_{\delta_1}(c_1))=\emptyset$. This completes the proof.
\end{cproof}

Therefore, given that $\delta$ and $\delta'$ are sphere decomposition of the same graph $G'$, we prove using \autoref{cell_in_cell2} that any cell of $\delta'$ intersecting $T$ in at least two point is contained in a cell of $\delta$.

\begin{claim}\label{claim:2}
For any $c'\in C(\delta')$ such that $|V(\sigma_{\delta'}(c'))\cap\pi_{\delta}(T)|\ge2$, there is a cell $c\in C(\delta)$ such that $c'$ is contained in $c$.
\end{claim}

\begin{cproof}
Given that $\delta$ is well-linked, that $\delta'$ is ground-maximal, and that $V(\Omega)\cap V(\sigma_{\delta'}(c'))=\emptyset$ by \autoref{claim:1}, the proof immediately follows from \autoref{cell_in_cell2}.
\end{cproof}

For $\delta_Y$ however, using \autoref{cell_in_cell2} is not that easy.
Indeed, $\delta$ and $\delta_Y$ are not sphere decompositions of the same graph, given that one is a sphere decomposition of $G'$ and the other of $G-Y$.
Therefore, we need to consider the restrictions $\tilde{\delta}$ and $\tilde{\delta}_Y$ of $\delta$ and $\delta_Y$ to $G'-Y$ and to apply \autoref{cell_in_cell2} to these sphere decompositions.
To do so, we need to check that a cell $c_Y$ of $\delta_Y$ near $T$ is still a cell of $\tilde{\delta}_Y$, that is ground-maximal, and that a cell $c$ of $\delta$ intersecting $c_Y$ is still a cell of $\tilde{\delta}$, that is well-linked.
Let us first prove that the vertices of $Y$ are not vertices of $\sigma(c)$, and thus that $c$ is still a cell of $\tilde{\delta}$.

\begin{claim}\label{claim:3}
Let $c_Y\in C(\delta_Y)$ be such that $|V(\sigma_{\delta_Y}(c_Y))\cap\pi_{\delta}(T)|\ge2$.
Then, for any $c\in C(\delta)$ that intersect $c_Y$, $Y\cap V(\sigma(c))=\emptyset$.
\end{claim}

\begin{cproof}
Suppose toward a contradiction that $Y\cap V(\sigma(c))\ne\emptyset$.
The perimeter of $W^{(q)}$ is $C_j$ for $j=(r-q)/2+1$. Given that $r\ge q+10$, this implies that $j\ge 6$.
Hence, by definition of a tilt of a flatness pair, we conclude that $Y\cap V(C_5)=\emptyset$.
Then, by connectivity, either $\sigma(c)$ or $\sigma_{\delta_Y}(c_Y)$ intersects $C_4$ and $C_5$.
Given that $W'$ is grounded in both $\delta$ and $\delta_Y$, neither $C_5$ nor $C_4$ is totally contained in $\sigma(c)$ or $\sigma_{\delta_Y}(c_Y)$.

Given that $|\pi_{\delta_Y}(\tilde{c_Y})\cap\pi_{\delta}(T)|\ge2$ and that $|\pi_{\delta_Y}(\tilde{c_Y})|\le3$, $C_4$ cannot intersect $\pi_{\delta_Y}(\tilde{c_Y})$ in two vertices.
Therefore, $\pi_\delta(c)$ must intersect both $C_4$ and $C_5$.
Hence $|\tilde{c}|\ge4$, a contradiction proving the claim.
\end{cproof}

Let us now prove that any cell of $\delta_Y$ intersecting $T$ in at least two points is contained in a cell of $\delta$.

\begin{claim}\label{claim:4}
For any $c_Y\in C(\delta_Y)$ such that $|V(\sigma_{\delta_Y}(c_Y))\cap\pi_{\delta}(T)|\ge2$, there is a cell $c\in C(\delta)$ such that $c_Y$ is contained in $c$.
\end{claim}

\begin{cproof}
Let $\tilde{\delta}=(\tilde{\Gamma},\tilde{\Dcal})$ and $\tilde{\delta}_Y=(\tilde{\Gamma}_Y,\tilde{\Dcal}_Y)$ be the restrictions of $\delta$ and $\delta_Y$, respectively, to $G'':=G'-Y$, i.e., $\tilde{\delta}=\delta-Y$ and $\tilde{\delta}_Y=\delta_Y-(V(G)\setminus V(G'))$.

Let $c\in C(\delta)$ be a cell that intersects $c_Y$.
By \autoref{claim:3}, $Y\cap V(\sigma(c))=\emptyset$, so $c\in C(\tilde{\delta})$.
Let us show that $c$ is still well-linked in $\tilde{\delta}$.
Given that $c$ is well-linked in $\delta$, there are three vertex-disjoint paths from $\pi_{\delta}(\tilde{c})$ to $V(\Omega)$.
Even if one of them contains a vertex of $Y$, we can reroute the paths using the wall so that they avoid $Y$.
Hence, $c$ is still well-linked after removing $v$.

Also, by \autoref{claim:1}, $V(\Omega)\cap V(\sigma_{\delta_Y}(c_Y))=\emptyset$, so $c_Y\in C(\tilde{\delta}_Y)$.
Therefore, $c_Y$ is ground-maximal in $\tilde{\delta}_Y$.
Indeed, otherwise, there is a \sdecomp\ $\tilde{\delta}_Y'$ of $G''$ that is more grounded than $\tilde{\delta}_Y$ such that, for any cell $c_Y'$ in $\tilde{\delta}_Y'$, $c_Y'$ is either contained in $c_Y$, or is equal to another cell of $\tilde{\delta}_Y$.
But then we can easily add $Y$ back to obtain a \sdecomp\ $\delta_Y'$ that is more grounded than $\delta_Y$, contradicting the maximality of $\delta_Y$.

Hence, the proof follows from \autoref{cell_in_cell2} applied on $(G'',\Omega)$. This completes the proof of the claim.%
\end{cproof}

Given that the cells of $\delta'$ and $\delta_Y$ near $T$ are contained in cells of $\delta$, it implies that no cell of $\delta'$ nor $\delta_Y$ intersects $T$.
This is what we prove in \autoref{claim:5}.

\begin{claim}\label{claim:5}
$\pi_\delta(T)\subseteq N(\delta')\cap N(\delta_Y)$.
\end{claim}

\begin{cproof}
Let $\delta_1\in\{\delta',\delta_Y\}$.
Suppose toward a contradiction that $\pi_\delta(T)\setminus N(\delta_1)\ne\emptyset$.
Then there is a cell $c_1\in C(\delta_1)$ and a vertex $x\in\pi_\delta(T)$ such that $x\in V(\sigma_{\delta_1}(c_1))\setminus \pi_{\delta_1}(\tilde{c}_1)$.
Therefore, $|V(\sigma_{\delta_Y}(c_Y))\cap\pi_{\delta}(T)|\ge2$.
Hence, by one of \autoref{claim:2} and \autoref{claim:4}, there exists a cell $c\in C(\delta)$ such that $c_1$ is contained in $c$.
However, given that $T$ is a track in $\delta$, this implies that $x\notin V(\sigma_{\delta}(c))\setminus \pi_{\delta}(\tilde{c})\supseteq V(\sigma_{\delta_1}(c_1))\setminus \pi_{\delta_1}(\tilde{c}_1)$, a contradiction proving the claim.
\end{cproof}

Therefore, there exists a $\delta'$-aligned disk $\Delta'$
such that ${\sf inner}_\delta(C_3)={\sf inner}_{\delta'}(\Delta')$,
and a $\delta_Y$-aligned disk $\Delta_v$
such that ${\sf inner}_\delta(C_3)-v={\sf outer}_{\delta_Y}(\Delta_v)$.
$\mathbb{S}^2-T$ is composed of two open disks.
Up to homeomorphism, we may assume that the closures of these disks are $\Delta_v$ and $\Delta'$, respectively.
Hence, we can define $\delta^*=(\Gamma^*,\Dcal^*)$ such that
$\Dcal^*:=\{\Delta_c\in\Dcal'\mid c\subseteq \Delta'\}\cup\{\Delta_c\in\Dcal_v\mid c\subseteq \Delta_v\}$ and $\Gamma^*=(\Gamma'\cap\Delta')\cup (\Gamma_v\cap\Delta_v)$.
Each cell $c^*$ of $\delta^*$ is either a cell of $\delta'$ or of $\delta_Y$, and is thus ground-maximal, hence the result.
\end{proof}

A direct corollary is that, if the compass of some flat wall $W$ of $G$ is $\Hcal^{(k)}$-planar and that $G-v$ is $\Hcal^{(k)}$-planar, where $v$ is the cental vertex of $W$, then $G$ is $\Hcal^{(k)}$-planar.

\begin{corollary}\label{lem_combinebis}
Let $\Hcal$ be a hereditary graph class.
Let $k,q,r\in\Nbbb$ with $q,r$ odd and $r\ge\max\{\sqrt{(k+7)/2}+2,q+10\}$.
Let $G$ be a graph, $(W,\frak{R}=(A,B,P,C,\rho))$ be a flatness pair of $G$ of height $r$, $G'$ be the $\frak{R}$-compass of $W$, and $Y$ be the vertex set of the $\frak{R}^{(q)}$-compass of a $W^{(q)}$-tilt $(\tilde{W}^{(q)},\frak{R}^{(q)})$ of $(W,\frak{R})$.
Then $G$ is a \yes-instance of \Hkpl\ if and only if $G'$ and $G-Y$ are both \yes-instances of \Hkpl.
\end{corollary}

\begin{proof}
Obviously, if one of $G'$ and $G-Y$ is not $\Hcal^{(k)}$-planar, then neither is $G$ by heredity of the $\Hcal^{(k)}$-planarity.
Let us suppose that both $G'$ and $G-Y$ are $\Hcal^{(k)}$-planar.
We want to prove that $G$ is $\Hcal^{(k)}$-planar. For this, we find a well-linked rendition $\delta$ of $(G',\Omega)$ and ground-maximal sphere decompositions $\delta'$ of $G'$ and $\delta_Y$ of $G-Y$. 

Let $C_i$ be the $i$-th layer of $W$ for $i\in[(r-1)/2]$ (so $C_1$ is the perimeter of $W$).
Let $\Omega$ be the cyclic ordering of the vertices of $A\cap B$ as they appear in $C_1$.
Hence, $\rho$ is a rendition of $(G'=G[B],\Omega)$.
Then, by \autoref{lem_rend_to_min_or_max}, $(G',\Omega)$ has a well-linked rendition $\delta$.

By \autoref{cor_compatible}, given that $r\ge\max\{\sqrt{(k+7)/2}+2,7\}$, there are two ground-maximal $\Hcal^{(k)}$-compatible \sdecomps\ $\delta'=(\Gamma',\Dcal')$ and $\delta_Y=(\Gamma_Y,\Dcal_Y)$ of $G'$ and $G-Y$, respectively, 
such that the $(r-2)$-central wall $W'$ of $W$ is grounded in both $\delta'$ and $\delta_Y$. 

Then, by \autoref{lem_newirr}, given that $r\ge q+10$, there exists a ground-maximal \sdecomp\ $\delta^*$ of $G$ such that each cell $\delta^*$ is either a cell of $\delta'$ or a cell of $\delta_Y$, and is thus $\Hcal^{(k)}$-compatible.
Thus, by \autoref{obs_sol_compatible}, $G$ is $\Hcal^{(k)}$-planar.
\end{proof}

If we take $q=3$, then we get the following.

\begin{corollary}\label{lem_combine}
Let $\Hcal$ be a hereditary graph class.
Let $k,r\in\Nbbb$ with $r\ge\max\{\sqrt{(k+7)/2}+2,13\}$.
Let $G$ be a graph, $(W,\frak{R}=(X,Y,P,C,\rho))$ be a flatness pair of $G$ of height $r$, $G'$ be the $\frak{R}$-compass of $W$, and $v$ be a central vertex of $W$.
Then $G$ is a \yes-instance of \Hkpl\ if and only if $G'$ and $G-v$ are both \yes-instances of \Hkpl.
\end{corollary}

\subsection{Proof of \autoref{lem_small_leaves}}\label{subsec_proof}

We can finally prove \autoref{lem_small_leaves}.

\smallleaves*

\begin{proof}
We apply \autoref{prop_flatwallth} to $G$,
with $k'=\lceil\sqrt{k+4}\rceil+2$ and $r=\max\{\odd(\sqrt{(k+7)/2}+2),13\}$.
It runs in time $\Ocal_k(n+m)$.

If $G$ has treewidth at most $f_{\ref{prop_flatwallth}}(k')\cdot r$, then we apply \autoref{Courcelle} to $G$ in time $\Ocal_k(n)$ and solve the problem. We can do so because 
the graphs in $\Hcal^{(k)}$ have a bounded size, so $\Hcal^{(k)}$ is a finite graph class, hence trivially CMSO-definable.
Therefore, by \autoref{obs_CMSO}, \Hkpl\ is expressible in CMSO logic.

If $G$ contains an apex grid of height $k'$ as a minor, then by \autoref{lem:obstructions}, we obtain that $G$ has no planar $\Hcal^{(k)}$-modulator and report a \no-instance.

Hence, we can assume that there is a flatness pair $(W,\frak{R})$ of height $r$ in $G$ whose $\frak{R}$-compass $G'$ has treewidth at most $f_{\ref{prop_flatwallth}}(k')\cdot r$.
Let $v$ be a central vertex of $W$.
We apply \autoref{Courcelle} to $G'$ in time $\Ocal_k(n)$ and 
we recursively apply our algorithm to $G-v$. 
If the outcome is a \no-instance for one of them, then this is also a \no-instance for $G$.
Otherwise, the outcome is a \yes-instance for both.
Then, by \autoref{lem_combine}, we can return a \yes-instance.

The running time of the algorithm is $T(n)=\Ocal_k(n+m)+T(n-1)=\Ocal_k(n(n+m))$.
\end{proof}

Note that the running time of \autoref{prop_FWth}, and thus \autoref{prop_flatwallth} can be modified so that the dependence in $t$ (resp. $k$) and $r$ is explicit.
Therefore, the only reason we cannot give an explicit dependence in $k$ here is Courcelle's theorem.

\section{Planar elimination distance}
\label{cosi46y}

In this section, we prove \autoref{cor_bigH} in \autoref{subsec_big_leaf}, and \autoref{lem_edk} in \autoref{subsec_small_ed}, after having given a necessary auxiliary result in \autoref{subsec_aux}.

\subsection{Finding a big leaf in $\Hcal$}\label{subsec_big_leaf}

In this section, we prove \autoref{cor_bigH}.\medskip

The algorithm uses the following result.

\begin{proposition}[\cite{ChitnisCHPP16desi}]\label{rd_sampling}
Given a set $V$ of size $n$ and $a,b\in[0,n]$, one can construct in time $2^{\Ocal(\min\{a,b\}\log(a+b))}\cdot n\log n$
a family $\Fcal_{a,b}$ of at most $2^{\Ocal(\min\{a,b\}\log(a+b))}\cdot \log n$ subsets of $V$ such that the following holds: 
for any disjoint sets $A,B\subseteq U$ with $|A|\le a$ and $|B|\le b$, there exists a set $R\in\Fcal_{a,b}$ such that $A\subseteq R$ and $B\cap R=\emptyset$.
\end{proposition}

\begin{proof}[Proof of \autoref{cor_bigH}]
The algorithm goes as follows. 
We construct the family $\Fcal_{a,k'}$ of \autoref{rd_sampling} in time $2^{\Ocal(k\log(a+k))}\cdot n\log n$.
For each $U\in\Fcal_{a,k'}$, we do the following.
We construct $\Ccal_U:=\{C\in\cc(G-U)\mid C\notin\Hcal\}$ in time $\Ocal(n^c+n+m)$ and set $Z_U:=\bigcup_{C\in\Ccal_U}V(C)$.
If $|A_U|\le k'$, where $A_U:=N_G(Z_U)$, then let $C_U\in\cc(G-A_U)$ be the unique component of size at least $a$, if it exists.
We compute, if it exists, a minimum solution $S_U$ of \pbdel for the instance $(C_U,k'-|A_U|)$ in time $k'\cdot f(k')\cdot n^c$. 
For each subset $Y_U\subseteq V(G)\setminus N_G[V(C_U)]$ (which has size at most $a-1$), we set $X_U:=Y_U\cup A_U\cup S_U$.
We check whether $\torso(G,X_U)\in\Gcal_k$ (in time at most $2^{\Ocal((a+k)^2)}+n+m$ given that there are at most $2^{|X_U|^2}$ graphs with $|X_U|\le k'+a-1$ vertices) and $G-X_U\in\Hcal$ (in time $\Ocal(n^c)$).
If that is the case, we return the set $X_U$.
If, for every $U\in\Fcal_{a,k'}$, we did not return anything, then we return a \no-instance.

\paragraph{Running time.} Given that $|\Fcal_{a,k'}|\le 2^{\Ocal(k\log(a+k))}\cdot \log n$, the algorithm takes time $f(k)\cdot 2^{\Ocal((a+k)^2)}\cdot \log n\cdot (n^c+n+m)$.

\paragraph{Correctness.} Obviously, if the algorithm outputs a set $X_U$, then it is a $\Gcal_k\triangleright\Hcal$-modulator of $G$.
It remains to show that if $G$ admits a big-leaf $\Gcal_k\triangleright\Hcal$-modulator, then the algorithm indeed outputs some set $X_U$.
Assume that $G$ admits a big-leaf $\Gcal_k\triangleright\Hcal$-modulator $X$ with big leaf $D$.

We set $L:=N_G[V(D)]$, $R:=V(G)\setminus V(D)$, $A:=L\cap R=N_G(V(D))$, and $B:=R\setminus L=V(G)-N_G[V(D)]$.
By \autoref{obs_clique}, we have $|N_G(V(D))|\le k'$.
Therefore, $(L,R)$ is a separation of $G$ of order at most $k'$.
Moreover, $G$ is $(a,k')$-unbreakable, so given that $|L\setminus R|=|V(D)|\ge a$, we conclude that $|B|\le a-1$.
Therefore, by \autoref{rd_sampling}, we can construct in time $2^{\Ocal(k\log(a+k))}\cdot n\log n$
a family $\Fcal$ of at most $2^{\Ocal(k\log(a+k))}\cdot \log n$ subsets of $V(G)$ such that there exists $U\in\Fcal$ with $A\subseteq U$ and $B\cap U=\emptyset$.

\begin{figure}[h]
\center
\includegraphics{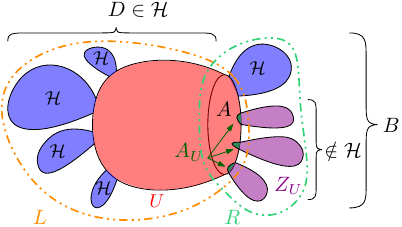}
\caption{Illustration for the correctness of \autoref{cor_bigH}. $C_U$ is the union of the red and the blue part.}
\end{figure}

By heredity of $\Hcal$, for each $C\in\Ccal_U$, $C$ is not a subgraph of $D$.
Therefore, $V(C)\subseteq B$, and thus $N_G(V(C))\subseteq A$.
Hence, $Z_U\subseteq B$ and $A_U\subseteq A$.
In particular, $|Z_U|\le|B|<a$.
Moreover, $V(D)\subseteq V(G)\setminus A_U$, so, given that $G$ is $(a,k')$-unbreakable, the component $C_U\subseteq \cc(G-A_U)$ such that $V(D)\subseteq V(C_U)$ is the unique component of size at least $a$, and $|V(G)\setminus N_G[V(C_U)]|<a$.
Given that, for each $C\in\cc(G-U)\setminus\Ccal_U$, $C\in\Hcal$, and that $\Hcal$ is union-closed, we conclude that $A\setminus A_U$ is a solution of \pbdel for the instance $(C_U,k'-|A_U|)$.
Therefore, the algorithm should find a minimum solution $S_U$ of \pbdel for the instance $(C_U,k'-|A_U|)$.

We set $Y_U:=X\setminus N_G[V(C_U)]$.
It remains to prove that $X_U:=Y_U\cup A_U\cup S_U$ is a $\Gcal_k\triangleright\Hcal$-modulator of $G$.
For each $C\in \cc(G-X_U)$, we have either $C\in \cc(G-X)$, or $V(C)\subseteq V(C_U)$.
In the first case, it immediately implies that $C\in \Hcal$.
In the second case, it implies $N_G(C)\subseteq A_U\cup S_U$. Hence, given that $C_U\in\cc(G-A_U)$ and that $S_U$ is a solution of \pbdel for the instance $(C_U,k'-|A_U|)$, we also conclude that $C_U\in\Hcal$.
It remains to prove that $\torso(G,X_U)\in\Gcal_k$.
We only write the proof for planar treedepth and planar treewidth, as the proof is simpler for treedepth and treewidth.

\begin{claim}
If $\Gcal_k$ is the class of graphs of planar treedepth at most $k$, then $\torso(G,X_U)\in\Gcal_k$.
\end{claim}

\begin{cproof}
Given that $X$ is a $\Gcal_k\triangleright\Hcal$-modulator of $G$, there is a certifying elimination sequence $X_1,\dots, X_k$.
We need to prove that $\torso(G,X_U)\in\Gcal_k$, or, equivalently, that there is a certifying elimination sequence $X_1',\dots, X_k'$ whose union is $X_U$.
Remember that $A':=A\setminus A_U$ is a solution of \pbdel for the instance $(C_U,k'-|A_U|)$.
By minimality of $S_U$, we have $|S_U|\le|A'|$, so there is an injective function $\tau$ from $S_U$ to $A'$.
We define the partition $(X_1',\dots,X_k')$ of $X_U$ such that, for $i\in[k]$, $X_i'=X_i\setminus A'\cup \tau^{-1}(A'\cap X_i)$ (remember that $X_i\setminus A'\subseteq Y\cup A_U$ and that $\tau^{-1}(A'\cap X_i)\subseteq S_U$).
We define $G_1'=G$ and $G_{i+1}'=G_i'-X_i'$ for $i\in[k]$.
It remains to prove that $\torso(G_i',X_i')$ is planar, for $i\in[k]$.
Given that $|A\cap X_i|\le 4$, we conclude that $(A_U\cup S_U)\cap X_i'$ induces at most a $K_4$ in $\torso(G_i',X_i')$.
Moreover, $N_G(A')\subseteq A_U$, and thus $N_{\torso(G_i',X_i')}(S_U\cap X_i')\subseteq A_U\cap X_i'$, so $\torso(G_i',X_i')$ is indeed planar.
Therefore, $X_1',\dots, X_k'$ is indeed a certifying elimination sequence.
\end{cproof}

\begin{claim}
If $\Gcal_k$ is the class of graphs of planar treewidth at most $k$, then $\torso(G,X_U)\in\Gcal_k$.
\end{claim}

\begin{cproof}
Let $(T,\beta)$ be a tree decomposition of 
$\torso(G,X)$ of planar width at most $k$.
Given that $A=N_G(V(D))$ induces a clique in $\torso(G,X)$, there is $t\in V(T)$ such that $A\subseteq \beta(t)$.
Let $A'=A\setminus A_U$.
Let $(T,\beta')$ be the tree decomposition of $\torso(G,X_U)$ such that $\beta'(t)=\beta(t)\setminus A'\cup S_U$ and $\beta'(t')=\beta(t')\setminus A'$ for $t'\in V(T)\setminus\{t\}$.
Observe that, for each $C\in\cc(G-X_U)$, either $C\subseteq C_U$, in which case $N_G(V(C))\subseteq A_U\cup S_U\subseteq \beta'(t)$, or $C\nsubseteq C_U$, in which case $N_G(V(C))\subseteq A_U\cup Z_U\subseteq\beta'(t')$ for some $t'\in V(T)$.
Additionally, $N_{\torso(G,X_U)}(S_U)\subseteq A_U$.
This justify that $(T,\beta')$ is indeed a tree decomposition of $\torso(G,X_U)$.
If $\beta(t)$ has size at most $k$, then so does $\beta'(t)$ given that $|S_U|\le |A'|$.
If $\beta(t)$ has a planar torso,  then $|A_U\cup S_U|\le|A|\le 4$, and $N_{\torso(G,X_U)}(S_U)\subseteq A_U$, so the torso of $\beta'(t)$ is also planar.
Therefore, $(T,\beta')$ has planar width at most $k$, hence the result.
\end{cproof}
\end{proof}

\subsection{An auxiliary lemma}\label{subsec_aux}

The following lemma says that, given a big enough flat wall $W$ and a vertex set $S$ of size at most $k$, there is a smaller flat wall $W'$ such that the central vertices of $W$ and the intersection of $S$ with the compass of $W'$ are contained in the compass of central 5-wall of $W'$.
This result is used in \cite{SauST23kapices}, but is not stated as a stand-alone result, so we reprove it here for completeness.

\begin{lemma}\label{claim_irr}
There exists a function $f_{\ref{claim_irr}}: \mathbb{N}^{3}\to \mathbb{N},$
whose images are odd integers, such that the following holds.

Let
$k,q,z\in\mathbb{N}$, with odd $q\geq 3$ and odd $z\ge 5$,
$G$ be a graph,
$S\subseteq V(G)$, where $|S|\le k$,
$(W,\mathfrak{R})$ be a flatness pair of $G$ of height at least $f_{\ref{claim_irr}}(k,z,q)$,
and $(W',\mathfrak{R}')$ be a $W^{(q)}$-tilt of $(W,\mathfrak{R})$.
Then, there is a flatness pair $(W^*,\frak{R}^*)$ of $G$ such that:
\begin{itemize}
\item $(W^*,\frak{R}^*)$ is a $\tilde{W}'$-tilt of $(W,\frak{R})$ for some $z$-subwall $\tilde{W}'$ of $W$, and
\item $V(\textsf{Compass}_{\mathfrak{R}'}(W'))$ and $S\cap V(\textsf{Compass}_{\mathfrak{R}^*}(W^*))$ are both subsets of the vertex set of the compass of every $W^{*(5)}$-tilt of $(W^*,\frak{R}^*)$. 
\end{itemize}
Moreover, $f_{\ref{claim_irr}}(k,z,q)=\odd((k+1)\cdot (z+1)+q).$
\end{lemma}

\begin{proof}
Let $r:=f_{\ref{claim_irr}}(k,z,q)$.
For every $i\in[r],$ we denote by $P_{i}$ (resp. $Q_{i}$) the $i$-th vertical (resp. horizontal) path of $W.$ Let $z':=\frac{z+1}{2}$
and observe that, since $z$ is odd, we have $z'\in \mathbb{N}.$
We also define, for every $i\in[k+1]$ the graph
\[B_{i}:=\bigcup_{j\in [z'-1]}P_{f_{z'}(i,j)} \cup\bigcup_{j\in [z']} P_{r+1-f_{z'}(i,j)}\cup \bigcup_{j\in [z'-1]}Q_{f_{z'}(i,j)} \cup \bigcup_{j\in [z']} Q_{r+1-f_{z'}(i,j)},\]
where $f_{z'}(i,j):=j+(i-1)\cdot (z'+1)$.
For every $i\in[k+1],$ we define ${W}_{i}$ to be the graph obtained from $B_{i}$
after repeatedly removing from $B_{i}$ all vertices of degree one (see \autoref{label_desasosegado} for an example).
Since $z=2z'-1,$ for every $i\in[k+1]$, ${W}_{i}$ is a $z$-subwall of $W.$
For every $i\in[k+1],$  we set $L^{i}_\mathsf{in}$ to be the inner layer of $W_{i}.$
Notice that $L^{i}_\mathsf{in},$ for $i\in[k+1],$ and $D(W^{(q)})$ are $\mathfrak{R}$-normal cycles of $\mathsf{Compass}_{\mathfrak{R}}(W).$

\begin{figure}[ht]
\centering
\includegraphics[scale=0.5]{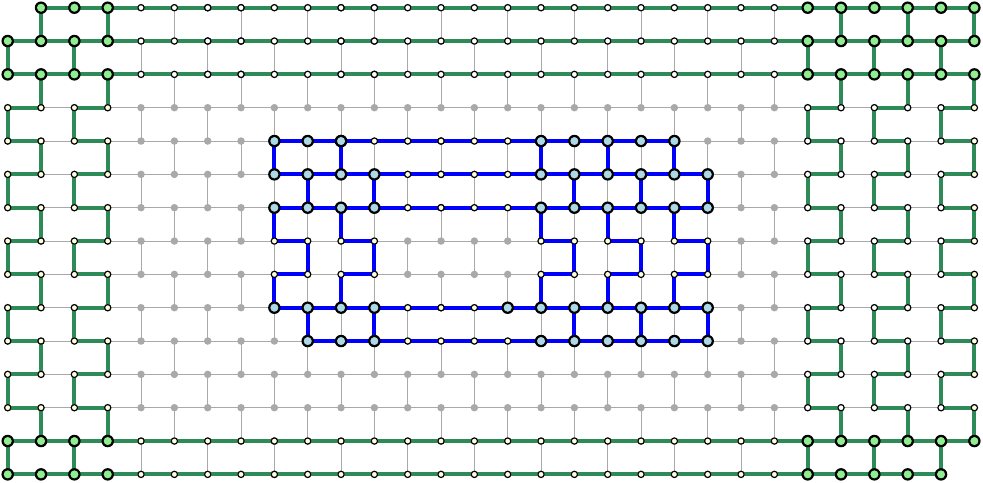}
\caption{A 15-wall and the 5-walls $W_1$ and $W_2$ as in the proof of  \autoref{claim_irr}, depicted in green and blue, respectively. The white vertices are subdivision vertices of the walls $W_1$ and $W_2$.}
\label{label_desasosegado}
\end{figure}

\smallskip
By definition of a tilt of a flatness pair, it holds that $V(\textsf{Compass}_{\mathfrak{R}'}(W'))\subseteq V(\cupall\mathsf{Influence}_{\mathfrak{R}}(W^{(q)}))$.
Moreover, for every $i\in[k+1],$ the fact that $r\geq (k+1)\cdot (z+2)+q$ implies that
$\cupall\mathsf{Influence}_{\mathfrak{R}} (W^{(q)})$ is a subgraph of $\cupall\mathsf{Influence}_{\mathfrak{R}} (L_{\mathsf{in}}^{i}).$
Hence, for every $i\in[k+1],$ we have that
$V(\textsf{Compass}_{\mathfrak{R}'}(W'))\subseteq V(\cupall\mathsf{Influence}_{\mathfrak{R}}(L_{\mathsf{in}}^{i})).$

\smallskip
For every $i\in[k+1],$ let $(W_{i}',\mathfrak{R}_{i})$ be a flatness pair of $G$ that is a $W_{i}$-tilt of $(W,\mathfrak{R})$ ({which exists due to \autoref{prop_tilt}}).
Also, note that, for every $i\in[k+1],$  $L_{\mathsf{in}}^i$ is the inner layer of $W_{i}$ and therefore it is
an $\mathfrak{R}_i$-normal cycle of $\mathsf{Compass}_{\mathfrak{R}_i}(W_i').$

\smallskip
For every $i\in[k+1],$ we set
$D_{i} := V(\textsf{Compass}_{\mathfrak{R}_i }(W_i'))\setminus V(\cupall\mathsf{Influence}_{\mathfrak{R}_i }(L_\mathsf{in}^{i}))$.
Given that the vertices of $V(W_i)$ are contained between the $((i-1)\cdot(z'+1)+1)$-th and the $(i\cdot(z'+1)-1)$-th layers of $W$ for $i\in[k+1]$, it implies that the vertex sets $D_i,$ $i\in[k+1]$, are pairwise disjoint.
Therefore, since that $|S|\le k$, there exists a $j\in [k+1]$ such that
$S\cap D_j=\emptyset.$
Thus, $S\cap V(\textsf{Compass}_{\mathfrak{R}_j}(W_j'))\subseteq V(\cupall\mathsf{Influence}_{\mathfrak{R}_j}(L_\mathsf{in}^j))$.

\smallskip
Let $Y$ be the vertex set of the compass of some $W_j^{(5)}$-tilt
of $(W_j',\mathfrak{R}_j)$.
Note that $L_{\sf in}^j$ is the perimeter of $W_j^{(3)}$, and therefore, we have
$S\cap V(\textsf{Compass}_{\mathfrak{R}_j}(W_j'))\subseteq V(\cupall\mathsf{Influence}_{\mathfrak{R}_j}(L_\mathsf{in}^j))\subseteq Y$ and $V(\textsf{Compass}_{\mathfrak{R}'}(W'))\subseteq V(\cupall\mathsf{Influence}_{\mathfrak{R}}(L_\mathsf{in}^j))\subseteq Y$.
Therefore, $(W_{j}',\mathfrak{R}_{j})$ is the desired flatness pair.
\end{proof}

\subsection{The algorithm}\label{subsec_small_ed}

We can now prove \autoref{lem_edk}.

\begin{proof}[Proof of \autoref{lem_edk}]
We set $\alpha=\sqrt{a+3}+2$, $d=\alpha^4$, $s'=a+4k-3$, $s=(d-1)\cdot g_{\ref{prop_FWth}}(k')+s'$, $z=\max\{\sqrt{(a+4k-2)/2}+6,11\}$, $r_3=f_{\ref{claim_irr}}(4(k-1),z,3)$, $r_2=\lceil \sqrt{s}\cdot (r_3+1)\rceil$, and $r_1=r_2+2\alpha$.
The algorithm goes as follows.
If $k=0$, then it reduces to checking whether $G\in\Hcal^{(a-1)}$, which can be done in time $\Ocal_a(1)$ given that $\Hcal^{(a-1)}$ is finite.
If $k=1$, then  it reduces to checking whether $G$ is $\Hcal^{(a-1)}$-planar.
Therefore, we can apply \autoref{lem_small_leaves} in time $\Ocal_a(n\cdot (n+m))$ and conclude. 
Hence, we now assume that $k\ge2$.
We apply the algorithm of \autoref{prop_FWth} with input $(G,k'+a,r_1)$, which runs in time $\Ocal_{k,a}(n)$.

If $K_{k'+a}$ is a minor of $G$, then we report a \no-instance.
We can do so because the graphs of planar treedepth at most $k$ are $K_{4k+1}$-minor-free, and thus so is the torso of any $\Pcal^k\triangleright\Hcal^{(a-1)}$-modulator of $G$, and the graphs in $\Hcal^{(a-1)}$ have at most $a-1$ vertices. Hence, if $G$ has \hptd{$\Hcal$} at most $k$, then $G$ is $K_{k'+a}$-minor-free.

If $G$ has treewidth at most $f_{\ref{prop_FWth}}(k'+a)\cdot r_1$, then we apply \autoref{Courcelle} to $G$ in time $\Ocal_{k,a}(n)$ and solve the problem.
We can do so because $\Hcal^{(a-1)}$-planarity is expressible in CMSO logic, and therefore, by induction, having \hptd{$\Hcal^{(a-1)}$} at most $k$ is also expressible in CMSO logic.

Hence, we can assume that there is a set $A\subseteq V(G)$ of size at most $g_{\ref{prop_FWth}}(k')$ and a flatness pair $(W_1,\frak{R}_1)$ of $G-A$ of height $r_1$ such that $\compass_{\frak{R}_1}(W_1)$ has treewidth at most $f_{\ref{prop_FWth}}(k'+a)\cdot r_1$.
Let $W_2$ be the central $r_2$-subwall of $W_1$.

Given that $r_2\ge\lceil \sqrt{s}\cdot (r_3+1)\rceil$, we can find a collection $\Wcal'=\{W'^1_3,\dots,W'^s_3\}$ of $r_3$-subwalls of $W_2$ such that the sets $\influence_{\frak{R}_i}(W_i)$ are pairwise disjoint.
Then, by applying the algorithm of \autoref{prop_tilt}, in time $\Ocal(n+m)$, we find a collection $\Wcal=\{W^1_3,\dots,W^s_3\}$ such that, for $i\in[s]$, $(W^i_3,\frak{R}^i_3)$ is a $W'^i_3$-tilt of $(W_1,\frak{R}_1)$, and the graphs $\compass_{\frak{R}^i_3}(W^i_3)$ are pairwise disjoint and have treewidth at most $f_{\ref{prop_FWth}}(t)\cdot r_1$.

Let $A^-$ denote the set of vertices of $A$ that are adjacent to vertices in the compass of at most $d-1$ walls of $\Wcal$, and let $A^+:=A\setminus A^-$.
$A^-$ can be constructed in time $\Ocal_k(n)$.
If $|A^+|\ge4k$, then return a \no-instance, as justified later in \autoref{cl_A+}.
Then, given that $s\ge (d-1)\cdot |A|+s'$, by the pigeonhole principle, there is $I\subseteq[s]$ of size $s'$ such that no vertex of $A^-$ is adjacent to the $\frak{R}^{i}_3$-compass of $W^i_3$ for $i\in I$.
Therefore, $(W^i_3,\frak{R}^{i\prime}_3)$ is a flatness pair of $G-A^+$ for $i\in I$, where $\frak{R}^{i\prime}_3$ is the 5-tuple obtained from $\frak{R}^{i}_3$ by adding $A^-$ to its first coordinate. 

For $i\in I$, let $F_i$ denote the $\frak{R}^{i\prime}_3$-compass of $W^i_3$.
Given that $F_i$ has treewidth at most $f_{\ref{prop_FWth}}(t)\cdot r_1$, we apply \autoref{Courcelle} to 
compute in time $\Ocal_{k,a}(n)$ the minimum $d_i\le k$, if it exists, such that $F_i$ is $(\Pcal^{d_i-1}\triangleright \Hcal)^{(a-1)}$-planar.
If, for all $i\in I$, such a $d_i$ does not exist,
then 
we report a \no-instance as justified later in \autoref{cl_dp_small}.
Otherwise, there is $p\in I$ such that $d_p$ is minimum. 

Let $v$ be a central vertex of $W_3^p$.
We apply recursively our algorithm to $G-v$ and return a \yes-instance if and only if it returns a \yes-instance.

\paragraph{Correctness.} We now prove the correctness of the algorithm.
Suppose that $G-v$ has \hptd{$\Hcal^{(a-1)}$} at most $k$.
Let us prove that it implies that $G$ has \hptd{$\Hcal^{(a-1)}$} at most $k$.
Let $X_1,\dots,X_k$ be a certifying elimination sequence of $G$ and let $G_i:=G_{i-1}-X_i$ for $i\in[k]$, where $G_0:=G-v$.

Let $j:=\max\{i\in[0,k]\mid \exists C_i\in\cc(G_i),|V(C_i)|\ge a\}$.
Given that $G_0=G-v$ contains as a connected subgraph the graph $W_1-v$ of size more that $a$, we conclude that there is $C_0\in\cc(G_0)$ such that $|V(C_0)|\ge a$, so $j$ is well-defined.
Let $C_j\in\cc(G_j)$ be such that $|V(C_j)|\ge a$.
See \autoref{fig_bigPlED} for an illustration.
Given that, for each $C\in\cc(G_k)$, $C\in\Hcal^{(a-1)}$, and thus $|V(C)|<a$, we conclude that $j<k$.

\begin{claim}\label{cl_Cj_pl}
$|V(G-v)\setminus V(C_j)|\le a+4k-5$ and 
$C_j$ is $(\Pcal^{k-j-1}\triangleright\Hcal)^{(a-1)}$-planar.
\end{claim}

\begin{cproof}
Given that $X_1,\dots,X_k$ is a certifying elimination sequence of $G-v$, it implies that $C_j$ has \hptd{$\Hcal^{(a-1)}$} at most $k-j$ and that $X_j\cap V(C_j)$ is a planar $\Pcal^{k-j-1}\triangleright\Hcal^{(a-1)}$-modulator of $C_j$.
Since $G$ is an $(a,k')$-unbreakable graph and that $k'= 4k$, it implies that $G-v$ is $(a,4k-1)$-unbreakable.
Note that $(N_{G-v}[V(C_j)],V(G-v)\setminus V(C_j))$ is a separation of order at most $|N_{G-v}(V(C_j))|\le4j\le4(k-1)<4k-1$, which implies that $|V(G-v)\setminus V(C_j)|\le a-1+4j\le a+4k-5$.
By maximality of $j$, for each $C\in\cc(C_j-X_{j+1})\subseteq \cc(G_{j+1})$, $|V(C)|<a$.
Therefore, given that $X_{j+1}\cap V(C_j)$ is a planar $\Pcal^{k-j-1}\triangleright\Hcal^{(a-1)}$-modulator of $C_j$, it is also a planar $(\Pcal^{k-j-1}\triangleright\Hcal)^{(a-1)}$-modulator of $C_j$.
\end{cproof}

The next claim justify that we can report a \no-instance if, for all $i\in I$, there is no $d_i\le k$ such that $F_i$ is $(\Pcal^{d_i-1}\triangleright \Hcal)^{(a-1)}$-planar.

\begin{claim}\label{cl_dp_small}
$d_p\le k-j$.
\end{claim}

\begin{cproof}
Given that, by \autoref{cl_Cj_pl}, $|V(G-v)\setminus V(C_j)|\le a+4k-5\le s'-2$, it implies, by the pigeonhole principle, that there is $q\in I\setminus\{p\}$ such that $F_q$ is entirely contained in $C_j$. 
Hence, by \autoref{cl_Cj_pl}, $F_q$ is $(\Pcal^{k-j-1}\triangleright\Hcal)^{(a-1)}$-planar, and 
therefore, $d_q\le k-j$.
Additionally, by minimality of $d_p$, $d_q\ge d_p$.
Therefore, $d_p\le k-j$.
\end{cproof}

We now prove that we have a \no-instance when $|A^+|\ge 4k$.

\begin{claim}\label{cl_A+}
$A^+\subseteq N_G(V(C_j))$, and hence $|A^+|\le 4k-1$.
\end{claim}
\begin{cproof}
Let $u\in A^+$.
By definition, $u$ is adjacent to the compass of at least $d\ge \alpha^4$ walls of $\Wcal$. These compasses are connected and pairwise disjoint, and are contained in the central $r_2$-subwall of $W_1$, where $r_2=r_1-2\alpha$.
Then, observe that $G$ contains as a minor an $(r_1\times r_1)$-grid (obtained by contracting the intersection of horizontal and vertical paths of $W_1$) along with a vertex (corresponding to $u$) that is adjacent to $d$ vertices of its central $(r_2\times r_2)$-subgrid (corresponding to $W_2$).
Thus, by \autoref{prop_apex}, $G$ contains a model of an apex grid $\Gamma_\alpha^+$ of height $\alpha$, where the branch set of the universal vertex is the singleton $\{u\}$.
Therefore, by \autoref{lem:obstructions}, given that $\alpha\ge\sqrt{a+3}+2$, and that $C_j$ admits a planar $(\Pcal^{k-j-1}\triangleright\Hcal)^{(a-1)}$-modulator, we conclude that $u\notin V(C_j)$.
Given that $|V(G-v)\setminus V(C_j)|\le a+4k-5$ and that $s\ge a+4k-4$, it implies that $u$ neighbors at least one vertex of $V(C_j)$, so we conclude that $u\in N_G(V(C_j))$.
Therefore, $|A^+|\le|N_G(V(C_j))|\le|N_{G-v}(V(C_j))|+1\le4k-1$.
\end{cproof}

We set $S:=N_{G-v}(V(C_j))$, which has size at most $4(k-1)$.
Remember that $(W_3^p,\frak{R}^p_3)$ is a flatness pair of $G-A^+$ of height $r_3\ge f_{\ref{claim_irr}}(4(k-1),z,3)$.
Note also that $v$ is contained in the compass of every $W^{(3)}$-tilt of $(W_3^p,\frak{R}^p_3)$.
By \autoref{claim_irr}, 
there is a flatness pair $(W^*,\frak{R}^*)$ of $G-A^+$ that is a $\tilde{W}$-tilt of $(W_3^p,\frak{R}^p_3)$ for some $z$-subwall $\tilde{W}$ of $W_3^p$ 
such that $v$ and $S\cap V(\compass_{\frak{R}^*}(W^*))$ are contained in the compass of every $W^{*(5)}$-tilt of $(W^*,\frak{R}^*)$.
Let $Y$ be the vertex set of the compass of some $W^{*(5)}$-tilt $(W_Y,\frak{R}_Y)$ of $(W^*,\frak{R}^*)$.

We set $G':=\compass_{\frak{R}^*}(W^*)$.
Let $G^*$ be the graph induced by $V(C_j)$ and $Y$, and let $\frak{R}''$ be the 5-tuple obtained from $\frak{R}^*$ after removing $V(G)\setminus V(C_j)$ from its first coordinate.

\begin{claim}\label{cl_fl}
$(W^*,\frak{R}'')$ is a flatness pair of $G^*$ and $G'=\compass_{\frak{R}''}(W^*)$.
\end{claim}

\begin{cproof}
Remember that $Y$ is the vertex set of the $\frak{R}_Y$-compass of the flatness pair $(W_Y,\frak{R}_Y=(A_Y,Y,P_Y,C_Y,\rho_Y))$ of $G-A^+$ and that 
$G'$ is the $\frak{R}^*$-compass of the flatness pair $(W^*,\frak{R}^*=(A^*,V(G'),P^*,C^*,\rho^*))$ of $G-A^+$.
Hence, $(W^*,\frak{R}''=(A^*\cap V(C_j),V(G'),P^*,C^*,\rho^*))$ is a flatness pair of $G-A^+-(A^*\setminus V(C_j))$.
Let us prove that $G-A^+-(A^*\setminus V(C_j))=G^*$.
Given that $V(G^*)=V(C_j)\cup Y$ and that $Y\subseteq V(G')\setminus A^*$, we trivially have that $G^*$ is an induced subgraph of $G-A^+-(A^*\setminus V(C_j))$.

For the other direction, it is enough to prove that $V(G')\subseteq V(C_j)\cup Y$, given that $A^*\cap V(C_j)\subseteq V(C_j)$.
Given that $W^*$ has height $z$ and that $W_Y$ has height five, $W^*-V(W_Y)$, and thus $G'-Y$, contains a wall of height $z-6$ has a subgraph.
Thus, $|V(G')\setminus Y|\ge2(z-6)^2-2\ge a+4k-4$.
Hence, by \autoref{cl_Cj_pl}, we conclude that $|(V(G')\setminus Y)\cap V(C_j)|\ge 1$.
Remember that $(N_{G-v}[V(C_j)],V(G-v)\setminus V(C_j))$ is a separation of $G-v$ with separator $N_{G-v}(V(C_j))=S$.
Given that $S\cap V(G')\subseteq Y$, it thus implies that $V(G')\cap (V(G-v)\setminus V(C_j))\subseteq Y$ by connectivity of the compass.
Since we also have $v\in Y$, we conclude that $V(G')\subseteq V(C_j)\cup Y$, and thus that $(W^*,\frak{R}'')$ is a flatness pair of $G^*$.
\end{cproof}

Remember that $v\in Y$.
Let us show that $G^*$ has \hptd{$\Hcal^{(a-1)}$} at most $k-j-1$.
We will later combine this result with the fact that $G-v$ has \hptd{$\Hcal^{(a-1)}$} at most $k$ to prove that $G$ has \hptd{$\Hcal^{(a-1)}$} at most $k$.

\begin{figure}[h]
\center
\includegraphics{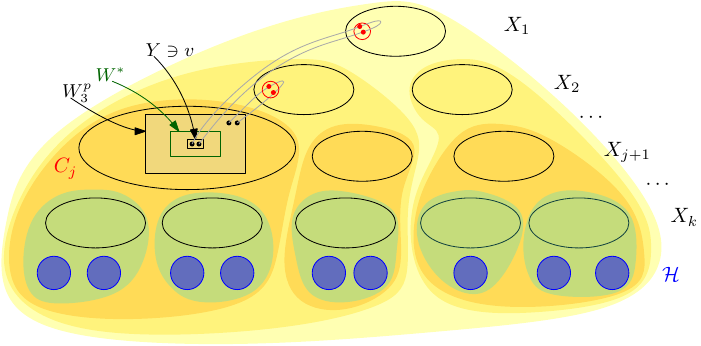}
\caption{Illustration for the correctness of \autoref{lem_edk}.}
\label{fig_bigPlED}
\end{figure}

\begin{claim}
$G^*$ has \hptd{$\Hcal^{(a-1)}$} at most $k-j$.
\end{claim}

\begin{cproof}
As discussed previously, $C_j$ is $(\Pcal^{(k-j-1)}\triangleright\Hcal)^{(a-1)}$-planar.
Given that $G^*-Y$ is an induced subgraph of $C_j$, it is also $(\Pcal^{(k-j-1)}\triangleright\Hcal)^{(a-1)}$-planar by heredity.
By \autoref{cl_fl}, $(W^*,\frak{R}'')$ is a flatness pair of $G^*$ and
$G'=\compass_{\frak{R}''}(W^*)$.
Given that $(W^*,\frak{R}^*)$ is a $\tilde{W}$-tilt of $(W_3^p,\frak{R}^p_3)$, $G'$ is an induced subgraph of $F_p$, and thus, by \autoref{cl_dp_small}, $G'$ is $(\Pcal^{(k-j-1)}\triangleright\Hcal)^{(a-1)}$-planar.
Therefore, by \autoref{lem_combinebis} applied for the graph $G^*$, $G'$, and $Y$, given that $z\ge\max\{\sqrt{(a+6)/2}+2,11\}$,
we conclude that $G^*$ is $(\Pcal^{(k-j-1)}\triangleright\Hcal)^{(a-1)}$-planar.
In particular, it means that $G^*$ has \hptd{$\Hcal^{(a-1)}$} at most $k-j$.
\end{cproof}

\begin{claim}
$G$ has \hptd{$\Hcal^{(a-1)}$} at most $k$.
\end{claim}

\begin{cproof}
Given that $G^*$ has \hptd{$\Hcal^{(a-1)}$} at most $k-j$, there exists a certifying elimination sequence $Y_1,\dots,Y_{k-j}$ .
We define $(X_1^*,\dots,X_k^*)$ and $(G_0^*,G_1^*,\dots,G_k^*)$ as follows.
For $i\in[j]$, $X_i^*:=X_i\setminus Y$.
For $i\in[k-j]$, we set $X_{j+i}^*:=X_{j+i}\setminus V(C_j)\cup Y_i$.
Finally, we set $G_0^*:=G$ and $G_i^*:=G_{i-1}^*-X_i^*$ for $i\in[k]$.
We want to prove that  prove that $X_1^*,\dots,X_k^*$ is a certifying elimination sequence of $G$, i.e. that
$\torso(G_{i-1}^*,X_i^*)$ is planar for $i\in[k]$ and that $C\in\Hcal^{(a-1)}$ for $C\in\cc(G_k^*)$.

For each $D\in\cc(G_{j}^*)$, either $D=G^*$, or $D\in\cc(G_j-Y)$.
Therefore, for each $C\in\cc(G_{k}^*)$, either (a)~$V(C)\subseteq V(G^*)$, in which case $C\in\cc(G^*-\bigcup_{i\in[k-j]}Y_i)$ and thus $C\in \Hcal^{(a-1)}$, given that $Y_1,\dots, Y_{k-j}$ is a certifying elimination sequence of $G^*$, or 
(b) there is $D\in\cc(G_j-Y)$ such that $V(C)\subseteq V(D)$, in which case $C\in\cc(G_k-Y)$ and thus again $C\in \Hcal^{(a-1)}$ by heredity of $\Hcal^{(a-1)}$.

Given that $(W^*,\frak{R}^*)$ is a flatness pair of $G-A^+$ and that $Y$ is the vertex set of some $W^{*(5)}$-tilt of $(W^*,\frak{R}^*)$, it implies that $N_G(Y)\subseteq V(G')\cup A^+$.
Additionally, by \autoref{cl_fl}, $V(G')\subseteq V(C_j)\cup Y$, and, by \autoref{cl_A+}, $A^+\subseteq N_{G}(V(C_j))$, so $N_G(Y)\subseteq N_G[V(C_j)]$.
So finally, given that $v\in Y$, we conclude that $N_G(Y)\subseteq N_{G-v}[V(C_j)]$.

For $i\in[j]$, $X_i^*=X_i\setminus Y$. $N_G(Y)\cap X_i$ already induces a clique in $\torso(G_{i-1},X_i)$ because $N_{G-v}[V(C_j)]\cap X_i=N_{G-v}(V(C_j))\cap X_i$ induces a clique in $\torso(G_{i-1},X_i)$.
Therefore, $\torso(G_{i-1}^*,X_i^*)$ is a subgraph of $\torso(G_{i-1},X_i)$, that is thus planar.
For $i\in[k-j]$, the connected components of $X_{j+i}^*$ are either connected components of $Y_i$ or connected components of $X_{j+i}$, so their torso is planar in either case.
Hence the result.
\end{cproof}
\end{proof}

\section{$\Hcal$-planar treewidth}\label{sec_tw}
In \autoref{subsec_twrend}, we prove that a graph has $\Hcal$-planar treewidth at most $k$ if and only if it has ground-maximal \sdecomp\ whose cells have property $\Pi_{\Hcal,k}$.
From this, we deduce in \autoref{subsec_small_tw} a proof of \autoref{lem_twk}.

\subsection{Expression as a sphere decomposition}\label{subsec_twrend}
\paragraph{Quasi-4-connectivity.}
Given $k\in\Nbbb_{\ge1}$, a graph $G$ is $k$-connected if for all separation $(L,R)$ of order at most $k-1$ of $G$, either $L\subseteq R$ or $R\subseteq L$.
A graph $G$ is \emph{quasi-4-connected} if it is 3-connected and that for all separation $(L,R)$ of order three of $G$, either $|L\setminus R|\le1$ or $|R\setminus L|\le1$.

\paragraph{Adhesions.}
Let $(T,\beta)$ be a tree decomposition.
The \emph{adhesion} of two nodes $t,t'\in V(T)$, denoted by $\adh(t,t')$, is the set $\beta(t)\cap\beta(t')$.
The \emph{adhesion} of $\Tcal$ is the maximum size of an adhesion over all pairs of adjacent nodes of~$T$.

\begin{proposition}[\cite{Grohe16}]\label{prop_Grohe}
Every $G$ has a tree decomposition $(T,\beta)$ of adhesion at most three such that, for each $t\in V(T)$, $\torso(G,\beta(t))$ is a minor of $G$ that is quasi-4-connected. 
\end{proposition}

\begin{lemma}\label{lem_adh3}
Let $\Hcal$ be a graph class and $k\in\Nbbb$.
Let $G$ be a graph of $\Hcal$-planar treewidth at most $k$ and $X$ be a $\Pcal\Tcal_k\triangleright\Hcal$-modulator of $G$.
Then $\torso(G,X)$ has a tree decomposition $(T,\beta)$ of adhesion at most three such that 
for each $t\in V(T)$, $\torso(G,\beta(t))$ is a minor of $\torso(G,X)$ that is quasi-4-connected, and either it is planar or it has treewidth at most $k$.
\end{lemma}
\begin{proof}
By \autoref{prop_Grohe}, there is a tree decomposition $(T,\beta)$ of $\torso(G,X)$ of adhesion at most three such that, for each $t\in V(T)$, $\torso(G,\beta(t))$ is a minor of $\torso(G,X)$ that is quasi-4-connected. 
Let $t\in V(T)$.
Given that $\torso(G,X)$ has planar treewidth at most $k$ and that $\torso(G,\beta(t))$ is a minor of $\torso(G,X)$, it implies that $\torso(G,\beta(t))$ has planar treewidth at most $k$.
Thus, there is a tree decomposition $(T^t,\beta^t)$ of $\torso(G,\beta(t))$ such that, for each $u\in V(T^t)$, $\torso(G,\beta^t(u))$ either is planar or has treewidth at most $k$.
We choose $(T^t,\beta^t)$ to have the minimum number of nodes.
This implies in particular that, for any $u,u'\in V(T^t)$, $\beta^t(u)\setminus\beta^t(u')\ne\emptyset$.

\smallskip
Suppose toward a contradiction that there are two adjacent nodes $u,u'$ in $V(T^t)$ and that one of them, say $u$, is such that $\torso(G,\beta^t(u))$ is planar. Then $|\adh(u,u')|\le 4$ and, by planarity of $\torso(G,\beta^t(u))$, there is $Y\subseteq \adh(u,u')$ that is a 3-separator in $\torso(G,\beta(t))$.
Then, by quasi-4-connectivity of $\torso(G,\beta(t))$, $|Y|=|\adh(u,u')|=3$ and either $|\beta^t(u)\setminus Y|=1$ or $|\beta^t(u')\setminus Y|=1$.

If $\beta^t(u')\setminus Y=\{v\}$, then, again by quasi-4-connectivity of $\torso(G,\beta(t))$, $Y$ induces a face of $\torso(G,\beta^t(u))$, so $\torso(G,\beta^t(u)\cup\{v\})$ is also planar, given that $N_{\beta(t)}(v)= Y$.
Therefore, the tree decomposition obtained by removing $u'$ and adding $v$ to $\beta^t(u)$ is a tree decomposition of $\torso(G,\beta(t))$ such that, for each $u\in V(T^t)$, $\torso(G,\beta^t(u))$ either is planar or has treewidth at most $k$, and with less nodes as $(T^t,\beta^t)$, a contradiction.

Assume now that $\beta^t(u)\setminus Y=\{v\}$ and $|\beta^t(u')\setminus Y|>1$. By symmetry, if $\torso(G,\beta^t(u'))$ is planar, we get a contradiction, so we assume that $\torso(G,\beta^t(u'))$ has treewidth at most $k$. But then, $\torso(G,\beta^t(u')\cup\{v\})$ also has treewidth at most $k$, so we can again construct a tree decomposition of $\torso(G,\beta(t))$ such that, for each $u\in V(T^t)$, $\torso(G,\beta^t(u))$ either is planar or has treewidth at most $k$, and with less nodes as $(T^t,\beta^t)$, a contradiction.

Therefore, for each $t\in V(T_1)$, either $\torso(G,\beta(t))$ is planar, or $\torso(G,\beta(t))$ has treewidth at most $k$.
Hence the result.
\end{proof}

\paragraph{Property $\Pi_{\Hcal,k}$.}
Let $\Hcal$ be a graph class and $k\in\Nbbb$.
Given a \sdecomp\ $\delta$ of a graph $G$, we say  that a cell $c\in C(\delta)$ has the \emph{property $\Pi_{\Hcal,k}$} if $G_c$ has a $\Pcal\Tcal_k\triangleright\Hcal$-modulator $X_c$
such that $\pi_\delta(\tilde{c})\subseteq X_c$, where $G_c$ is there graph obtained from $\sigma(c)$ by making a clique out of $\pi_\delta(\tilde{c})$.
We say that $\delta$ has the \emph{property $\Pi_{\Hcal,k}$} if each cell of $\delta$ has the property $\Pi_{\Hcal,k}$.

\begin{lemma}\label{lem_twrend1}
Let $\Hcal$ be a graph class and $k\in\Nbbb$.
Let $G$ be a graph.
Suppose that $G$ has a \sdecomp\ $\delta$ with the property $\Pi_{\Hcal,k}$.
Then $G$ has $\Hcal$-planar treewidth at most $k$.
\end{lemma}

\begin{proof}
We claim that $X:=\bigcup_{c\in C(\delta)}X_c$ is a $\Pcal\Tcal_k\triangleright\Hcal$-modulator of $G$.
Indeed, for each $c\in C(\delta)$, let $(T_c,\beta_c)$ be a tree decomposition of $\torso(G_c,X_c)$ of planar width at most $k$.
Since $\pi_\delta(\tilde{c})\subseteq X_c$ induces a clique in $G_c$, there is $t_c\in V(T_c)$ such that $\pi_\delta(\tilde{c})\subseteq \beta_c(t_c)$.
We define a tree decomposition of $\torso(G,X)$ as follows.
We set $T$ to be the union of the trees $T_c$ for $c\in C(\delta)$ and a new vertex $t$ with an edge $tt_c$ for each $c\in C(\delta)$.
Obviously, $T$ is a tree.
We set $\beta$ to be the function such that $\beta(t)=\pi_\delta(N(\delta))$ and, for $c\in C(\delta)$, $\beta|_{V(T_c)}=\beta_c$.
$(T,\beta)$ is a tree decomposition of $\torso(G,X)$ such that $\torso(G,\beta(t))$ is planar, by definition of a \sdecomp.
Hence the result.
\end{proof}

\begin{lemma}\label{lem_twrend2}
Let $\Hcal$ be a hereditary graph class and let $a,k,r\in\Nbbb$ with $r\ge\max\{a+3,k+1,7\}$.
Let $G$ be an $(a,3)$-unbreakable graph and 
$W$ be an $r$-wall of $G$.
Suppose that $G$ has $\Hcal^{(a-1)}$-planar treewidth at most $k$.
Then $G$ has a 
\sdecomp\ $\delta$ with the property $\Pi_{\Hcal,k}$ such that the $(r-2)$-central wall $W'$ of $W$ is grounded in $\delta$ and such that each cell of $\delta$ has size at most $a-1$.
\end{lemma}

\begin{proof}
Let $X$ be a $\Pcal\Tcal_k\triangleright\Hcal^{(a-1)}$-modulator of $G$.
Let $V'\subseteq V(W')$ be the branch vertices of $W$ that are vertices of $W'$.

By \autoref{lem_adh3},
$\torso(G,X)$ has a tree decomposition $(T,\beta)$ of adhesion at most three such that 
for each $t\in V(T)$, either $\torso(G,\beta(t))$ is planar and quasi-4-connected, or $\tw(\torso(G,\beta(t)))\le k$.
By \autoref{obs_smallsepwall}, for each $tt'\in E(T)$, there is one connected component $D\in\cc(G-\adh(t,t'))$
that contains all but at most one vertex of $V'\setminus \adh(t,t')$, and $G- V(D)$ contains no cycle of $W'$.
Given that $G$ is $(a,3)$-unbreakable, that $|\adh(t,t')|\le 3$ and that $|V(D)|\ge|V'\setminus \adh(t,t')|\ge a$, we conclude that $|V(G)\setminus V(D)|<a$.
Therefore, there is $t\in V(T)$ such that, for each $t'\in V(T)$ adjacent to $t$, 
the connected component $D_{t'}$ containing $\beta(t)\setminus\adh(t,t')$ is such that $G- V(D_{t'})$ contains no cycle of $W'$ and $|V(G)\setminus V(D_{t'})|<a$.

\begin{claim}
$\torso(G,\beta(t))$ is planar.
\end{claim}

\begin{cproof}
Suppose toward a contradiction that $\tw(\torso(G,\beta(t)))\le k$ and let $(T_0,\beta_0)$ be a tree decomposition of $\torso(G,\beta(t))$ of width at most $k$.
For each $tt'\in E(T)$, there is $u_{t'}\in V(T_0)$ such that $\adh(t,t')\subseteq \beta_0(u_{t'})$.
Let $\Fcal$ be the set of $H\in \cc(G-X)$ such that there is no $tt'\in E(T)$ such that $V(H)\subseteq V(D_{t'})$.
Again, for each $H\in\Fcal$, there is $u_{H}\in V(T_0)$ such that $N_G(H)\subseteq \beta(u_H)$.
Hence, we can define the tree decomposition $(T',\beta')$ of $G$ where $T'$ is obtained from $T_0$ by adding a vertex $v_{t'}$ adjacent to $u_{t'}$  for each $tt'\in E(T)$ and a vertex $v_H$ adjacent to $u_H$ for each $H\in\Fcal$,
and $\beta'$ defined by $\beta'|_{V(T_0)}=\beta_0$, $\beta'(v_{t'})=V(D_{t'})\cup\adh(t,t')$ for $tt'\in E(T)$, and $\beta'(v_H)=N_G[H]$ for $H\in\Fcal$.
Given that $|\beta'(v_{t'})|\le a+2$ for $tt'\in E(T)$, that $|\beta'(v_H)|\le a+3$ for $H\in\Fcal$, and that $|\beta'(u)|\le k+1$ for $u\in V(T_0)$,
it implies that $G$ has treewidth at most $\max\{a+2,k\}$.
This contradicts the fact that $\tw(G)\ge\tw(W)=r\ge 1+\max\{a+2,k\}$.
\end{cproof}

We set $S:=\beta(t)$.
Let $\delta=(\Gamma,\Dcal)$ be a sphere embedding of $\torso(G,S)$.
Remember that $(T,\beta)$ has adhesion at most three, so, for $D\in\cc(G-S)$, $|N_G(V(D))|\le 3$.
As in \autoref{cl_compat}, for each $D\in\cc(G-S)$, there is a $\delta$-aligned disk $\Delta_D$ such that:
\begin{itemize}
\item $N_G(V(D))=\pi_\delta(N(\delta)\cap\Delta_D)$ and
\item the graph induced by $B_D:=V(\inG_\delta(\Delta_D)\cap V(Z_D))$ contains no cycle of $W'$.
\end{itemize}
Note that to prove this, we use the fact that $r-2\ge 5$ and $|V'|\ge2(r-2)^2-2\ge a+4$.
Also, if $N_G(V(D))=N_G(V(D'))$, then we can assume that $\Delta_D=\Delta_{D'}$.

Let $\Dcal^*$ be the inclusion-wise maximal elements of $\Dcal\cup\{\Delta_D\mid D\in\cc(G-S)\}$.
By maximality of $\Dcal^*$ and planarity of $\torso(G,S)$, any two distinct $\Delta_D,\Delta_{D'}\in\Dcal^*$ may only intersect on their boundary.
For each $C\in\cc(G-S)$, we draw $C$ in a $\Delta_D\in\Dcal$ such that $V(C)\subseteq B_D$, and add the appropriate edges with $\pi_\delta(N(\delta)\cap \bd(\Delta_D))$.
We similarly draw the edges of $G[S]$ to obtain a drawing $\Gamma^*$ of $G$.

For each cell $c$ of $\delta^*=(\Gamma^*,\Dcal^*)$, there is $D\in\cc(G-S)$ such that $\sigma_{\delta^*}(c)$ contains no cycle of $W'$, since $\sigma_{\delta^*}(c)$ is a subgraph of $G[B_D]$, so $W'$ is grounded in $\delta^*$.

It remains to prove that $\delta^*$ has property $\Pi_{\Hcal,k}$.
Let $c\in C(\delta^*)$ and $X_c:=X\cap V(\sigma_{\delta^*}(c))$.
Let $(T_c,\beta_c)$ be the restriction of $(T,\beta)$ to $X_c$.
Either $\sigma_{\delta^*}(c)$ is an edge of $G$, or there is $D\in\cc(G-S)$ such that $\pi_{\delta^*}(\tilde{c})=N_G(V(D))$. Given that $(T,\beta)$ is a tree decomposition of $\torso(G,X)$ and that $N_G(V(D))$ is a clique in $\torso(G,X)$ for $D\in\cc(G-S)$, we conclude that $(T_c,\beta_c)$ is a tree decomposition of $G_c$ where $G_c$ is the graph obtained from $\sigma_{\delta^*}(c)$ by making a clique out of $\pi_{\delta^*}(\tilde{c})$.
Additionally, given that $S\subseteq X$, we have $\pi_{\delta^*}(\tilde{c})\subseteq X_c$.
Finally, by heredity, $(T_c,\beta_c)$ is a tree decomposition of $\Hcal$-planar treewidth at most $k$.
Hence the result.
\end{proof}

\begin{lemma}\label{lem_twrend3}
Let $\Hcal$ be a hereditary graph class and $k\in\Nbbb$.
Let $G$ be a graph.
If $G$ has a
\sdecomp\ $\delta$ with the property $\Pi_{\Hcal,k}$ that is not ground-maximal, then
$G$ has
\sdecomp\ $\delta'$ with the property $\Pi_{\Hcal,k}$ that is strictly more grounded than $\delta$.
\end{lemma}

\begin{proof}
Let $\delta=(\Gamma,\Dcal)$ be a \sdecomp\ of $G$ with the property $\Pi_{\Hcal,k}$ that is not ground-maximal.
Hence, there is a cell $c\in C(\delta)$ that is not ground-maximal.
Since $c$ has the property $\Pi_{\Hcal,k}$, there is a $\Pcal\Tcal_k\triangleright\Hcal$-modulator $X_c$ of $\torso(G_c,\sigma(c))$ such that $\pi_\delta(\tilde{c})\subseteq X_c$.
By \autoref{lem_adh3}, there is a tree decomposition $(T,\beta)$ of $\torso(G_c,X_c)$ of adhesion at most three such that, for each $t\in V(T)$, $\torso(G_c,\beta(t))$ is a minor of $\sigma(c)$ that is quasi-4-connected, and either it is planar or it has treewidth at most $k$.
Without loss of generality, we can assume that, for each $t,t'\in V(T)$, $\beta(t)\setminus\beta(t')\ne\emptyset$.
Given that $\pi_\delta(\tilde{c})$ is a clique in $\torso(G_c,\sigma(c))$, there is $r\in V(T)$ such that $\pi_\delta(\tilde{c})\subseteq \beta(r)$.
We root $T$ at $r$.

Suppose that there is a child $t$ of $r$ such that $\adh(r,t)\subsetneq\pi_\delta(\tilde{c})$.
Let $\Tcal$ be the set of all children of $r$ such that $\adh(r,t')=\adh(r,t)$. Let $V'$ be the set of all vertices of $G- \beta(r)$ that belong to a bag of a subtree of $T$ rooted at a node in $\Tcal$.
Let $\delta'$ be the \sdecomp\ of $G$ obtained by removing $c$ from $\delta$ and instead adding two cells $c_1$ and $c_2$ such that $\pi_{\delta'}(\tilde{c}_1)=\adh(t,r)$, $\sigma_{\delta'}(c_1)=G[V'\cup\adh(r,t')]$, $\pi_{\delta'}(\tilde{c}_2)=\pi_{\delta}(\tilde{c})$, and $\sigma_{\delta'}(c_1)=\sigma_{\delta}(c)-V'$.
Note that, for $i\in[2]$, $X_{c_i}:=X_c\cap\sigma_{\delta'}(c_i)$ is a $\Pcal\Tcal_k\triangleright\Hcal$-modulator of $\torso(G_c,\sigma_{\delta'}(c_i))$.
Thus, $\delta'$ is more grounded than $\delta$ and has the property $\Pi_{\Hcal,k}$.

Suppose that $\beta(r)=\pi_\delta(\tilde{c})$.
If there is a child $t$ of $r$, then we have $\adh(r,t)\subseteq\beta(r)=\pi_\delta(\tilde{c})$, and so $\adh(r,t)=\beta(r)$ given that we already handled the case when $\adh(r,t)\subsetneq\pi_\delta(\tilde{c})$.
Therefore, $\beta(r)\subseteq \beta(t)$,
which contradicts the fact that $\beta(r)\setminus\beta(t)\ne\emptyset$.
Thus, $r$ is the only node of $T$.
Then $c$ is $\Hcal$-compatible, so by \autoref{lem_compatible}, for any \sdecomp\ $\delta'$ that is more grounded than $\delta$, the cells of $\delta'$ contained in $c$ are $\Hcal$-compatible and thus have property $\Pi_{\Hcal,k}$.
So we can assume that $\beta(r)\setminus\pi_\delta(\tilde{c})\ne\emptyset.$

Suppose that $|\tilde{c}|\le2$.
Let $v\in\beta(r)\setminus\pi_\delta(\tilde{c})$ be a vertex adjacent to a vertex of $\pi_\delta(\tilde{c})$ in $\torso(G_c,\beta(r))$.
It exists given that $\torso(G_c,\beta(r))$ is quasi-4-connected, and thus connected.
Let $\delta'$ be the the \sdecomp\ of $G$ obtained by adding a new point $u$ in the boundary $\tilde{c}$ of $c$, and setting $\pi_{\delta'}(u)=v$. 
We still have $\pi_{\delta'}(\tilde{c})\subseteq X_c$, and $\delta'$ is a \sdecomp\ of $G$ with property $\Pi_{\Hcal,k}$ that is more grounded than $\delta$.

We can thus suppose that $|\tilde{c}|=3$.
Note that, if there is a child $t$ of $r$ such that $\pi_{\delta}(\tilde{c})=\adh(r,t)$, then, given that $\beta(r)\setminus\adh(r,t)\ne\emptyset$ and $\beta(t)\setminus\adh(r,t)\ne\emptyset$, it implies that $c$ is already ground-maximal, so we can assume that $\pi_{\delta}(\tilde{c})\ne\adh(r,t)$ for all $t\in V(T)$.
If $\torso(G_c,\beta(r))$ is planar, given that is is also quasi-4-connected, $\pi_\delta(\tilde{c})$ either induces a face of $\torso(G_c,\beta(r))$, or there is a vertex $v$ such that $(\pi_\delta(\tilde{c})\cup\{v\},V(\torso(G_c,\beta(r)))\setminus\{v\})$ is a separation of $\torso(G_c,\beta(r))$. In the second case, then again, $c$ is already ground-maximal.
In the first case, then let $\delta_c$ be a sphere embedding of $\torso(G_c,\beta(r))$ whose outer face has vertex set $\pi_\delta(\tilde{c})$.
Then the \sdecomp\ obtained by replacing $c$ with $\delta_c$ still has property $\Pi_{\Hcal,k}$, and it is more grounded than $\delta$.

We now suppose that $\torso(G_c,\beta(r))$ has treewidth at most $k$.
Let $\delta'$ be a \sdecomp\ of $G$ more grounded than $\delta$ such that each cell $c'\in C(\delta)\setminus\{c\}$ is equivalent to a cell of $C(\delta')$.
We choose $\delta'$ such that the number of ground vertices $|N(\delta')|$ is minimum among all \sdecomps\ of $G$ that are more grounded than $\delta$ and distinct from $\delta$.
Suppose toward a contradiction that there is a cell $c'\in C(\delta')$ contained in $c$ and $v\in\pi_{\delta'}(\tilde{c}')$ such that $v\notin\beta(r)$.
Let $t\in V(T)$ be the child of $r$ whose subtree contains a node $t'$ with $v\in\beta(t')$.
Thus, $\adh(r,t)$ is a separator of size at most three between $v$ and $\pi_\delta(\tilde{c})$.

If $\adh(r,t)\subseteq \pi_{\delta'}(N(\delta'))$, then there is $\delta'$-aligned disk $\Delta$ whose boundary is $\adh(r,t)$.
Let $\delta''$ be the \sdecomp\ of $G$ obtained by removing the cells of $\delta'$ in $\Delta$ and adding instead a unique cell with boundary $\adh(r,t)$. $\delta''$ is more grounded than $\delta$ and, given that $\adh(r,t)\ne\pi_{\delta}(\tilde{c})$, $\delta''$ is distinct from $\delta$, a contradiction to the minimality of $\delta'$.

Otherwise, there is $x\in\adh(r,t)\setminus\pi_{\delta'}(N(\delta'))$.
Given that $\torso(G,\beta(r))$ is quasi-4-connected, there are three internally vertex-disjoint paths in $\torso(G_c,\beta(r))$ from $x$ to the three vertices in $\pi_{\delta}(\tilde{c})$.
Additionally, there is a path from $x$ to $v$ disjoint from $\beta(r)$ (aside from its endpoint $x$).
Thus, given that $\torso(G,\beta(r))$ is a minor of $\sigma(c)$, that are four internally vertex-disjoint paths in $\sigma(c)$ from $x$ to $v$ and the three vertices in $\pi_{\delta}(\tilde{c})$, respectively.
This implies that $x\in\sigma_{\delta'}(c'')$ for some cell $c''$ with $|\tilde{c}''|\ge4$, a contradiction to the fact that $|\tilde{c}''|\le3$.

Therefore, for any $c'\in C(\delta')$, $\pi_{\delta'}(\tilde{c}')\subseteq \beta(r)$.
Given that $\torso(G_c,\beta(r))$ has treewidth at most $k$, by heredity of the treewidth, we easily get that $X_{c'}:=X_c\cap \sigma_{\delta'}(c')$ is a $\Pcal\Tcal_k\triangleright\Hcal$-modulator of $\torso(G_c,\sigma_{\delta'}(c'))$.
Hence the result.
\end{proof}

From the previous results of this section and \autoref{lem_newirr}, we deduce an irrelevant vertex technique tailored for $\Hcal$-planar treewidth.

\begin{lemma}\label{lem_combinetw}
Let $\Hcal$ be a hereditary graph class.
Let $a,k,r\in\Nbbb$ with odd $r\ge\max\{a+3,k+1,7\}$.
Let $G$ be an $(a,3)$-unbreakable graph, $(W,\frak{R})$ be a flatness pair of $G$ of height $r$, $G'$ be the $\frak{R}$-compass of $W$, and $v$ be a central vertex of $W$.
Then $G$ has $\Hcal^{(a-1)}$-planar treewidth at most $k$ if and only if $G'$ and $G-v$ both have $\Hcal^{(a-1)}$-planar treewidth at most $k$.
\end{lemma}

\begin{proof}
Obviously, if one of $G'$ and $G-v$ is has $\Hcal^{(a-1)}$-planar treewidth at least $k+1$, then so does $G$ by heredity of the $\Hcal^{(a-1)}$-planar treewidth.
Let us suppose that both $G'$ and $G-v$ have $\Hcal^{(a-1)}$-planar treewidth at most $k$.
We want to prove that $G$ has $\Hcal^{(a-1)}$-planar treewidth at most $k$. 
For this, we find a well-linked  rendition $\delta$ of $(G',\Omega)$ and ground-maximal sphere decompositions $\delta'$ of $G'$ and $\delta_v$ of $G-v$. 

Let $C_i$ be the $i$-th layer of $W$ for $i\in[(r-1)/2]$ (so $C_1$ is the perimeter of $W$).
Let $\Rcal=(X,Y,P,C,\delta)$ and $\Omega$ be the cyclic ordering of the vertices of $X\cap Y$ as they appear in $C_1$.
Hence, $\delta=(\Gamma,\Dcal)$ is a rendition of $(G'=G[Y],\Omega)$.
By \autoref{lem_rend_to_min_or_max}, we can assume $\delta$ to be well-linked.

By \autoref{lem_twrend2} and \autoref{lem_twrend3}, given that $r\ge\max\{a+3,k+1,7\}$, there are two ground-maximal \sdecomps\ $\delta'=(\Gamma',\Dcal')$ and $\delta_v=(\Gamma_v,\Dcal_v)$ of $G'$ and $G-v$, respectively, that have property $\Pi_{\Hcal,k}$ and 
such that the $(r-2)$-central wall $W'$ of $W$ is grounded in both $\delta'$ and $\delta_v$. 

Then, by \autoref{lem_newirr}, there exists a ground-maximal \sdecomp\ $\delta^*$ of $G$ such that each cell $\delta^*$ is either a cell of $\delta'$ or a cell of $\delta_v$, and is thus has property $\Pi_{\Hcal,k}$.
Thus, by \autoref{lem_twrend1}, $G$ has $\Hcal^{(a-1)}$-planar treewidth at most $k$.
\end{proof}

\subsection{The algorithm}\label{subsec_small_tw}

In this section, we prove \autoref{lem_twk}.\medskip

Let us first remark that we can use Courcelle's theorem (\autoref{Courcelle}) to check whether a graph of bounded treewidth as $\Hcal$-planar treewidth at most $k$.

\begin{lemma}\label{lem_twCMSO}
If $\Hcal$ is a CMSO-definable graph class, then having $\Hcal$-planar treewidth at most $k$ is expressible in CMSO logic.
\end{lemma}

\begin{proof}
The class $\Pcal\Tcal_k$ of graphs with planar treewidth at most $k$ is minor-closed.
Therefore, by Robertson and Seymour's seminal result, the number of obstructions of $\Pcal\Tcal_k$ is finite.

As mentioned in \autoref{obs_CMSO}, $\torso(G,X)$ is expressible in CMSO logic.
Moreover, as proven in \cite[Subsection 1.3.1]{CourcelleJ12grap}, the fact that a graph $H$ is a minor of a graph $G$ is expressible in CMSO logic, which implies that $\Pcal\Tcal_k$ is expressible in CMSO logic. Indeed, it suffice to check for a graph $G$ whether it contains or not the obstructions of $\Pcal\Tcal_k$.
\end{proof}

\begin{proof}[Proof of \autoref{lem_twk}]
We apply \autoref{prop_flatwallth} to $G$,
with $b=\lceil\sqrt{a+3}\rceil+2$ and $r=\odd(\max\{a+3,k+1,7\})$.
It runs in time $\Ocal_{k,a}(n+m)$.

If $G$ has treewidth at most $f_{\ref{prop_flatwallth}}(b)\cdot r$, then we apply \autoref{Courcelle} to $G$ in time $\Ocal_{k,a}(n)$ and solve the problem. We can do so because 
the graphs in $\Hcal^{(a-1)}$ have a bounded size, so $\Hcal^{(a-1)}$ is a finite graph class, hence trivially CMSO-definable.
Therefore, by \autoref{lem_twCMSO}, having $\Hcal$-planar treewidth at most $k$ is expressible in CMSO logic.

If $G$ contains an apex grid of height $b$ as a minor, then, given that $b\ge\sqrt{a+3}+2$, by \autoref{lem:obstructions}, we obtain that $G$ has no planar $\Gcal^{(a-1)}$-modulator, where $\Gcal$ is the class of all graphs.
Hence, by \autoref{obs_sol_compatible}, $G$ has no sphere decomposition $\delta$ whose cells all have size at most $a-1$.
This implies, by \autoref{lem_twrend2}, that $G$ has $\Hcal^{(a-1)}$-planar treewidth at least $k+1$.
Therefore, we report a \no-instance.

Hence, we can assume that there is a flatness pair $(W,\frak{R})$ of height $r$ in $G$ whose $\frak{R}$-compass $G'$ has treewidth at most $f_{\ref{prop_flatwallth}}(b)\cdot r$.
Let $v$ be a central vertex of $W$.
We apply \autoref{Courcelle} to $G'$ in time $\Ocal_{k,a}(n)$ and 
we recursively apply our algorithm to $G-v$. 
If the outcome is a \no-instance for one of them, then this is also a \no-instance for $G$.
Otherwise, the outcome is a \yes-instance for both.
Then, by \autoref{lem_combinetw}, given that $r\ge\max\{a+3,k+1,7\}$, we can return a \yes-instance.

The running time of the algorithm is $T(n)=\Ocal_{k,a}(n+m)+T(n-1)=\Ocal_{k,a}(n(n+m))$.
\end{proof}

\section{Applications}\label{sec_applications}
This section contains several algorithmic applications of  \autoref{th_th}, as well as a discussion on similar applications for \autoref{th_param}.
To combine nice algorithmic properties of planar graphs and graphs from $\Hcal$, 
  we need a variant of  \autoref{th_th} that guarantees the existence of a polynomial-time algorithm computing a planar $\Hcal$-modulator  in an $\Hcal$-planar graph. (Let us remind that 
\autoref{th_th} only guarantees the existence of a polynomial-time algorithm deciding whether an input graph is $\Hcal$-planar.)  Thus, we start with the proof of the following theorem. The proof of the theorem combines self-reduction arguments with the hereditary properties of $\Hcal$.

 \begin{restatable}{theorem}{selfreduce} 
 \label{cor_modulator}
Let $\Hcal$ be a hereditary, CMSO-definable, and polynomial-time decidable graph class. 
Then there exists a polynomial-time algorithm constructing for a given $\mathcal{H}$-planar graph $G$ a planar $\Hcal$-modulator.
\end{restatable}

\begin{proof}
 If $\Hcal$ is trivial in the sense that $\Hcal$ includes every graph, then we choose a planar $\Hcal$-modulator $X=\emptyset$ of $G$.  Otherwise, because $\Hcal$ is hereditary, there is a graph $F$ of the minimum size such that $F\notin\Hcal$. We say that $F$ is a \emph{minimum forbidden subgraph}. Because the inclusion in $\Hcal$ is decidable, $F$ can be found in constant time by brute force checking all graphs of size $1,2,\ldots$ where the constant depends on $\Hcal$.

We consider two cases depending on whether $F$ is connected or not.

Assume that $F$ is connected. We pick an arbitrary vertex $r\in V(F)$ and declare it to be the \emph{root} of $F$. We construct five copies $F_1,\ldots,F_5$ of $F$ rooted in $r_1,\ldots,r_5$, respectively. 
Then we define $F'$ to be the graph obtained from $F_1,\ldots,F_5$ by identifying their root into a single vertex $r$ which is defined to be the root of $F'$. Our self-reduction arguments are based on the following claim.

\begin{claim}\label{cl_padding-one}
Let $v$ be a vertex of a graph $G$. Then $G$ has a planar $\Hcal$-modulator $X$ with $v\in X$ if and only if the graph $G_v$ obtained from $G$ and $F'$ by identifying $v$ and the root of $F'$ has a planar $\Hcal$-modulator. 
\end{claim}

\begin{cproof}
Suppose that $X$ is a planar $\Hcal$-modulator of $G$ with $v\in X$. Then each connected component of $G_v-X$ is either a connected component of $G-X$ or a connected component of one of the graphs $F_1-r_1,\ldots,F_5-r_5$.
Because $F_i-r_i\in\Hcal$ for $i\in[5]$, we have that each connected component of $G_v-X$ is in $\Hcal$. Furthermore, the torsos of $X$ with respect to $G$ and $G_v$ are the same because the root $r$ of $F'$ is the unique vertex of $F'$ in $X$. Therefore, $X$ is a planar $\Hcal$-modulator of $G_v$.

Assume that $G_v$ has a planar $\Hcal$-modulator $X$. We show that $r=v\in X$. For the sake of contradiction, assume that $v\notin X$. Then because $F_1,\ldots,F_5\notin \Hcal$, for each $i\in[5]$, there is $x_i\in V(F_i)$ distinct from $r_i$ such that $x_i\in X$. We pick each $x_i$ to be a vertex in minimum distance from $r_i$ in $F_i$. This means that each $F_i$ contains a $(x_i,r_i)$-path $P_i$ such that each vertex of  $P_i-x_i$ is not in $X$.  
Then the vertices $\bigcup_{i=1}^5(V(P_i)\setminus \{x_i\})$ are in the same connected component of $G_v-X$. However, this means that $\{x_1,\ldots,x_5\}$ is a clique of size five in the torso of $X$ contradicting planarity. Thus, $v\in X$.   
Because $v\in X$ and $v=r$ is the unique common vertex of $G$ and $F'$ in $G_v$, we have that $X'=X\setminus (V(F')\setminus\{r\})$ is a  planar $\Hcal$-modulator of $G_v$. This completes the proof of the claim.
\end{cproof}

By~\autoref{cl_padding-one}, 
we can use self-reduction. First, we apply \autoref{th_th} to check in polynomial time whether $G$ is a yes-instance of {\sc $\Hcal$-planarity}. If not,  then there is no planar $\Hcal$-modulator. Otherwise,
denote by $v_1,\ldots,v_n$ the vertices of $G$, set $G_0=G$, and set $X:=\emptyset$ initially. Then, for each $i:=1,\ldots,n$, we do the following:
\begin{itemize}
\item set $G'$ to be the graph obtained from $G_{i-1}$ and $F'$ by identifying $v_i$ and $r$,
\item run the algorithm for  {\sc $\Hcal$-planarity} on $G'$,
\item if the algorithm returns a yes-answer then set $G_i=G'$ and $X:=X\cup\{v_i\}$, and set $G_i=G_{i-1}$, otherwise. 
\end{itemize}
\autoref{cl_padding-one} immediately implies that $X$ is a planar $\Hcal$-modulator for $G$. Clearly, if {\sc $\Hcal$-planarity} can be solved in polynomial time then the above procedure is polynomial. This concludes the proof for the case when $F$ is connected.

 Now, $ F $ is not connected. Our arguments are very similar to the connected case, and therefore, we only sketch the proof. Again, we construct five copies $F_1,\ldots,F_5$ of $F$. Then, we construct a new root vertex $r$ and make it adjacent to every vertex of the copies of $F$.  The following observation is in order.
 
 \begin{claim}\label{cl_padding-two}
Let $v$ be a vertex of a graph $G$. Then $G$ has a planar $\Hcal$-modulator $X$ with $v\in X$ if and only if the graph $G_v$ obtained from $G$ and $F'$ by identifying $v$ and the root of $F'$ has a planar $\Hcal$-modulator. 
\end{claim}

\begin{cproof}
Suppose that $X$ is a planar $\Hcal$-modulator of $G$ with $v\in X$. Then each connected component of $G_v-X$ is either a connected component of $G-X$ or a connected component of one of the graphs $F_1,\ldots,F_5$. 
Because each connected component of $F_i$ is in $\Hcal$ for $i\in[5]$ by the minimality of $F$, we have that each connected component of $G_v-X$ is in $\Hcal$. Also, 
 the torsos of $X$ with respect to $G$ and $G_v$ are the same because $r$ is the unique vertex of $F'$ in $X$. Thus, $X$ is a planar $\Hcal$-modulator of $G_v$.

Assume that $G_v$ has a planar $\Hcal$-modulator $X$. We show that $r=v\in X$. For the sake of contradiction, assume that $v\notin X$. Then because $F_1,\ldots,F_5\notin \Hcal$, for each $i\in[5]$, there is $x_i\in F_i$ distinct from $r_i$ such that $x_i\in X$. Since each $x_i$ is adjacent to $r\notin X$,  $\{x_1,\ldots,x_5\}$ is a clique of size five in the torso of $X$ contradicting planarity. Thus, $v\in X$.   
Because $v\in X$ and $v=r$ is the unique common vertex of $G$ and $F'$ in $G_v$,  $X'=X\setminus \bigcup_{i=1}^5V(F_i)$ is a  planar $\Hcal$-modulator of $G_v$. This completes the proof of the claim.
\end{cproof}

\autoref{cl_padding-two} immediately implies that we can apply the same self-reduction procedure as for the connected case. This concludes the proof.
\end{proof}

We remark that similar arguments can be used to construct decompositions for graphs with $\Hcal$-planar treewidth of treedepth at most $k$. We sketch the algorithms in the following corollaries.

\begin{corollary}\label{cor:twmodulator}
Let $\Hcal$ be a hereditary, CMSO-definable, polynomial-time decidable, and union-closed graph class.
Suppose that there is an \FPT algorithm solving \pbdel parameterized by the solution size.  
Then there exists a polynomial-time algorithm constructing, for a given graph $G$ with $\Hcal$-planar treewidth at most $k$,
a $\Pcal\Tcal_k\triangleright\Hcal$-modulator $S$ of a graph $G$ and a tree decomposition $\Tcal$ of $\torso_G(S)$ of planar width at most $k$.
\end{corollary}

\begin{proof}[Sketch of the proof]
We follow the proof of \cref{cor_modulator}. The problem is trivial if $\Hcal$ includes every graph.  Otherwise, because $\Hcal$ is hereditary, there is a minimum forbidden  subgraph $F$. We use $F$ to identify a $\Pcal\Tcal_k\triangleright\Hcal$-modulator. We have two cases depending on whether $F$ is connected or not. In this sketch, we consider only the connected case; the second case is analyzed in the same way as in the proof of  \cref{cor_modulator}.  We pick an arbitrary vertex $r\in V(F)$ and declare it to be the \emph{root} of $F$. We construct $h=\max\{5,k+2\}$ copies $F_1,\ldots,F_h$ of $F$ rooted in $r_1,\ldots,r_h$, respectively. 
Then we define $F'$ to be the graph obtained from $F_1,\ldots,F_h$ by identifying their root into a single vertex $r$ which is defined to be the root of $F'$. Then similarly to~\Cref{cl_padding-one}, we have the following property for every $v\in V(G)$: $G$ has 
a $\Pcal\Tcal_k\triangleright\Hcal$-modulator $S$ with $v\in S$ if and only if the graph $G_v$ obtained from $G$ and $F'$ by identifying $v$ and the root of $F'$ has a $\Pcal\Tcal_k\triangleright\Hcal$-modulator $S$.  Thus, we can find $S$ by using self-reduction by calling the algorithm from~\cref{th_param}. 

In the next step, we again use self-reducibility to construct a tree decomposition $\Tcal$ of $G'=\torso_G(S)$ of planar width at most $k$. For this, we consider all pairs $\{u,v\}$ of non-adjacent vertices of $G'$. For each $\{u,v\}$, we add the edge $uv$ to the considered graph, and then use the algorithm from~\cref{th_param} to check whether the obtained graphs 
admits  a tree decomposition of planar width at most $k$ (in this case, $\Hcal$ contains the unique empty graph). If yes, we keep the edge $uv,$ and we discard the pair $\{u,v\}$, otherwise. Let $G^*$ be the graph obtained as the result of this procedure.  We have that $G^*$ has a tree decomposition of planar width at most $k$ where every bag of size at most $k+1$ is a clique. Furthermore, the adhesion sets for every bag of size at least $k+2$, whose torso is planar, are cliques of size at most 4.  Then we can find all such bags by decomposing $G^*$ via clique-separators of size at most 4. After that, we can find the remaining bags using the fact that they are cliques. This completes the sketch of the proof.
\end{proof}

\begin{corollary}\label{cor:tdseq}
Let $\Hcal$ be a hereditary, CMSO-definable, polynomial-time decidable, and union-closed graph class.
Suppose that there is an \FPT algorithm solving \pbdel parameterized by the solution size.  
Then there exists a polynomial-time algorithm constructing, for a given graph $G$ with $\Hcal$-planar treedepth at most $k$,
a certifying elimination sequence $X_1,\dots,X_k$.
\end{corollary}

\begin{proof}[Sketch of the proof]
We again follow the proof of \cref{cor_modulator} and use self-reduction. Assume that $\Hcal$ does not include every graph and let $F$ be a rooted minimum forbidden  subgraph. Assume that $F$ is connected; the disconnected case is analyzed similarly to the proof of  \cref{cor_modulator}. We construct the graph $F'$ as follows:
\begin{itemize}
\item construct a copy of the complete graph $K_{4k}$.
\item for every vertex $v$ of the complete graph, construct five copies $F_1,\ldots F_5$ of $F$ rooted in $r_1,\ldots,r_5$, respectively, and then identify $r_1,\ldots,r_5$ with $v$.
\end{itemize} 
Then we declare an arbitrary  vertex $r$ of the complete graph in $F'$ to be its root. Then similarly to 
to~\Cref{cl_padding-one}, we have the following property for every $v\in V(G)$: $G$ admits 
a certifying elimination sequence $X_1,\dots,X_k$ with $v\in X_1$ if and only if the graph $G'$ obtained from $G$ and $F'$ by identifying $v$ and the root of $F'$ has $\Hcal$-planar treedepth at most $k$.
We use this observation to identify $X_1$. Then we find $X_2,\ldots,X_k$ by using the same arguments inductively. This concludes the sketch of the proof.
\end{proof}

Now, we can proceed to algorithmic applications of~\autoref{cor_modulator}, \cref{cor:twmodulator}, and \cref{cor:tdseq}.

\subsection{Colourings}
Our first algorithmic application is the existence of a polynomial-time additive-approximation algorithm for graph coloring on $\Hcal$-planar graphs. 

\coloring*

\begin{proof} By \Cref{cor_modulator}, there is a polynomial-time algorithm computing a planar $\Hcal$-modulator
 $S$ of  an  $\Hcal$-planar graph $G$.
Given that $G[S]$ is planar, there is a proper coloring of $G[S]$ with colors $[1,4]$ by the Four Color theorem~\cite{AppelH89,RobertsonSST96effic}. Moreover, the proof of the Four Color theorem is constructive and yields a polynomial-time algorithm producing a 4-coloring of a planar graph. 
For each component $C\in\cc(G-S)$, by the assumptions of the theorem, we can compute $\chi(G[C])\leq \chi(G)$ in polynomial time. Since all components are disjoint, we use at most $\chi(G)$ to color all the components of $ \cc(G-S)$ and then additional four colors to color $G[S]$. 
This gives a proper $(\chi(G)+4)$-coloring of $G$.
\end{proof}

For a graph with bounded \hptd{$\Hcal$}, we can use \cref{cor:tdseq} to find a certifying elimination sequence.
Then by repetitive applications of \autoref{lemma_coloring}, we immediately obtain 
a bound on the chromatic number. 
\begin{corollary}
Let $\Hcal$ be a hereditary, CMSO-definable, and polynomial-time decidable graph class.
Moreover, assume that there is a polynomial-time algorithm computing the chromatic number $\chi(H)$ for $H\in\Hcal$.
Then, there exists a polynomial-time algorithm that, given a graph $G$ with $\Hcal\mbox{-}\ptd(G)\leq k$,  
produces a proper coloring of $G$ using at most $\chi(G)+4k$ colors.
\end{corollary}

We can also prove a similar result for graphs of bounded $\Hcal$-planar treewidth.

\begin{theorem}
Let $\Hcal$ be a hereditary, CMSO-definable, and polynomial-time decidable graph class.
Moreover, assume that there is a polynomial-time algorithm computing the chromatic number $\chi(H)$ for $H\in\Hcal$.
Then, there exists a polynomial-time algorithm that, given a graph $G$ with $\Hcal\mbox{-}\ptw(G)\leq k$, 
produces a proper coloring of $G$ using at most $\chi(G)+\max\{4,k+1\}$ colors.
\end{theorem}

\begin{proof}
Given a graph $G$, we use the algorithm from \cref{cor:twmodulator} to find a $\Pcal\Tcal_k\triangleright\Hcal$-modulator $S$ of $G$ with a tree decomposition $\Tcal$ of $\torso_G(S)$ of planar width at most $k$.
As in the proof of \autoref{lemma_coloring}, for each component $C\in\cc(G-S)$, by the assumptions of the theorem, we can compute $\chi(G[C])\leq \chi(G)$ in polynomial time.
Let us now find a proper coloring of $\torso_G(S)$ with $\max\{4,k+1\}$ colors.
For this, it is enough to find a proper coloring of the torso of each bag of $\Tcal$ with at most $\max\{4,k+1\}$ colors.
If a bag as at most $k+1$ vertices, then $k+1$ colors suffices, and if a bag has a planar torso, then four colors are enough~\cite{AppelH89,RobertsonSST96effic}, hence the result.
\end{proof}

\subsection{Counting perfect matchings}
\label{coudif_jkiolmd}

Our second example is an extension of  the celebrated FKT algorithm \cite{fisher1961statistical,kasteleyn1961statistics,temperley1961dimer} for counting perfect matchings from planar to $\Hcal$-planar graphs. Recall that, given a graph $G$ with an edge weight function $w: E(G)\to\Nbbb$, the weighted number of perfect matchings is 
$$\pmm(G)=\sum_{M}\prod_{uv\in M} w(uv)$$ where the sum is taken over all perfect matchings $M$. 

\perfect*

To prove this, we use Valiant ``matchgates'' \cite{Valiant08holo}, or more precisely the gadgets from \cite{StraubTW16coun}.
\begin{proposition}[\!\!\cite{StraubTW16coun}]\label{prop_pmm}
Let $G$ be a graph and $(A,B)$ be a separation of $G$ of order two (resp. three) with $A\cap B=\{a,b\}$ (resp. $A\cap B=\{a,b,c\}$).
Denote by $G_B$ the graph obtained from $G[B]$ by removing the edges with both endpoints in $A\cap B$.
Let also $p_\gamma:=\pmm(G_B-\gamma)$ for each $\gamma\subseteq A\cap B$.
Then, $\pmm(G)=\pmm(G')$, where $G'$ is the graph obtained from $G[A]$ by adding the corresponding planar gadget of \autoref{fig_valiant2} (that depends on the size of $A\cap B$ and the parity of $B$).
\end{proposition}

\begin{figure}[h]
\center
\scalebox{.97}{\includegraphics{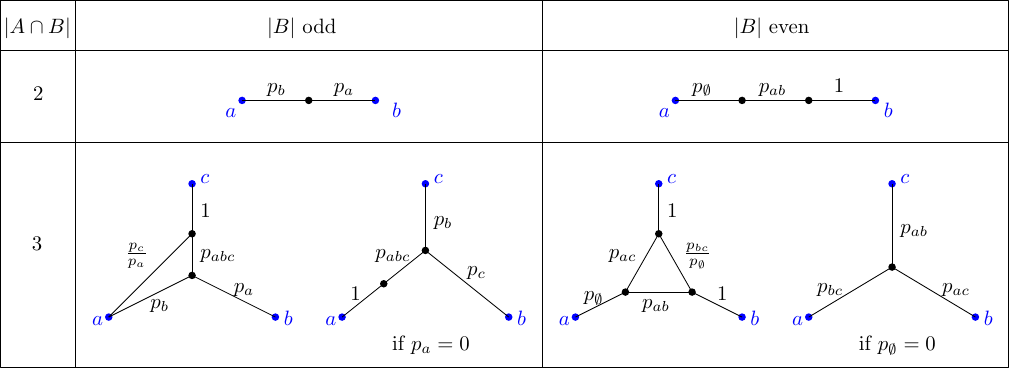}}
\caption{The gadget (from \cite{StraubTW16coun}) we replace $G_B$ with, depending on the size of $A\cap B$ and the parity of $|B|$; to simplify notation, we write the indices of the weight $p_\gamma$ 
as the lists of the elements of $\gamma$ (e.g., we write $p_{abc}$ instead of $p_{\{a,b,c\}}$). 
Note that, if $p_a=0$ or $p_\emptyset=0$, then there may be a variant gadget.}
\label{fig_valiant2}
\end{figure}

We now prove \autoref{thm_perfectmatching}.

\begin{proof}[Proof of \autoref{thm_perfectmatching}]
Let $G$ be an $\Hcal$-planar graph and $w:E(G)\to\Nbbb$ be a weight function (with $w=1$ in the unweighted case).

If $G$ is not connected, then the weighted number of perfect matchings in $G$ is the product of the number of perfect matchings in the connected components of $G$.
Hence, we assume without loss of generality that $G$ is connected.

Obviously, if $G$ has an odd number of vertices, then $\pmm(G)=0$.
So we assume that $|V(G)|$ is even.
Additionally, if $G$ has a cut vertex $v$, let $C_1\in\cc(G-v)$ and $C_2=\bigcup_{C\in\cc(G-v)\setminus\{C_1\}}C$.
Given that $|V(G-v)|$ is odd, exactly one of $C_1$ and $C_2$, say $C_1$, has an odd number of vertices.
Then the (weighted) number of perfect matchings is the product of $\pmm(C_1\cup v)$ and $\pmm(C_2)$.
Therefore, we can assume without loss of generality that $G$ is 2-connected.

By \Cref{cor_modulator}, there is a polynomial-time algorithm computing a planar $\Hcal$-modulator
 $S$ of  an  $\Hcal$-planar graph $G$.
Let $\delta'$ be a sphere embedding of $\torso(G,S)$.
Let $\delta$ be the sphere embedding of $G[S]$ obtained from $\delta'$ by removing the edges that do not belong to $G[S]$.
By \autoref{cl_compat}, for every $D\in\cc(G-S)$, there is a $\delta$-aligned disk $\Delta_D$ such that:
\begin{itemize}[itemsep=-2pt] \setlength\itemsep{0em}
\item the vertices of $N_G(V(D))$ are in the disk $\Delta_D$, i.e. $N_G(V(D))\subseteq \pi_\delta(N(\delta)\cap \Delta_D)$, 
\item with all but at most one (in the case  $|N_G(V(D))|\le 4$) being exactly the vertices of the boundary of $\Delta_D$, i.e. there is a set $X_D\subseteq N_G(V(D))$ of size $\min\{|N_G(V(D))|,3\}$ such that $X_D=\pi_\delta(\bd(\Delta_D)\cap N(\delta))$.
\end{itemize}
Actually, \autoref{cl_compat} is proved for $\delta'$, but it is still true after removing edges to obtain $\delta$.
Note that, given that $G$ is 2-connected, we have $2\le |X_D|\le 3$.

Let $\{\Delta_1,\dots,\Delta_p\}$ be the set of all such disks $\Delta_D$, ordered so that there is no $i<j\in[p]$ with $\Delta_i\supseteq \Delta_j$.
Note that for any $D,D'$ such that $X_D=X_{D'}$, we assume that $\Delta_D=\Delta_D'$, and thus, that there is a unique $i$ such that $\Delta_i=\Delta_D=\Delta_D'$.
Given that there are at most $\binom{n}{3}$ separators in $G$, we have $p\le n^3$.

For $i\in[p]$, let $Z_i:=\{C\in\cc(G-S)\mid N_{G}(C)\subseteq \pi_\delta(N(\delta)\cap \Delta_i)\}$ be the set of connected components whose neighborhood is in $\Delta_i$.
For $i\in[p]$ in an increasing order, we will construct by induction a tuple $(G_i,S_i,\delta_i,w_i)$ such that 
\begin{enumerate}[itemsep=-2pt] \setlength\itemsep{0em}
\item $S_i$ is a planar $\Hcal$-modulator of the graph $G_i$ with $\cc(G_i-S_i)=\cc(G-S)\setminus\bigcup_{j\le i}Z_j$,
\item $\delta_i$ is a sphere embedding of $G_i[S_i]$
with $\delta_i\setminus\bigcup_{j\le i}{\sf int}(\Delta_j)=\delta\setminus\bigcup_{j\le i}{\sf int}(\Delta_j)$, and
\item $w_i:E(G_i)\to\Nbbb$ is a weight function such that $\pmm(G_i)=\pmm(G)$.
\end{enumerate}
Therefore, $\cc(G_p-S_p)=\emptyset$, so $G_p=G[S_p]$ is a planar graph with $\pmm(G_p)=\pmm(G)$.
Then, by \cite{Kasteleyn67grap}, $\pmm(G_p)$, and thus $\pmm(G)$, can be computed in polynomial time.

Obviously, the conditions are originally respected for $(G_0,S_0,\delta_0,w_0)=(G,S,\delta,w)$.
Let $i\in[p]$. 
Suppose that we constructed $(G_{i-1},S_{i-1},\delta_{i-1},w_{i-1})$.
Let $B_i:=V(\inG_\delta(\Delta_i))\cup V(Z_i)$ be the set of vertices that are either in $\Delta_i$ or in a connected component whose neighborhood is in $\Delta_i$.
Let also $X_i:=\pi_\delta(N(\delta_{i-1})\cap \Delta_i)$ be the vertices on the boundary of $\Delta_i$.
Given that there are no $j<i$ such that $\Delta_i\subseteq\Delta_j$, and by the induction hypothesis (2), there is $D\in\cc(G-S)$ such that $X_i=X_D$, so we have $2\le|X_i|\le3$.
If $|X_i|=2$, we label its elements $a$ and $b$, and if $|X_i|=3$, we label its elements $a,b,c$.
Let $H_i$ be the graph induced by $B_i$ but where we remove the edges whose endpoints are both in $X_i$.

For $\gamma\subseteq X_i$, let $$p_\gamma:=\pmm(H_i-\gamma).$$
Then, by \autoref{prop_pmm}, if we manage to compute $p_\gamma$ for each $\gamma\subseteq X_i$, then we can replace $H_i$ with the corresponding planar gadget $F_i$ of \autoref{fig_valiant2} (that depends on the size of $X_i$ and the parity of $|V(H_i)|=|B_i|$) to obtain~$G_i$ and the weight function $w_i$, so that $\pmm(G_i)=\pmm(G_{i-1})=\pmm(G)$.
So Item (3) of the induction holds.
Then $S_i:=S_{i-1}\setminus Z_i\cup V(F_i)$ is a planar $\Hcal$-modulator of $G_i$ with $\cc(G_i-S_i)=\cc(G_{i-1}-S_{i-1})\setminus Z_i=\cc(G-S)\setminus\bigcup_{j\le i}Z_j$.
So Item (1) of the induction holds.
Additionally, $\delta_i$, that is obtained from $\delta_{i-1}$ by adding the gadget to the planar embedding, is such that $\delta_i\setminus{\sf int}(\Delta_i)=\delta_{i-1}\setminus{\sf int}(\Delta_i)$, and thus respects Item (2) by induction hypothesis.
Then, the obtained tuple $(G_i,S_i,\delta_i,w_i)$ would respects the induction hypothesis.

What is left is to explain how to compute $p_\gamma$ for $\gamma\subseteq X_i$.
For each $v\in X_i\setminus\gamma$, we guess with which vertex of $B_i\setminus X_i$ $v$ is matched.
There are $\Ocal(n)$ choices for each such $v$, and therefore $\Ocal(n^3)$ guesses to match all vertices in $X_i\setminus\gamma$.
Let $H_i'$ be the graph obtained from $H_i$ after removing $X_i$ and the vertices matched with vertices of $X_i\setminus\gamma$.
We have $p_\gamma=\pmm(H_i')+|X_i\setminus\gamma|$.
Note that for any $C\in Z_i$, $X_i\subseteq N_{G_{i-1}}(C)=N_G(V(C))$.
Indeed, otherwise, there would be $j>i$ such that $\Delta_C=\Delta_j$, with $\Delta_j\subseteq \Delta_i$, a contradiction.
Therefore, for any $C\in Z_i$, if $|N_G(V(C))|\le3$, then $C$ is a connected component of $H_i'$, and if $|N_G(V(C))|=4$, then $C$ is a block of $H_i'$.
The rest of $H_i'$ is $\inG_{\delta_{i-1}}(\Delta_i)$, that is a planar graph.
Therefore, $H_i'$ is a graph whose blocks are either planar or in $\Hcal$.
As said above, for each cut vertex, we know in which block it should be matched depending on the parity of the blocks.
Hence, we can compute the weighted number of perfect matchings in each block separately.
Note that the weights added by the gadgets of \autoref{fig_valiant2} are only present in planar blocks, where we know how to compute the weighted number of perfect matchings in polynomial time.
In blocks belonging to $\Hcal$, the weight of edges did not change, so in the unweighted (resp. weighted) case, we know how to compute the unweighted  (resp. weighted) number of perfect matchings in polynomial time.
Thus, we can compute $\pmm(H_i')$ and, therefore, $p_\gamma$ in polynomial time.
\end{proof}

Again, by repetitive applications of \autoref{thm_perfectmatching}, we immediately obtain the following for graphs with bounded \hptd{$\Hcal$}.
\begin{corollary}
Let $\Hcal$ be a hereditary, CMSO-definable, and polynomial-time decidable graph class. Moreover, assume that the weighted (resp. unweighted) number of perfect matchings  $\pmm(H)$ is computable in polynomial time for graphs in $\Hcal$. Then, there exists an algorithm that, given a weighted (resp. unweighted) graph $G$ with 
$\Hcal\mbox{-}\ptd(G)\leq k$, 
computes its weighted (resp. unweighted) number of perfect matchings $\pmm(G)$ in time $n^{O(k)}$.
\end{corollary}

If $G$ has $\Hcal$-planar treewidth at most $k$, even given a $\Pcal^{k}\triangleright\Hcal$-modulator of $G$, computing the number of perfect matchings is more difficult.
Still, combining \cref{thm_perfectmatching} with the dynamic programming approach of \cite{ThilikosW22}, it is easy to derive  an  algorithm that, given a weighted graph, computes the weighted  number of its perfect matchings in time $n^{O(k)}$.

\begin{theorem}
Let $\Hcal$ be a hereditary, CMSO-definable, and polynomial-time decidable graph class. Moreover, assume that the weighted (resp. unweighted) number of perfect matchings  $\pmm(H)$ is computable in polynomial time for graphs in $\Hcal$. Then, there exists an algorithm that, given 
$\Hcal\mbox{-}\ptw(G)\leq k$,
computes its weighted (resp. unweighted) number of perfect matchings $\pmm(G)$ in time $n^{O(k)}$.
\end{theorem}

\subsection{EPTAS for Independent Set}
\label{baker_more_r}
Baker's technique \cite{Baker94appr} provides PTAS and EPTAS  for many optimization problems on planar graphs. This method is based on the fact that planar graphs have bounded \emph{local treewidth}. In particular,  for every planar graph $G$ and every vertex $v\in V(G)$, $\tw(N_G^r[v])=\Ocal(r)$, where $N_G^r[v]$ is the \emph{closed $r$-neighborhood} of $v$, that is, the set of vertices at distance at most $r$ from $v$
(see \cite{RobertsonS84III}).

Given that the torso of a planar $\Hcal$-modulator $X$ is planar, we extend Baker's technique to $\Hcal$-planar graphs, by replacing the role of treewidth with \emph{$\Hcal$-treewidth}.

\begin{definition}[$\Hcal$-treewidth.]
Let $\Hcal$ be a graph class and $G$ be a graph.
A \emph{tree $\Hcal$-decomposition} of $G$ is a triple $\Tcal=(T,\beta,L)$ where $(T,\beta)$ is a tree decomposition of $\torso(G,X)$ for some set $X\subseteq V(G)$ such that each connected component in $L:=\cc(G-X)$ belongs to $\Hcal$.
The width of $\Tcal$ is the width of $(T,\beta)$.
The $\Hcal$-treewidth of $G$ is the minimum width of a tree $\Hcal$-decomposition of $G$,
 or,
in other words, the
minimum treewidth of the torso of a set $X\subseteq V(G)$ such that each component of $G-X$ belongs to $\Hcal$. 
\end{definition}

As mentioned in the introduction,  $\Hcal$-treewidth has been defined by Eiben, Ganian,  Hamm,  and  Kwon \cite{EibenGHK21}. Its algorithmic properties have been studied 
in \cite{EibenGHK21,JansenK021verte,AgrawalKLPRSZ22}.

\medskip
Let us explain how the EPTAS works on {\sc Independent Set}.
The problem of {\sc Independent Set} asks for an \emph{independent set} of maximum size, that is a set $S$ such that no edge of $G$ has both endpoints in $S$.
The following is proved in \cite[Theorem 3.16]{JansenK021verte} (it is proven for the dual problem of {\sc Vertex Cover}).

\begin{proposition}[\!\!  \cite{JansenK021verte}]\label{prop_IS}
Let $\Hcal$ be a hereditary graph class on which {\sc Independent Set} is polynomial-time solvable.
Then {\sc Independent Set} can be solved in time $2^k\cdot n^{\Ocal(1)}$ when given a tree $\Hcal$-decomposition of width at most $k-1$ consisting of $n^{\Ocal(1)}$ nodes.
\end{proposition}

With  \Cref{cor_modulator} on hands, the proof of the following theorem is almost identical to Baker's style algorithms for planar graphs. The main observation is that the supergraph of an input $\Hcal$-planar graph $G$ obtained by adding to $G$ all edges of the torso of its planar modulator, is of bounded local  $\Hcal$-treewidth.

\independent*

\begin{proof}
Let $G$ be an $\Hcal$-planar graph. We assume without loss of generality that $G$ is a connected graph and $\varepsilon<1$.
By \Cref{cor_modulator}, there is a polynomial-time algorithm computing a planar $\Hcal$-modulator
 $X$ of $G$. 
 Let $v$ be an arbitrary vertex in $X$.
For $i\in\Nbbb$, let $L_i$ denote the set of vertices of $X$ at distance $i$ of $v$ in $\torso(G,X)$.
Observe that the endpoints of an edge in $\torso(G,X)$ belong to at most two (consecutive) $L_i$s.
Given that, for $C\in\cc(G-X)$, $N_G(V(C))$ induces a clique in $\torso(G,X)$, it therefore implies that the vertices of $N_G(V(C))$ belong to at most two (consecutive) $L_i$s.

Let $k$ be the smallest integer such that $2/k\le\varepsilon$.
For $i\in[0,k-1]$, let $S_i$ be the union of the $L_j$ for which $j$ is equal to $i$ (mod $k$).
Let $X_i:=X\setminus S_i$ and let $\Ccal_i$ be the set of components $C\in\cc(G-X)$ such that $N_G(V(C))\subseteq X_i$.
Let $G_i$ be the induced subgraph of  $G$ induced by $X_i$ and $\Ccal_i$. 
Given that the connected components of $G_i-X_i$ are connected components of $G-X$, it implies that $\torso(G_i,X_i)$ is a subgraph of $\torso(G,X)-S_i$, and is thus planar.
Therefore, as proven in \cite[Corollary 7.34]{cygan2015parameterized} (see also \cite[(2.5)]{RobertsonS84III}), 
$\torso(G_i,X_i)$ has treewidth $\Ocal(k)$ and there is a polynomial-time algorithm constructing a tree decomposition  $(T_i,\beta_i)$ of $\torso(G_i,X_i)$ of width $\Ocal(k)$.  
Therefore, $(T_i,\beta_i,\Ccal_i)$ is a tree $\Hcal$-decomposition of $\torso(G_i,X_i) \cup G[\Ccal_i]$, and hence of  $G_i$,  of width   $\Ocal(k)$.
We apply \autoref{prop_IS} to find in time $2^{\Ocal(k)}\cdot n^{\Ocal(1)}=2^{\Ocal(1/\varepsilon)}\cdot n^{\Ocal(1)}$ a maximum independent set $I_i$ of $G_i$.
Note that $I_i$ is also an independent set of $G$.
We return as a solution the maximum size set $I_i$  for $i\in[0, k-1]$.

Let $I^*$ be a maximum independent set of $G$, of size $|I^*|=\alpha(G)={\rm OPT}$.
For $i\in[0,k-1]$, let $\Dcal_i:=\cc(G-X)\setminus\Ccal_i$ and let $S_i':=S_i\cup V(\Dcal_i)$.
Let $C\in\cc(G-X)$.
Given that the vertices of $N_G(V(C))$ belong to at most two (consecutive) $L_i$s, it implies that $C$ belongs to at most two (consecutive) $\Dcal_i$s.
Additionally, the $S_i$s partition $X$.
Therefore, 
$$\sum_{i=0}^{k-1}|S_i'|=|X|+\sum_{i=0}^{k-1}|V(\Dcal_i)|\le|X|+2|V(G)\setminus X|=2|V(G)|-|X|\le2|V(G)|.$$
Hence, there is $j\in[0,k-1]$ such that $|I^*\cap S_j'|\le2|I^*|/k$.
Since $I^*\setminus S_j'$ is an independent set in $G_j$, we have the following:
$$|I_j|\ge|I^*\setminus S_j'|=|I^*|-|I^*\cap S_j'|\ge(1-2/k)|I^*|\ge(1-\varepsilon){\rm OPT}.$$
Thus, the algorithm returns a solution of size at least $(1-\varepsilon){\rm OPT}$ in time $2^{\Ocal(1/\varepsilon)}\cdot n^{\Ocal(1)}$.
\end{proof}

As with the previous applications, it is possible to generalize the technique for the graphs of bounded $\Hcal$-planar treewidth or treedepth. We sketch our algorithms in the following two corollaries.

\begin{corollary}\label{cor:indtw}
Let $\Hcal$ be a hereditary, CMSO-definable, and polynomial-time decidable graph class.
Moreover, assume that the maximum independent set  is computable in polynomial time for graphs in $\Hcal$. 
Then, there exists a polynomial-time algorithm that, given $\varepsilon>0$
and a graph $G$ with $\Hcal\mbox{-}\ptw(G)\leq k$, 
computes in time $2^{\Ocal(\max\{k,1/\varepsilon\})}|V(G)|^{\Ocal(1)}$ an independent set of $G$ of size at least $(1-\varepsilon)\cdot\alpha(G)$.
\end{corollary}

\begin{proof}[Sketch of the proof]
We assume without loss of generality that $G$ is a connected graph and $\varepsilon<1$.
We use \cref{cor:twmodulator} to find a $\Pcal\Tcal_k\triangleright\Hcal$-modulator $S$ of $G$ with a tree decomposition $\Tcal$ of $\torso_G(S)$ of planar width at most $k$.
Let $X_1,\ldots,X_\ell$ be the bags of the decomposition such that $G_i=\torso_G(X_i)$ is a planar graph for $i\in[\ell]$. Because $G$ is connected, each $G_i$ is a connected graph. 
We select arbitrarily $\ell$ vertices $v_i\in X_i$ for $i\in[\ell]$. For each $i\in[\ell]$ and every integer $j\geq 0$, we denote by  $L_j^i$ the set of vertices of $X_i$ at distance $j$ from $v_i$ in $G_i$.  Let $h\geq 2$ be the smallest integer such that $2/h\leq \varepsilon$.   
For $j\in[0,h-1]$, let $S_j$ be $\bigcup_p\bigcup_{i=1}^\ell L_p^i$ where the first union is taken over all $p\geq 0$ for which  $p$ is equal to $i$ (mod $h$). 
For each $j\in[0,h-1]$, let $Y_j=X\setminus S_j$. The crucial observation is that, because each $G_i$ is planar, the adhesion set of $X_i$ with other bags in  $\Tcal$ is in at most two consecutive sets $L_q^i$ and $L_{q+1}^i$ for some $j\geq 0$. This implies that $\tw(G[Y_i])=\Ocal(\max\{h,k\})$. This implies that we can apply the same Baker's style arguments as in the proof of \cref{lem_IS}. This concludes the sketch of the proof.
\end{proof}

\begin{corollary}\label{cor:indtd}
Let $\Hcal$ be a hereditary, CMSO-definable, and polynomial-time decidable graph class. Moreover, assume that the maximum independent set  is computable in polynomial time for graphs in $\Hcal$. Then, there exists a polynomial-time algorithm that, given $\varepsilon>0$, a graph $G$ with $\Hcal\mbox{-}\ptd(G)\leq k$,
computes in time $2^{\Ocal(k+1/\varepsilon)}|V(G)|^{\Ocal(1)}$ an independent set of $G$ of size at least $(1-\varepsilon)\cdot\alpha(G)$.
\end{corollary}

\begin{proof}[Sketch of the proof]
The main idea is the same as in \cref{cor:indtw}. 
We assume without loss of generality that $G$ is a connected graph and $\varepsilon<1$. 
We use~\cref{cor:tdseq} to find a certifying elimination sequence $X_1,\dots,X_k$.
For $i\in[k]$ and $j\in[\ell_i]$, denote by $Y_{ij}\subseteq X_i$  the inclusion maximal subsets of $X_i$ such that  $G_{ij}=\torso_{G-\cup_{h=1}^{i-1}X_h}(Y_{ij})$ is connected. We arbitrarily select vertices $v_{ij}\in Y_{ij}$ for $i\in[k]$ and $j\in[\ell_i]$.  For $i\in[k]$, $j\in[\ell_i]$, and each integer $h\geq 0$, denote by $L_h^{ij}$ the set of vertices of $Y_{ij}$ at distance $h$ from $v_{ij}$ in $G_{ij}$. 
Let $s\geq 2$ be the smallest integer such that $2/s\leq \varepsilon$.  For $t\in[0,s-1]$, let $S_t$ be $\bigcup_p\bigcup_{i=1}^k\bigcup_{j=1}^{\ell_i} L_p^{ij}$
where the first union is taken over all $p\geq 0$ for which  $p$ is equal to $i$ (mod $s$). 
For each $t\in[0,s-1]$, let $Z_t=X\setminus S_t$. Now we have that $\tw(G[Z_t])=\Ocal(s+k)$. Using the same arguments as in \cref{lem_IS} and \cref{cor:indtw}, we show the desired approximation.  This concludes the sketch of the proof.
\end{proof}

Baker's style arguments and other approaches apply to many optimization problems on planar graphs. 
For instance, dynamic programming algorithms on graphs of bounded $\Hcal$-treewidth,  as  the one of \Cref{prop_IS}, exist for other problems (see \cite{JansenK021verte,EibenGHK21}) that are amenable to Baker's technique.
Investigating which of these problems admit good approximations on
 $\Hcal$-planar graphs and whether there are meta-algorithmic theorems in the style of \cite{DemaineFHT04bidimen, FominLS18exclude,Grohe03loc}, is an interesting research direction that goes beyond the scope of this paper.

\section{Necessity of conditions}\label{sec_lower}
In our algorithmic results, we require that the considered graph classes $\Hcal$ should satisfy certain properties. In this section, we discuss tightness of these properties.  
First, we show that  
the heredity condition for $\Hcal$ is crucial for the existence of a polynomial-time algorithm for  \Hpl. For this, we prove that \Hpl is \NP-hard  even if $\Hcal$ consists of a single  graph.  

\begin{theorem}\label{thm_lower}
Let $\Hcal$ be the class consisting of the complete graph on four vertices. Then \Hpl{} is \NP-complete.
\end{theorem}
 
 \begin{proof}
We show \NP-hardness by reducing from a variant of the \textsc{Planar SAT} problem. Consider a Boolean formula $\varphi$ in the conjunctive normal form on $n$ variables $x_1,\ldots,x_n$ with clauses $C_1,\ldots,C_m$. We define the graph $G(\varphi)$ with $2n+m$ vertices constructed as follows:
\begin{itemize}
\item for each $i\in[n]$, construct  two vertices $x_i$ and $\overline{x}_i$ and make them adjacent,
\item for each $j\in[m]$, construct a vertex $C_j$,
\item for each $i\in[n]$ and $j\in [m]$, make $x_i$ and $C_j$ adjacent if the clause $C_j$ contains the literal $x_i$, and make $\overline{x}_i$ and $C_j$ adjacent if $C_j$ contains the negation of $x_i$.
\end{itemize}
It is known that the \textsc{SAT} problem is \NP-complete when restricted to instances where (i)~$G(\varphi)$ is planar, (ii)~each variable occurs at most 4 times in the clauses---at most two times in positive and at most two times with negations, and (iii)~each clause contains either two or three literals~\cite{Kratochvil94}.

 \begin{figure}[ht]
\centering
\includegraphics{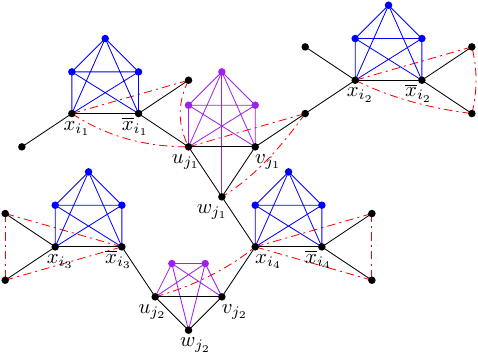}
\caption{Construction of  $G'$ and $G''$. The vertices and edges of $G'$ are shown in black and the additional vertices and edges of $G''$ are blue (for the variable) and purple (for the clauses). The construction is shown for four variables $x_{i_1},x_{i_2},x_{i_3},x_{i_4}$ and two clauses $\overline{x}_1\vee x_{i_2}\vee x_{i_4}$ and $\overline{x}_{i_3}\vee x_{i_4}$. Here, $x_{i_1},x_{i_3}$ occur in three clauses and $x_{i_2},x_{i_4}$ occur in four clauses. We also show in dashed red the edges of the torso of $X$ constructed for the assignment $x_{i_1}=x_{i_2}=x_{i_4}=true$ and $x_{i_3}=false$. 
}
\label{fig_lb}
\end{figure} 

Consider such an instance $\varphi$ of   \textsc{Planar SAT} with the variables $x_1,\ldots,x_n$ and the clauses $C_1,\ldots,C_m$.  We construct the following graph $G'$:
\begin{itemize}
\item for each $i\in[n]$, construct  two \emph{variable} vertices $x_i$ and $\overline{x}_i$ and make them adjacent,
\item for each $j\in[m]$, construct three \emph{clause} vertices  $u_j,v_j,w_j$ and make them pairwise adjacent, 
\item for each $i\in[n]$ and $j\in[m]$, make $x_i$  adjacent to one of the vertices $u_j,v_j,w_j$ if the clause $C_j$ contains the literal $x_i$, and make $\overline{x}_i$ adjacent to one of the vertices $u_j,v_j,w_j$ if $C_j$ contains the negation of $x_i$; we select the neighbor of $x_i$ or $\overline{x}_i$ among $u_j,v_j,w_j$ in such a way that no clause vertex is adjacent to two variable vertices.    
\end{itemize}
Notice that because $G(\varphi)$ is planar, $G'$ is a planar graph. In the next stage of our reduction, we construct $G''$ from $G'$:
\begin{itemize}
\item for each $i\in[n]$, construct three vertices $y_i^1,y_i^2,y_i^3$ and make them pairwise adjacent and adjacent to both $x_i$ and $\overline{x}_i$,
\item for each $j\in[m]$ such that $|C_j|=3$, construct a set of three vertices $Z_j=\{z_j^1,z_j^2,z_j^3\}$ and make the vertices of $Z_j$ pairwise adjacent and adjacent to  $u_j,v_j,w_j$,
\item for each $j\in[m]$ such that $|C_j|=2$, construct a set of two vertices $Z_j=\{z_j^1,z_j^2\}$ and make the vertices of $Z_j$ adjacent and adjacent to  $u_j,v_j,w_j$.
\end{itemize}
The construction of $G'$ and $G''$ is shown in \autoref{fig_lb}.
In our construction, we assume that each clause vertex $w_j$ constructed for $C_j$ of size two is not adjacent to any variable vertex.

We claim that $\varphi$ has a satisfying assignment if and only if $G''$ has a planar $\Hcal$-modulator.

Suppose that $\varphi$ has a satisfying assignment and the variables $x_1,\ldots,x_n$ have values satisfying $\varphi$. 
We define the induced subgraphs $H_1,\ldots,H_n$ and $H_1',\ldots,H_m'$ of $G''$ by setting 
\begin{itemize}
\item for every $i\in[n]$, $V(H_i)=\{x_i,y_i^1,y_i^2,y_i^3\}$ if $x_i=false$ and $V(H_i)=\{\overline{x}_i,y_i^1,y_i^2,y_i^3\}$, otherwise,
\item for every $j\in[m]$, consider a literal $\ell=true$ in $C_j$ and let $r\in\{u_j,v_j,w_j\}$ be the clause vertex adjacent to the variable vertex corresponding to $\ell$, and then set $V(H_j')=(\{r\}\cup Z_j)$ if $|C_j|=3$ and $V(H_j')=(\{r,w_j\}\cup Z_j)$ if $|C_j|=2$.
\end{itemize}
Notice that $H_1,\ldots,H_n$ and $H_1',\ldots,H_m'$ are copies of $K_4$, that is, they are in $\Hcal$.  Furthermore, all these subgraphs of $G''$ are disjoint, and distinct subgraphs have no adjacent edges.
We set \[X=V(G'')\setminus\Big(\big(\bigcup_{i=1}^n V(H_i)\big)\cup\big(\bigcup_{j=1}^m V(H_j')\big)\Big).\]
Then $G''-X$ is the disjoint union of $n+m$ copies of $K_4$ and we obtain that each connected component of $G''-X$ is in $\Hcal$. 
Because $G'$ is planar, we have that the torso of $X$ is planar as well. To see this, notice that the torso of $X$ can be obtained from $G'$ by edge contractions and replacements of $Y$-subgraphs, that is, subgraph isomorphic to the star $K_{1,3}$ whose central vertex has degree three in $G'$, by triangles formed by the leaves of the star (see~\autoref{fig_lb}).   Since these operations preserve planarity, the torso of $X$ is planar. 
Thus,  $G''$ has a planar $\Hcal$-modulator.

For the opposite direction, assume that $G''$ has a planar $\Hcal$-modulator. Let $X\subseteq V(G'')$ be such that the torso of $X$ is planar and $G''-X$ is the disjoint union of copies of $K_4$. Consider $i\in[n]$. The clique $R_i=\{x_i,\overline{x}_i,y_i^1,y_i^2,y_i^3\}$ of size 5, so 
$X\cap R_i\neq\emptyset$. Observe that, by the construction of $G''$, any clique of size four containing a vertex of $H_i$ is contained in $R_i$. Thus, $G''-X$ has a connected component $H_i$ that is copy of $K_4$ with $V(H_i)\subseteq R_i$. 
Therefore, at least one of the vertices $x_i$ and $\overline{x}_i$ is in $H_i$. We set the value of the variable $x_i=true$ if $x_i\notin V(H_i)$ and $x_i=false$, otherwise. 
We perform this assignment of the values to all variables $x_1,\ldots,x_n$. We claim that this is a satisfying assignment for $\varphi$. 

Consider $j\in[m]$. Suppose first that $|C_j|=3$. Let $R_j'=\{u_j,v_j,w_j,z_j^1,z_j^2,z_j^3\}$. Since $R_j'$ is a clique of size 6, $X\cap R_j'\neq \emptyset$. The construction of $G''$ implies that each clique of 
size four containing a vertex of $R_j'$ is contained in $R_j'$. This means that $G''-X$ has a connected component $H_j'$ isomorphic to $K_4$ such that $V(H_j')\subseteq R_j'$. Then at least one of the vertices $u_j,v_j,w_j$ is in $V(H_j')$. By symmetry, we can assume without loss of generality that $u_j\in V(H_j')$. Also by symmetry, we can assume that $u_j$ is adjacent to $x_i$ for some $i\in[n]$. Because $H_i$ and $H_j'$ are disjoint and have no adjacent vertices, $x_i\notin V(H_i)$. This means that $x_i=true$ and the clause $C_j$ is satisfied. The case $|C_j|=2$ is analyzed in a similar way. Now we let $R_j'=\{u_j,v_j,w_j,z_j^1,z_j^2\}$ and have that $R_j'$ is a clique of size 5. Then $X\cap R_j'\neq \emptyset$. We again have that  each clique of 
size four containing a vertex of $R_j'$ is contained in $R_j'$. Therefore, $G''-X$ has a connected component $H_j'$ isomorphic to $K_4$ such that $V(H_j')\subseteq R_j'$. Recall that we assume that $w_j$ is not adjacent to any variable vertex. Then either $u_i$ or $v_i$ is in $V(H_j')$. We assume without loss of generality that $u_j\in V(H_j)$ and  $u_i$ are adjacent to $x_i$ for some $i\in[n]$. Because $H_i$ and $H_j'$ are disjoint and have no adjacent vertices, $x_i\notin V(H_i)$. Thus, $x_i=true$ and the clause $C_j$ is satisfied. This proves that each clause is satisfied by our truth assignment for the variables. We conclude that $\varphi$ has a satisfying assignment.

It is straightforward that $G''$ can be constructed in polynomial time. Therefore, \Hpl\ is \NP-hard for $\Hcal=\{K_4\}$.  As the inclusion of  \Hpl\ in \NP{} is trivial, this concludes the proof.
\end{proof}

In \Cref{th_param}, the existence of the respective algorithms for $\p\in\{\ptd,\ptw\}$ is shown under condition that  the {\sc $\Hcal$-Deletion} problem parameterized by the solution size is \FPT. 
While we leave open the question whether the existence of an \FPT algorithm for checking that the \hptd{$\Hcal$} is at most $k$ or 
 the  $\Hcal$-planar treewidth at most $k$ would imply that there is an \FPT algorithm solving {\sc $\Hcal$-Deletion}, we note that for natural hereditary graph classes $\Hcal$ for which {\sc $\Hcal$-Deletion} is known to be  {\sf W[1]}-hard or {\sf W[2]}-hard, deciding whether the elimination distance to $\Hcal$ and whether the   $\Hcal$-planar treewidth at most $k$ is also can be shown to be hard. This follows from the observation that for graphs of high vertex connectivity, the three problems are essentially equivalent. In particular, the following can be noticed.
 
 \begin{observation}\label{obs:equivalent-high-conn}
 Let $G$ be a $4k+2$-connected graph for an integer $k\geq 0$. Then, for any class $\Hcal$, a graph $G^*\in\Hcal$ can be obtained by at most $4k$ vertex deletions from $G$ if and only if the \hptd{$\Hcal$} of $G$ is at most $k$.
 \end{observation}
 
\begin{proof}
Suppose that there is a set $X\subseteq V(G)$ of size at most $4k$ such that $G-X\in \Hcal$. Because $|X|\leq 4k$, $G-X$ is connected. 
Also, since $|X|\leq 4k$, there is a partition $\{X_1,\ldots,X_s\}$ of $X$ such that  $s\leq k$ and $|X_i|\leq 4$ for each $i\in[s]$.  Notice that for each $i\in[s]$, $\torso(G_i,X_i)$, where $G_i=G-\bigcup_{j=1}^{i-1}X_j$ is planar. Then we decompose $G$ by consecutively selecting $X_1,\ldots,X_s$. 

For the opposite direction, we use induction on $k$. We prove that if the \hptd{$\Hcal$} of a graph $G$ is at most $k$ then 
a graph $G'\in\Hcal$ can be obtained by deleting at most $4k$ vertices. The claim trivially holds if $G\in\Hcal$ as  the deletion distance to $\Hcal$ is $0\leq k$. In particular, this proves the claim for $k=0$. Assume that $k\geq 1$ and $G\notin\Hcal$. Because $\ptd_{\Hcal}(G)\leq k$, there is a nonempty set $X\subseteq V(G)$ such that 
(i)~for each connected component $C$ of $G'=G-X$,  $\ptd_{G'}(C)\leq k-1$ and (ii)~$\torso(G,X_1)$ is planar. Because $G$ is $6$-connected, $G$ is not planar and, therefore, $X\neq V(G)$. Consider an arbitrary connected component $C$ of $G-S$. Because   $\torso(G,X)$ is planar, $N_G(V(C))\leq 4$. Since $N_G(V(C))$ cannot be a separator in a 6-connected graph, we obtain that $X=N_G(V(C))$ is of size at most four and $C=G'$ is a unique connected component of $G-X$. Because $|X|\leq 4$, $G'$ is $4(k-1)+2$-connected and we can apply induction. Then there is a set $Y\subseteq V(G')$ such that $|Y|\leq 4(k-1)$ and $G^*=G'-Y\in \Hcal$. Consider $Z=X\cup Y$. Then $G^*=G'-Y=G-Z$ and $|Z|\leq 4k$. This proves that the deletion distance of $G$ to $\Hcal$ is at most $4k$.  
\end{proof} 

Given a computational lower bound for {\sc $\Hcal$-Deletion}, typically, it is easy to show that the hardness holds for highly connected graphs and for $k$ divisible by four. Here, we give just one example. If was proved by Heggernes et al.~\cite{HeggernesHJKV13} that \textsc{Perfect Deletion}, that is, {\sc $\Hcal$-Deletion} when $\Hcal$ is the class of perfect graphs, is {\sf W[2]}-hard when parameterized by the solution size. Then we can show the following theorem using \Cref{obs:equivalent-high-conn}.

\begin{theorem}\label{thm_prefect}
Let $\Hcal$ be the class of perfect graphs.
The problem of deciding whether the \hptd{$\Hcal$} is at most $k$ is {\sf W[2]}-hard when parameterized by $k$.  
\end{theorem}   
 
 \begin{proof}
 We reduce from \textsc{Perfect Deletion}. Let $(G,k)$ be an instance of the problem. First, we can assume without loss of generality that $k=4k'$ for an integer $k'\geq 0$. Otherwise, we add $4-(k\mod 4)$ disjoint copies of the cycle on five vertices to $G$ using the fact that at least one vertex should be deleted form each odd hole to obtain a perfect graph. Then we observe that the class of perfect graphs is closed under adding universal vertices. This can be seen, for example, from  the strong perfect graph theorem~\cite{ChudnovskyRST2006}. Let $G'$ be the graph obtained from $G$ by adding $6k'$ vertices, making them adjacent to each other, and adjacent to each every vertex of $G$. Then the deletion distance of $G$ to perfect graphs is the same as the       deletion distance of $G'$. Thus, the instance $(G,k)$ of \textsc{Perfect Deletion} is equivalent to the instance $(G',k)$. By \Cref{obs:equivalent-high-conn}, we obtain that  $(G',4k)$ is a yes-instance of  \textsc{Perfect Deletion}  if and only if  the \hptd{$\Hcal$} of $G'$ is at most $k'$. This completes the proof.  
 \end{proof} 
  
The same arguments also could be used for the case when $\Hcal$ is the class of weakly chordal graph---it was proved by  Heggernes et al.~\cite{HeggernesHJKV13} that \textsc{Weakly Chordal  Deletion}  is {\sf W[2]}-hard and the class is also closed under adding universal vertices. Also, it is known that the \textsc{Wheel-Free Deletion}, that is, the problem asking whether $k$ vertices may be deleted to obtain a wheel-free graph (a graph is wheel-free if it does not contain a \emph{wheel}, i.e., a graph obtained from a cycle by adding a universal vertex, as an induced subgraph) is {\sf W[2]}-hard when parameterized by $k$ by the result of Lokshtanov~\cite{Lokshtanov08}. We remark that it is possible to show that the lower bound holds for highly connected graphs and obtain the {\sf W[2]}-hardness for the \hptd{$\Hcal$} when $\Hcal$ is the class of wheel-free graphs.  
  
Finally in this section, we note that the variant of \Cref{obs:equivalent-high-conn} holds the planar  $\Hcal$-planar treewidth.

 \begin{observation}\label{obs:equivalent-high-conn-tw}
 Let $G$ be a $(\max\{4,k\}+2)$-connected graph for an integer $k\geq 0$. Then, for any class $\Hcal$, a graph $G^*\in\Hcal$ can be obtained by at most $k$ vertex deletions from $G$ if and only if the $\Hcal$-planar treewidth of $G$ is at most $k$.
 \end{observation}
 
Then the lower bound for \textsc{Deletion to $\Hcal$} from~\cite{HeggernesHJKV13,Lokshtanov08}  can be used to show the {\sf W[2]}-hardness for deciding whether the   $\Hcal$-planar treewidth is at most $k$.

\section{Conclusion}\label{sec_concl}
All our  techniques heavily rely on the planarity of the torso of $X$.
In particular, this concerns our method for finding 
an irrelevant vertex inside a flat wall no matter how the graph outside the modulator interacts with this flat wall. For this, we present our techniques for the \Hpl problem and 
later, in \cref{cosi46y} and \cref{sec_tw} we explain how these techniques can be extended 
on graphs where $\Hcal$-\ptd\ or $\Hcal$-\ptw\ is bounded.
However, \ptd\ and   \ptw\ are not the only minor-monotone parameters to plug in \cref{constrxaint}, further than 
$\td$ and $\tw$. The most general question one may ask in this direction is the following.

\begin{quote}
Given some target class $\Hcal$ and some  minor-monotone parameter $\p$, which conditions on $\Hcal$ may imply an {\sf FPT}-algorithm deciding 
whether $\Hcal$-$\p(G)≤k$?
\end{quote}

We conjecture that the conditions of \cref{th_param} can be copied for \textsl{every} minor-monotone parameter $\p$, in other words we conjecture the following: \textsl{If  $\Hcal$ is hereditary, CMSO-definable, and union-closed then  an {\sf FPT}-algorithm for checking $\Hcal$-$\size(G)≤k$ 
implies an {\sf FPT}-algorithm for checking $\Hcal$-$\p(G)≤k$, no matter the choice of $\p$}.

The main obstacle for applying our technique to this general problem is that now the interaction of the modulator with the rest of the graph is not bounded by a fixed constant (that in our case is $≤4$ for {\sc $\Hcal$-Planarity}) but is bounded by some function of the parameter $k$. 
We believe that our techniques are a promising departure point for resolving this problem. However, this appears to be a quite challenging task.

Another open problem stems from the condition in \cref{th_th} that $\Hcal$ is CMSO-definable. Clearly, we cannot drop the polynomial-time decidability condition for $\Hcal$ if we aim to obtain a polynomial-time algorithm for \Hpl. Also we proved in \cref{thm_lower} that the heredity condition is also crucial. However, the CMSO-definability condition is needed because we are using the meta-theorem of  \cite{LokshtanovR0Z18redu}. This meta-theorem is a powerful tool allowing us to reduce solving \Hpl{} to solving \Hkpl{} for bounded $k$, making it possible to apply our version of the irrelevant vertex technique.  However, we do not exclude that other tools may be used in order to avoid the CMSO-definability requirement. In particular, it can be noted that we aim to decompose in a special way the input graph via separators of size at most three, and graph decompositions of this type were introduced by Grohe~\cite{Grohe16}.  It is very interesting whether avoiding using the meta-theorem and dropping the CMSO condition is possible. This also could allow us to get rid of the non-uniformity of  \cref{th_th}.

\newpage

\bibliographystyle{plainurl}

\end{document}

%% file: planarext_macros.tex
\usepackage[
plainpages = true,
pdfpagelabels,
hyperfootnotes=true,
pdfstartview =,
bookmarks=true,
bookmarksopen = true,
bookmarksnumbered = true,
breaklinks = true,
hyperfigures,
pagebackref,
urlcolor = brown,
anchorcolor = green,
hyperindex = true,
colorlinks = true,
linkcolor = blue,
citecolor = red
]{hyperref}

\usepackage{latexsym,
enumitem,
todonotes,
amsthm,
color,
amssymb,
url,
bm,
cite,
amssymb,
microtype,
amsmath,
caption,
aliascnt,
colordvi,
amsfonts,
mathtools,
amsthm,
nicefrac,
dsfont,
xspace}

\newcommand{\pname}{\textsc}
\newcommand{\ProblemFormat}[1]{\pname{#1}}
\newcommand{\ProblemIndex}[1]{\index{problem!\ProblemFormat{#1}}}
\newcommand{\ProblemName}[1]{\ProblemFormat{#1}\ProblemIndex{#1}{}\xspace}

\newcommand{\probHplan}{\ProblemName{$\Hcal$-Planarity}}

\newcommand{\cqed}{\ensuremath{\lhd}}
\makeatletter
\newenvironment{cproof}{\par
  \pushQED{\cqed}%
  \normalfont \topsep6\p@\@plus6\p@\relax
  \trivlist
  \item\relax
  {\itshape
    Proof of the claim\@addpunct{.}}\hspace\labelsep\ignorespaces
}{%
  \hfill\popQED\endtrivlist\@endpefalse
}
\makeatother

\usepackage{thm-restate}
\usepackage[nameinlink]{cleveref}
\crefname{subsection}{Subsection}{Subsections}
\crefname{section}{Section}{Sections}
\crefname{proposition}{Proposition}{Propositions}
\crefname{theorem}{Theorem}{Theorems}
\crefname{lemma}{Lemma}{Lemmas}
\crefname{corollary}{Corollary}{Corollaries}
\crefname{observation}{Observation}{Observations}
\crefname{figure}{Fig.}{Figures}

\newtheorem{observation}{Observation}
\newtheorem{proposition}{Proposition}
\newtheorem{lemma}{Lemma}
\newtheorem{corollary}{Corollary}
\newtheorem{claim}{Claim}
\newtheorem{theorem}{Theorem}
\newtheorem{definition}{Definition}

\usepackage{nicefrac}
\usepackage[mathscr]{euscript}

\usepackage{alphabeta}

\makeatletter
\input{colordvi}

\newcommand{\remove}[1]{}

\RequirePackage[T1]{fontenc}
\RequirePackage[utf8]{inputenc}
\usepackage[greek,russian,english]{babel}
\usepackage{alphabeta}

\newcommand{\Ccal}{\mathcal{C}}
\newcommand{\Dcal}{\mathcal{D}}

\newcommand{\Fcal}{\mathcal{F}}
\newcommand{\Gcal}{\mathcal{G}}
\newcommand{\Hcal}{\mathcal{H}}

\newcommand{\Ocal}{\mathcal{O}}
\newcommand{\Pcal}{\mathcal{P}}
\newcommand{\Qcal}{\mathcal{Q}}
\newcommand\Rcal{\mathcal{R}}
\newcommand{\Scal}{\mathcal{S}}
\newcommand{\Tcal}{\mathcal{T}}

\newcommand{\Wcal}{\mathcal{W}}

\newcommand{\Nbbb}{\mathbb{N}}

\newcommand{\Rbbb}{\mathbb{R}}

\RequirePackage{stmaryrd}
\usepackage{textcomp}
\DeclareUnicodeCharacter{2286}{\subseteq}
\DeclareUnicodeCharacter{2192}{\ifmmode\to\else\textrightarrow\fi}
\DeclareUnicodeCharacter{2203}{\ensuremath\exists}
\DeclareUnicodeCharacter{183}{\cdot}
\DeclareUnicodeCharacter{2200}{\forall}
\DeclareUnicodeCharacter{2264}{\leq}
\DeclareUnicodeCharacter{2265}{\geq}
\DeclareUnicodeCharacter{8614}{\mathbin{\mapsto}}
\DeclareUnicodeCharacter{8656}{\Leftarrow}
\DeclareUnicodeCharacter{8657}{\Uparrow}
\DeclareUnicodeCharacter{8658}{\Rightarrow}
\DeclareUnicodeCharacter{8659}{\Downarrow}
\DeclareUnicodeCharacter{8669}{\rightsquigarrow}
\newcommand{\eqdef}{\stackrel{{\scriptsize\rm def}}{=}}
\DeclareUnicodeCharacter{8797}{\eqdef}
\DeclareUnicodeCharacter{8870}{\vdash}
\DeclareUnicodeCharacter{8873}{\Vdash}
\DeclareUnicodeCharacter{22A7}{\models}
\DeclareUnicodeCharacter{9121}{\lceil}
\DeclareUnicodeCharacter{9123}{\lfloor}
\DeclareUnicodeCharacter{9124}{\rceil}
\DeclareUnicodeCharacter{2208}{\in}
\DeclareUnicodeCharacter{9126}{\rfloor}
\DeclareUnicodeCharacter{9655}{\triangleright}
\DeclareUnicodeCharacter{9665}{\triangleleft}
\DeclareUnicodeCharacter{9671}{\diamond}
\DeclareUnicodeCharacter{9675}{\circ}
\DeclareUnicodeCharacter{10178}{\bot}
\DeclareUnicodeCharacter{10214}{} 
\DeclareUnicodeCharacter{10215}{} 
\DeclareUnicodeCharacter{10229}{\longleftarrow}
\DeclareUnicodeCharacter{10230}{\longrightarrow}
\DeclareUnicodeCharacter{10231}{\longleftrightarrow}
\DeclareUnicodeCharacter{10232}{\Longleftarrow}
\DeclareUnicodeCharacter{10233}{\Longrightarrow}
\DeclareUnicodeCharacter{10234}{\Longleftrightarrow}
\DeclareUnicodeCharacter{10236}{\longmapsto}
\DeclareUnicodeCharacter{10238}{\Longmapsto} 
\DeclareUnicodeCharacter{10503}{\Mapsto}    
\DeclareUnicodeCharacter{10971}{\mathrel{\not\hspace{-0.2em}\cap}}
\DeclareUnicodeCharacter{65294}{\ldotp}
\DeclareUnicodeCharacter{65372}{\mid}

\definecolor{MidnightBlack}{rgb}{0.1,0.1,.34}
\definecolor{MidnightBlue}{rgb}{0.1,0.1,0.43}
\definecolor{Black}{rgb}{0,0, 0}
\definecolor{Blue}{rgb}{0, 0 ,1}
\definecolor{Red}{rgb}{1, 0 ,0}
\definecolor{White}{rgb}{1, 1, 1}
\definecolor{grey}{rgb}{.6, .6, .6}
\definecolor{Mygreen}{rgb}{.0, .7, .0}
\definecolor{Yellow}{rgb}{.55,.55,0}
\definecolor{Mustard}{rgb}{1.0, 0.86, 0.35}
\definecolor{applegreen}{rgb}{0.55, 0.71, 0.0}
\definecolor{darkturquoise}{rgb}{0.0, 0.81, 0.82}
\definecolor{celestialblue}{rgb}{0.29, 0.59, 0.82}
\definecolor{green_yellow}{rgb}{0.68, 1.0, 0.18}
\definecolor{crimsonglory}{rgb}{0.75, 0.0, 0.2}
\definecolor{darkmagenta}{rgb}{0.30, 0.0, 0.30}
\definecolor{magenta}{rgb}{0.50, 0.0, 0.50}
\definecolor{internationalorange}{rgb}{1.0, 0.31, 0.0}
\definecolor{darkorange}{rgb}{1.0, 0.55, 0.0}
\definecolor{ao}{rgb}{0.0, 0.5, 0.0}
\definecolor{awesome}{rgb}{1.0, 0.13, 0.32}
\definecolor{darkcyan}{rgb}{0.0, 0.50, 0.50}
\definecolor{violet}{rgb}{0.93, 0.51, 0.93}
\definecolor{brown}{rgb}{0.65, 0.16, 0.16}
\definecolor{orange}{rgb}{1.0, 0.65, 0.0}
\definecolor{darkgreen}{rgb}{0, 0.7, 0}

\definecolor{Red}{rgb}{1, 0 ,0}
\definecolor{Blue}{rgb}{0, 0 ,1}

\newcommand{\hh}{\end{document}}

\newcommand{\FPT}{\textsf{FPT}\xspace}

\newcommand{\NP}{\textsf{NP}\xspace}

\newcommand{\sdecomp}{{sphere decomposition}}
\newcommand{\sdecomps}{{\sdecomp s}}
\newcommand{\Hpl}{{\sc $\Hcal$-Planarity}\xspace}
\newcommand{\Hkpl}{{\sc $\Hcal^{(k)}$-Planarity}\xspace}
\newcommand{\blHpl}{{\sc Big-Leaf $\Hcal$-Planarity}\xspace}
\newcommand{\slHpl}{{\sc Small-Leaves $\Hcal$-Planarity}\xspace}
\newcommand{\pbdel}{{\sc $\Hcal$-Deletion}\xspace}

\newcommand{\tw}{{\sf tw}}
\newcommand{\td}{{\sf td}}

\newcommand{\adh}{{\sf adh}}

\newcommand{\size}{{\sf size}}

\newcommand{\yes}{{\sf yes}}
\newcommand{\no}{{\sf no}}

\newcommand{\ptd}{{\sf ptd}}
\newcommand{\odd}{{\sf odd}}
\newcommand{\compass}{{\sf compass}}
\newcommand{\flaps}{{\sf flaps}}
\newcommand{\torso}{{\sf torso}}
\newcommand{\ground}{{\sf ground}}
\newcommand{\influence}{{\sf influence}}
\newcommand{\cupall}{\pmb{\pmb{\bigcup}}}

\newcommand{\cc}{{\sf cc}}
\newcommand{\pmm}{{\sf pmm}}
\newcommand{\ptw}{{\sf ptw}}

\newcommand{\inG}{\text{$\mathsf{inner}$}\xspace}
\newcommand{\outG}{\text{$\mathsf{outer}$}\xspace}

\newcommand{\hptd}[1]{{#1-planar treedepth}\xspace}

\newcommand{\p}{{\sf p}\xspace}
\newcommand{\bd}{{\sf bd}\xspace}
\usepackage{color}